\documentclass[a4paper,UKenglish,cleveref, autoref, thm-restate]{lipics-v2019}

\usepackage{basic_stylesheet}

\newcommand{\lam}{\lambda}
\newcommand{\ptime}{\textsc{PTime}}
\newcommand{\uscore}{\mbox{\tt\char`\_}}

\newcommand{\m}[1]{\mathsf{#1}}

\newcommand{\mi}[1]{\mbox{\it #1}}

\newcommand{\abra}[1]{\langle #1 \rangle}
\newcommand{\abs}[1]{\lvert #1 \rvert}
\newcommand{\dbra}[1]{\llbracket #1 \rrbracket}

\newcommand{\dom}[1]{\text{dom} (#1)}

\newcommand{\degree}[1]{\m{deg} (#1)}

\newcommand{\cost}[1]{\text{cost} (#1)}

\newcommand{\closure}[3]{\text{closure} (#1; #2. #3)}

\newcommand{\inhpoly}[1]{\m{inhpoly} (#1)}

\newif\ifLongVersion

\LongVersiontrue

\title{Typable Fragments of Polynomial Automatic Amortized Resource Analysis} 

\author{Long Pham}{Carnegie Mellon University, Pittsburgh, PA, USA}{longp@andrew.cmu.edu}{https://orcid.org/0000-0001-5153-8140}{}

\author{Jan Hoffmann}{Carnegie Mellon University, Pittsburgh, PA, USA}{janh@andrew.cmu.edu}{https://orcid.org/0000-0001-8326-0788}{} 

\funding{
	This article is based on research supported by DARPA under AA Contract FA8750-18-C-0092 and by the National Science Foundation under SaTC Award 1801369, CAREER Award 1845514, and SHF Awards 1812876 and 2007784. Any opinions, findings, and conclusions contained in this document are those of the authors and do not necessarily reflect the views of the sponsoring organizations.
	Long Pham gratefully acknowledges the support of the Funai Overseas Scholarship by the Funai Foundation of Information Technology, Japan. 
}

\authorrunning{L. Pham and J. Hoffmann} 

\Copyright{Long Pham and Jan Hoffmann} 



\category{} 

\relatedversion{} 

\supplement{}

\nolinenumbers 

\EventEditors{Christel Baier and Jean Goubault-Larrecq}
\EventNoEds{2}
\EventLongTitle{29th EACSL Annual Conference on Computer Science Logic (CSL 2021)}
\EventShortTitle{CSL 2021}
\EventAcronym{CSL}
\EventYear{2021}
\EventDate{January 25--28, 2021}
\EventLocation{Ljubljana, Slovenia (Virtual Conference)}
\EventLogo{}
\SeriesVolume{183}
\ArticleNo{26}

\begin{document}

\maketitle



\begin{abstract}
	Being a fully automated technique for resource analysis, automatic amortized resource analysis (AARA) can fail in returning worst-case cost bounds of programs, fundamentally due to the undecidability of resource analysis. 
	For programmers who are unfamiliar with the technical details of AARA, it is difficult to predict whether a program can be successfully analyzed in AARA. 
	Motivated by this problem, this article identifies classes of programs that can be analyzed in type-based polynomial AARA.
	Firstly, it is shown that the set of functions that are typable in univariate polynomial AARA coincides with the complexity class \ptime{}. 
	Secondly, the article presents a sufficient condition for typability that axiomatically requires every sub-expression of a given program to be polynomial-time. 
	It is proved that this condition implies typability in multivariate polynomial AARA under some syntactic restrictions.
\end{abstract}


\section{Introduction}

There exists a wide range of effective techniques for automatically
or semi-automatically analyzing the resource consumption of programs.
These techniques derive symbolic bounds on the
worst-case~\cite{Kincaid2017pldi}, best-case~\cite{FrohnNHBG16,Ngo2017}, or
expected~\cite{ChatterjeeFM17,Ngo2018} resource consumption and are
based on type
systems~\cite{Crary2000,Vasconcelos2008,Danielsson2008,LagoG11,Avanzini2017,Cicek2017,Handley2020},
recurrence
relations~\cite{Wegbreit1975,Grobauer2001,Albert07,Kincaid2017popl,Kavvos2019},
relational reasoning~\cite{Cicek2017,Radicek2018}, and term
rewriting~\cite{AvanziniM13,BrockschmidtEFFG14,HofmannM14}.

State-of-the-art resource analyses can automatically derive complex
bounds for large programs, 
and
making analyses more practical by improving their efficiency and range
is a main driving force in this area.
However, resource analysis for Turing-complete languages is
undecidable, and even for the most sophisticated tools there will
remain programs that cannot be analyzed automatically.
Diagnosing the cause and modifying the program so that the analysis
can derive a bound often require in-depth knowledge of the
implemented techniques.
As a result, the usability of more sophisticated analysis tools is
hampered by their complexity.

To improve the usability of automatic resource analysis for
non-experts, this article develops easy-to-understand
characterizations of the programs that can be analyzed with automatic
amortized resource analysis (AARA).
Such characterizations can serve as explanations for an unsuccessful
resource analysis and guide program development without revealing
technical details of the underlying analysis.

AARA is a type-based analysis that is based on the potential method
of amortized analysis.
It has been first introduced by Hofmann and Jost~\cite{Hofmann2003}
for deriving linear heap-space bounds for a first-order
language with lists.
AARA has subsequently been extended to univariate polynomial bounds
\cite{Hoffmann2010}, multivariate polynomial bounds
\cite{Hoffmann2011,Hoffmann2012}, and exponential bounds
\cite{Kahn2020}.
Furthermore, AARA has been extended to other language features such as
higher-order and polymorphic functions~\cite{Jost2010,Hoffmann2017},
lazy evaluation~\cite{Jost2016}, and probabilistic programming
\cite{Ngo2018}.
The analysis has been implemented in the programming language Resource-Aware ML (RaML)~\cite{Hoffmann2017}.  
An overview of polynomial AARA can be found in Section~\ref{sec:resource-aware ML and AARA}.
We are not aware of previous work that studies the characterization of
typable fragments of AARA.%

Our first contribution (Section~\ref{sec:embedding polynomial-time
  Turing machines in AARA}) is a characterization of the
(mathematical) functions that can be implemented in AARA.
We demonstrate that it is possible to embed every polynomial-time Turing
machine in AARA. 
That is, for every such Turing machine, there exists
an equivalent polynomial-time program that is typable in polynomial AARA.
This result shows that polynomial AARA corresponds to the
complexity class \ptime{}
and
is in the tradition of implicit computational complexity
(ICC)~\cite{Bellantoni1992,Leviant1993,Hofmann2002}, which studies
linguistic characterizations of complexity classes.
For a user of RaML, this result means that an implementation of a
\ptime{} function can always be rewritten so that a worst-case cost bound can be
automatically derived.
However, it does not provide guidance on how to rewrite an implementation.

An ideal resource analysis should automatically derive a cost bound for
every program that has a polynomial bound.
However, such an analysis does not exist, because the problem of
deciding whether a given program runs in polynomial time is
undecidable~\cite{Hoffmann2011}.
Moreover, AARA is a type-based analysis that derives the bound of an
expression from its sub-expressions.
So we can only expect to derive a bound for an expression
which is \emph{inherently polynomial time}, that is, every
subexpression is in \ptime{} if viewed as a function.

Our second contribution is an axiomatic definition of inherently
polynomial time that implies typability in multivariate polynomial
AARA for a Turing-complete first-order language with lists (Section~\ref{sec:resource-aware ML and AARA}) under some restrictions:
Programs can only use primitive recursion instead of general
recursion, some variables are affine, and the use of nested lists is
restricted.
Although this characterization is far from being a necessary condition,
we believe that it can be a valuable guide to users.
A key concept is the notion of \emph{uniform resource annotations}
which is essential in the proof that inherently polynomial time is a
sufficient condition for typability in multivariate polynomial AARA.



\section{Automatic Amortized Resource Analysis (AARA)}
\label{sec:resource-aware ML and AARA}

Among approaches to resource analysis is AARA. 
Given a program $P$, consider its history of execution, that is, a sequence of transitioning program states. 
As in Sleator and Tarjan's potential method in amortized analysis \cite{Tarjan1985}, we assign a certain (non-negative) amount of potential to the initial state of this sequence. 
If we can ensure that (i) the amount of potential never becomes negative throughout $P$'s run and (ii) the actual computational cost in each transition of $P$ is less than or equal to the change in the amount of potential, then we know that the total resource usage of $P$ is bounded above by the initial potential. 
This is essentially how AARA works. 

More concretely, each sub-expression of $P$ is assigned a \emph{resource-annotated type}: a conventional (i.e.~simple) type augmented with an expression that indicates how much potential is stored. 
In polynomial AARA \cite{Hoffmann2010,Hoffmann2012}, we use polynomial functions to express potential. 
Initially, AARA only assigns templates of resource-annotated types where coefficients of polynomials are left blank. 
AARA then collects constraints on these coefficients that respect the cost semantics of $P$. 
Finally, as these constraints are all linear, we can simply solve them using an off-the-shelf liner program solver, thereby inferring resource-annotated types. 
A worst-case cost bound of $P$ can be extracted from its resource-annotated type. 

\subsection{Resource-Aware ML}
\label{sec:resource-aware ML}

Resource-Aware ML (RaML) is a Turing-complete functional programming language used in the study of AARA \cite{Hoffmann2010}. 

The original version of RaML is first-order (i.e.~no higher-order types or functions appear in RaML) and only offers a relatively small set of language features. 
Subsequent versions of RaML support more language features such as higher-order functions and polymorphic functions \cite{Hoffmann2017}. 
In this section, we describe a variant of RaML that only differs from the original version in a few minor details; e.g.~the $\m{tick}$ construct and the support for sum types. 

The base types (denoted by $b$) and simple types (denoted by $\tau$) of RaML are formed by
\begin{alignat*}{4}
	b ::= {} & \mathbf{1} &\qquad& \text{unit type} & \qquad \qquad \tau ::= {} & b &\qquad& \text{base type}\\
	& b_1 + b_2 && \text{sum type} & & b_1 \rightarrow b_2 && \text{arrow type} \\
	& b_1 \times b_2 && \text{product type} \\
	& L (b) && \text{list type}. 
\end{alignat*}
The set of all base types will be denoted by $\mathbb{B}$. 

Fix a set $\mathcal{V} = \{x, y, x_1, x_2, \ldots \}$ of variable symbols and a set $\mathcal{F} = \{f, \ldots \}$ of function symbols.
The grammar of RaML is
\begin{alignat*}{2}
	e ::= {} & x && \qquad \text{variable} \\
	\mid {} & \abra{\,} && \qquad \text{unit element} \\
	\mid {} & \ell \cdot x \mid r \cdot x \mid \m{case} \; x \; \{\ell \cdot y \hookrightarrow e_{\ell} \mid r \cdot y \hookrightarrow e_{r} \} && \qquad \text{sum constructors and destructor} \\
	\mid {} & \abra{x_1, x_2} \mid \m{case} \; x \; \{ \abra{x_1, x_2} \hookrightarrow e \} && \qquad \text{pair constructor and destructor} \\
	\mid {} & [\,] \mid x_1 :: x_2 \mid \m{case} \; x \; \{ [\,] \hookrightarrow e_0 \mid (x_1 :: x_2) \hookrightarrow e_1 \} && \qquad \text{list constructors and destructor} \\
	\mid {} & \m{fun} \; f \; x = e && \qquad \text{function definition} \\
	\mid {} & f \; x && \qquad \text{function application} \\
	\mid {} & \m{tick} \; q && \qquad \text{resource consumption}; q \in \mathbb{Q} \\
	\mid {} & \m{let} \; x = e_1 \; \m{in} \; e_2 && \qquad \text{let-binding} \\
	\mid {} & \m{share} \; x \; \m{as} \; x_1, x_2 \; \m{in} \; e && \qquad \text{variable sharing}. 
\end{alignat*}
In a function definition, $e$ is allowed to mention $f$. 
Therefore, we can implement not only primitive recursion but also general recursion. 
As standard, we use the let-normal form, where we only permit function application of the form $x_1 \; x_2$ as opposed to $e_1 \; e_2$. 
For convenience in resource analysis, we require each variable symbol to be used in a affine manner (i.e.~can only be used at most once). 
To use a variable symbol multiple times, we duplicate the symbol with the $\m{share}$ construct. 

In the interest space, we will not present a type system of this language here. 
It is available in Appendix~\ref{sec:type system of RaML}.

RaML programs are evaluated using the call-by-value strategy. 
Computational costs accrue only when $\m{tick} \; q$ is executed, and this cost metric is known as the tick metric. 
The general cost semantics of RaML can be found in \cite{Hoffmann2010}. 
In the case of the running time, which is a specific cost metric, of RaML, the judgment of the cost semantics has the form
\begin{equation*}
	V \vdash e \Downarrow v \mid n, 
\end{equation*}
where $V$ is an environment (i.e.~a set of pairs of variable symbols and semantic values), $v$ is a semantic value, and $n \in \mathbb{N}$ is the running time of evaluating program $e$ to $v$. 
The running time is formally defined in Appendix~\ref{sec:running time of resource-aware ML}. 

\subsection{Univariate AARA}
\label{sec:basics of univariate AARA}

In univariate AARA, each list is annotated with a polynomial indicating the amount of the potential stored in the list. 
Univariate AARA does not let us mix potential of two lists, that is, multiply polynomials of two lists' potential. 
This is why univariate AARA is called \emph{univariate}. 

Resource-annotated base types (denoted by $b$) and resource-annotated simple types (denoted by $\tau$) are formed by the following grammar:
\begin{alignat*}{4}
	b ::= {} & \mathbf{1} &\qquad& \text{unit type} & \qquad \qquad B ::= {} & \abra{b, q} &\qquad& q \in \mathbb{Q}_{\geq 0} \\
	& b_1 + b_2 && \text{sum type} & \tau ::= {} & b && \text{base type} \\
	& b_1 \times b_2 && \text{product type} && B_1 \rightarrow B_2 && \text{arrow type} \\
	& L^{\vec{q}} (b) && \text{list type}. 
\end{alignat*}
Here, $\vec{q}$ is a finite vector of $\mathbb{Q}_{\geq 0}$. 

Given a semantic value $v: b$, where $b$ is a resource-annotated base type, the potential stored in $v$ is inductively defined as
\begin{align*}
	\Phi (v: \mathbf{1}) & := 0 & \Phi ([\,] : L^{\vec{q}} (b)) & := 0\\
	\Phi (\ell \cdot v : b_1 + b_2) & := \Phi (v:b_1) & \Phi (v_1::v_2 : L^{\vec{q}} (b)) & := \Phi (v_1:b) + \phi(\abs{v_1::v_2}, \vec{q}) \\
	\Phi (r \cdot v : b_1 + b_2) & := \Phi (v:b_2), 
\end{align*}
where $\abs{\cdot}$ denotes the length of an input list.
Given $n \in \mathbb{N}$ and $\vec{q} = (q_1, \ldots, q_k)$, $\phi (n, \vec{q})$ is defined as
$\phi (n, \vec{q}) := \sum_{i = 1}^{k} q_i \binom{n}{i}$. 
If $n < i$, then $\binom{n}{i} = 0$. 

The typing judgment of univariate AARA has the form
\begin{equation*}
	\Gamma_{\text{anno}}; p \vdash e: B, 
\end{equation*}
where $\Gamma_{\text{anno}}$ is a resource-annotated typing context and $p \in \mathbb{Q}_{\geq 0}$. 
We sometimes write $\Sigma_{\text{anno}}; \Gamma_{\text{anno}}; p \vdash e: B$, where a resource-annotated typing context is split into $\Sigma_{\text{anno}}$ for arrow-type variables and $\Gamma_{\text{anno}}$ for base-type variables. 
The type system of univariate AARA is available in Appendix~\ref{sec:type system of univariate AARA}. 

To give examples of judgments in univariate AARA, consider two programs: (i) $\mi{append}$ that appends the first input list to the second, and (ii) $\mi{quicksort}$ that performs quicksort. 
The running time of $\mi{append}$ is proportional to the size of the first input, and the running time of $\mi{quicksort}$ is bounded by the square of the input size. 
For simplicity, we will not work out the exact coefficients of polynomial bounds. 
Instead, we simply assume that the running time of $\mi{append}$ is bounded by the function $n, m \mapsto n$, where $n$ and $m$ are the lengths of the two input lists. 
Likewise, we assume that the running time of  $\mi{quicksort}$ is bounded by $n \mapsto n^2$, respectively. 
It then makes sense that these two programs can be typed in univariate AARA as
\begin{equation*}
	\mi{append} : \abra{\abra{L^{1} (b), L^{0} (b)}, 0} \rightarrow \abra{L^{0} (b), 0} \qquad	\mi{quicksort} : \abra{L^{(1, 2)} (b), 0} \rightarrow \abra{L^{0} (b), 0}. 
\end{equation*}
The univariate resource annotation $(1,2)$ of $\mi{quicksort}$ represents polynomial $n \mapsto 1 \cdot \binom{n}{1} + 2 \cdot \binom{n}{2} = n^2$. 
The implementations of $\mi{append}$ and $\mi{quicksort}$ are given in Appendix~\ref{sec:implementation of append and quicksort presented as examples of typing judgments of univariate AARA}. 

Univariate AARA is sound with respect to the cost semantics (specifically, the running time) of RaML:  
\begin{theorem}[Soundness of univariate AARA \cite{Hoffmann2010}]
\label{theorem:soundness of univariate AARA}
Given term $e$, suppose $\Gamma_{\text{anno}}; p \vdash e : \abra{b_{\text{anno}}, q}$ is derived in univariate AARA.
Let $V$ be an environment such that $V \vdash e \Downarrow v \mid n$; that is, $e$ runs in $n$ units of time under $V$. 
We then have 
\begin{equation*}
	n \leq p + \Phi (V: \Gamma_{\text{anno}}) - q - \Phi (v: b_{\text{anno}}), 
\end{equation*} 
where $\Phi (V: \Gamma_{\text{anno}}) = \sum_{x \in \dom{\Gamma_{\text{anno}}}} \Phi (V (x) : \Gamma_{\text{anno}} (x))$. 
\end{theorem}

\subsection{Multivariate AARA}

In contrast to univariate AARA, multivariate AARA allows us to mix potential of different lists. 
For example, we can have $\abs{\ell_1} \cdot \abs{\ell_2}$'s worth of potential, where $\abs{\cdot}$ denotes the length of a list, in multivariate AARA. 
Due to this multivariate nature, multivariate AARA has a single global resource annotation represented by a multivariate polynomial over all size variables occurring in a given term. 
This global resource annotation is separate from individual types in a typing context. 

Multivariate AARA is strictly more expressive than univariate one. 
This is surprising in light of the fact that multivariate polynomials can always be bounded by univariate polynomials; e.g.~$x y$ is bounded by $x^2 + y^2$. 
Examples of programs that cannot be typed in univariate AARA but are typable in multivariate AARA are in Section~\ref{sec:high-level desgin of inherently polynomial time} and Section~\ref{sec:typable fragment of resource-aware ML}. 

\paragraph*{Resource-Annotated Types}

Resource-annotated types in multivariate AARA are formed by
\begin{alignat*}{4}
	b ::= {} & \mathbf{1} &\qquad& \text{unit type} & \qquad \qquad B ::= {} & \abra{b, Q} \\
	& b_1 + b_2 && \text{sum type} & \tau ::= {} & b &\qquad& \text{base type} \\
	& b_1 \times b_2 && \text{product type} && B_1 \rightarrow B_2 && \text{arrow type} \\
	& L (b) && \text{list type}. 
\end{alignat*}
In $\abra{b, Q}$, $Q$ is a multivariate resource annotation over the size variables inside $b$.
This will be formalized shortly.  

Given a base type $b \in \mathbb{B}$, its base polynomial is a function of type $\dbra{b} \rightarrow \mathbb{N}$, where $\dbra{b}$ is the set of semantic values of type $b$. 
The set of base polynomials associated with type $b$, denoted by $\mathcal{B} (b)$, is inductively defined as follow: 
\begin{align*}
	\mathcal{B} (\mathbf{1}) & := \{\lam v. 1 \} \\
	\mathcal{B} (b_1 + b_2) & := \{\lam (\ell \cdot v). p (v) \mid p \in \mathcal{B} (b_1) \} \cup \{\lam (r \cdot v). p (v) \mid p \in \mathcal{B} (b_2) \}\\
	\mathcal{B} (b_1 \times b_2) & := \{\lam \abra{v_1, v_2}. p_1 (v_1) \cdot p_2 (v_2) \mid p_{i} \in \mathcal{B} (b_{i}) \} \\
	\mathcal{B} (L (b)) & := \{\lam [v_1, \ldots, v_n]. \sum_{1 \leq j_1 < \cdots < j_{k} \leq n} \; \prod_{1 \leq i \leq k} p_{i} (v_{j_i}) \mid k \in \mathbb{N}, p_{i} \in \mathcal{B} (b) \}. 
\end{align*}
For $b_1 + b_2$, we have a set of base polynomials for the $\ell$-tag and another set for the $r$-tag. 
If a base polynomial is applied to a value with a wrong tag, we assume that the output is 0. 
For instance, if we feed a value $\ell \cdot \abra{\,}$ to $\lam (r \cdot v). 1$, the output should be 0. 
In the definition of $\mathcal{B} (L (b))$, if $n < k$, the function should return 0 since it is the identity of summation. 

Given base type $b$, a resource polynomial $p: \dbra{b} \rightarrow \mathbb{Q}_{\geq 0}$ is a non-negative linear combination of finitely many base polynomials from $\mathcal{B} (b)$.  
It is straightforward to prove that $\mathcal{B} (b)$ for any $b$ contains $\lam v. 1$. 
Therefore, a resource polynomial is always capable of expressing constant potential. 

For convenience, it is desirable to have a succinct notation for base polynomials. 
This is achieved by introducing indexes of base polynomials:  
\begin{align*}
	\mathcal{I} (\mathbf{1}) & := \{*\} \\
	\mathcal{I} (b_1 + b_2) & := \{\ell \cdot i \mid i \in \mathcal{I} (b_1)\} \cup \{r \cdot i \mid i \in \mathcal{I} (b_2)\} \\
	\mathcal{I} (b_1 \times b_2) & := \{\abra{i_1, i_2} \mid i_1 \in \mathcal{I} (b_1), i_2 \in \mathcal{I} (b_2) \} \\
	\mathcal{I} (L (b)) & := \{[i_1, \ldots, i_k] \mid k \in \mathbb{N}, i_{j} \in \mathcal{I} (b)\}. 
\end{align*}

An index is usually used as a subscript for a (meta)-variable representing a coefficient of a base polynomial. 
For instance, $q_{\abra{*, *}} \in \mathbb{Q}_{\geq 0}$ is a meta-variable representing a coefficient of base polynomial $\lam \abra{v_1, v_2}. 1$. 
For any base type $b$, we will write $0_{b}$ for the index $\lam v. 1$. 

For example, consider $\mathcal{I} (L (\mathbf{1})) = \{*, [*], [*, *], [*, *, *], \ldots\}$.
The index $[*, *]$ represents the polynomial function
\begin{align*}
	\lam [v_1, \ldots, v_n]. \sum_{1 \leq j_1 < j_{2} \leq n} \; \prod_{1 \leq i \leq 2} ((\lam v. 1) \; v_{j_i}) & = \lam [v_1, \ldots, v_n]. \sum_{1 \leq j_1 < j_{2} \leq n} 1 \\
	& = \lam [v_1, \ldots, v_n]. \binom{n}{2}. 
\end{align*}
Thus, the multivariate index $[*, *]$ represents a quadratic function on the input list's length. 

The degree of an index is defined by
\begin{align*}
	\degree{*} & := 0 & \degree{\abra{i_1, i_2}} & := \degree{i_1} + \degree{i_2} \\
	\degree{\ell \cdot i}, \degree{r \cdot i} & := \degree{i} & \degree{[i_1, \ldots, i_k]} & := k + \sum_{1 \leq j \leq k} \degree{i_j}. 
\end{align*}
Intuitively, $\degree{i}$ is equal to the degree of the polynomial function that index $i$ represents. 
Because a resource polynomial can only have non-zero coefficients for finitely many base polynomials, any resource polynomial (or a finite set of resource polynomials) has a bounded degree. 
In practice, we ask a user of AARA to supply an upper bound on the degree of base polynomials. 

\paragraph*{Resource Annotations of Typing Contexts}

Given a base-type typing context $\Gamma = \{x_1: b_1, \ldots, x_n: b_n\}$, its multivariate resource annotation is given by a resource polynomial of type $b_1 \times \cdots \times b_n$.
In other words, we treat a typing context as one big tuple and assign a single multivariate annotation to this tuple. 

With regard to an arrow-type typing context $\Sigma = \{f_1: b_{1,1} \rightarrow b_{1,2}, \ldots, f_{m}: b_{m, 1} \rightarrow b_{m, 2} \}$, its multivariate resource annotation has the form 
\begin{equation*}
	\Sigma_{\text{anno}} = \{f_1: B_{1,1} \rightarrow B_{1,2}, \ldots, f_{m}: B_{m, 1} \rightarrow B_{m, 2} \}, 
\end{equation*}
where each $B_{i, j}$ is a pair $\abra{b_{i, j}, Q}$ such that $Q$ is a multivariate resource annotation of $b_{i, j}$. 

\paragraph*{Typing Judgment}

The typing judgment of multivariate AARA takes the form
\begin{equation*}
	\Gamma; P \vdash e : \abra{b, Q}, 
\end{equation*}
where $\Gamma$ and $b$ are free of resource annotations.
$P$ and $Q$ are multivariate annotation over $\Gamma$ and $b$, respectively. 
The type system of multivariate AARA is available in Appendix~\ref{sec:type system of multivariate AARA}. 

To give examples of judgments in multivariate AARA, consider $\mi{append} \; \abra{\ell_1, \ell_2}$, which appends $\ell_1$ to $\ell_2$. 
Suppose that the output must store $n \mapsto n^2$ much potential, where $n$ is the output's length.
It is reasonable that the total potential required for this program is $\abs{\ell_1} + (\abs{\ell_1} + \abs{\ell_2})^2$, out of which $\abs{\ell_1}$ is used to account for the running time. 
This can be expressed by the judgment $\ell_1: L (\mathbf{1}), \ell_2: L (\mathbf{1}); P \vdash \mi{append} \; \abra{\ell_1, \ell_2} : \abra{L (\mathbf{1}), Q}$, where the positive coefficients of $P$ and $Q$ are
\begin{align*}
	& P (\abra{[*], *}) = P (\abra{[*, *], *}) = P  (\abra{[*], [*]}) = P (\abra{*, [*, *]}) = 2 &&
	P (\abra{*, [*]}) = 1 \\
	& Q ([*, *]) = 2 && Q ([*]) = 1.  
\end{align*}
$P$ amounts to $2 \cdot \left( \binom{\abs{\ell_1}}{1} + \binom{\abs{\ell_1}}{2} + \binom{\abs{\ell_1}}{1} \cdot \binom{\abs{\ell_2}}{1} + \binom{\abs{\ell_2}}{2} \right) + 1 \cdot \binom{\abs{\ell_2}}{1}$, which is equal to $\abs{\ell_1} + (\abs{\ell_1} + \abs{\ell_2})^2$ as desired. 
Similarly, $Q$ amounts to $2 \cdot \binom{n}{2} + 1 \cdot \binom{n}{1} = n^2$ as desired, where $n = \abs{\ell_1} + \abs{\ell_2}$. 

The multivariate equivalent of the soundness theorem (Theorem~\ref{theorem:soundness of univariate AARA}) holds \cite{Hoffmann2012}. 


\section{Embedding Polynomial-Time Turing Machines in AARA}
\label{sec:embedding polynomial-time Turing machines in AARA}

In this section, we show that every polynomial-time Turing machine can be expressed as a typable RaML program while preserving the semantics and worst-case cost bounds. 
More formally, we have

\begin{restatable}[Embedding of polynomial-time Turing machines in RaML]{theorem}{embeddingofTuringmachines}
\label{theorem:polynomial-time Turing machines can be embeded in RaML}
Let $M$ be a polynomial-time Turing machine that inputs and outputs bit strings from $\{0, 1\}^{*}$. 
There exists a RaML program $M': \{0, 1\}^{*} \to \{0, 1\}^{*}$ such that
\begin{itemize}
	\item For every input $w \in \{0, 1\}^{*}$, we have $M (w) = M' (w)$;
	\item The computational cost of $M'$ (according to the tick metric) is larger than or equal to the running time of $M$;
	\item Univariate AARA can infer a polynomial upper bound of the computational cost of $M'$. 
\end{itemize}
\end{restatable}

Theorem~\ref{theorem:polynomial-time Turing machines can be embeded in RaML} only tells us the existence of a RaML program $M'$ that is typable in univariate AARA and that simulates $M$ faithfully. 
In our proof of the theorem, we assume that a polynomial bound on the running time of $M$ is known. 
Thus, if we do not have access to this polynomial bound, we cannot construct $M'$. 
In fact, the problem of determining whether a given Turing machine runs in polynomial time or not is undecidable \cite{Hoffmann2011}. 

It is fairly easy to prove that the cost of any program according to the tick metric is asymptotically bounded by its running time. 
Therefore, in the statement of Theorem~\ref{theorem:polynomial-time Turing machines can be embeded in RaML}, we can replace the ``tick metric'' with the ``running time'' of RaML. 

A detailed proof of Theorem~\ref{theorem:polynomial-time Turing machines can be embeded in RaML} is available in Appendix~\ref{sec:upplementary results for the embedding of polynomial-time Turing machines in AARA}. 

\subsection{Preliminaries}
\label{sec:definitions of Turing machines and Resource-Aware ML}

\begin{restatable}[Turing machine]{definition}{Turingmachine}
A (deterministic) Turing machine $M$ is specified by an 8-tuple $(Q, \Sigma, \Gamma, \vdash, \sqcup, \delta, q_0, q_{\text{final}})$, where
\begin{itemize}
	\item $Q$ is a finite set of machine states.
	\item $\Sigma$ is a finite input alphabet.
	$\Gamma$ is a finite alphabet for symbols written on $M$'s tape. 
	Since an input will be initially placed on the tape, every input symbol is also a tape symbol. 
	\item ${\vdash} \in \Gamma \setminus \Sigma$ is the left end marker that demarcates the left end of a semi-infinite working tape, and $\sqcup \in \Gamma \setminus \Sigma$ is the blank symbol for the tape. 
	\item $\delta: Q \times \Gamma \rightarrow Q \times \Gamma \times \{L, R\}$ is the transition function.
	\item $q_0 \in Q$ is the initial state, and $q_{\text{final}} \in Q$ is the final state.
\end{itemize}
\end{restatable}

In the initial configuration of a Turing machine, an input string $w$ is placed immediately after the left end marker $\vdash$ on the tape.
The state of the machine is initially $q_0$, and the read/write head is positioned on the first symbol of $w$. 
The rest of the tape is filled with $\sqcup$. 

The Turing machine first (i) reads the content of the cell currently under the tape head and (ii) identifies the current state of the machine.
The machine then overwrites the current cell (if necessary), updates the machine's state, and moves the head to the left or right according to the transition function $\delta$. 
The machine terminates as soon as it enters $q_{\text{final}}$.
Upon termination, the content of the tape before the first blank symbol is considered as the machine's output. 
The running time is defined as the number of steps the Turing machine makes before termination. 

Without loss of generality, we will henceforth only consider Turing machines with $\Sigma = \{0, 1\}$ and $\Gamma = \Sigma \cup \{\vdash, \sqcup\}$. 

To enhance clarity, we will introduce two type aliases, $\text{State}$ and $\text{Sym}$, which are defined as $L (\mathbf{1} + \mathbf{1})$; i.e.~bit strings or natural numbers. 
The type $\text{State}$ represents machine states of $M$, and $\text{Sym}$ represents tape symbols of $M$. 
In fact, because $M$ has finitely many machine states and tape symbols, $\text{State}$ and $\text{Sym}$ can alternatively be encoded as $\mathbf{1} + \cdots + \mathbf{1}$. 

\subsection{Embedding}

Fix a polynomial-time Turing machine $M = (Q, \Sigma, \Gamma, \vdash, \sqcup, \delta, q_0, q_{\text{final}})$.
Assume that the running time of $M$ is bounded above by $p (n)$ for some polynomial $p: \mathbb{N} \rightarrow \mathbb{N}$. 
The target program of the translation will be denoted by $M'$, and this is what we are about to define. 
$M'$ works as described in Algorithm~\ref{alg:target RaML program}. 
A RaML implementation of $M'$ is available in Appendix~\ref{sec:target RaML programs in the embedding of polynomial-time Turing machines}. 

\begin{algorithm}
	\begin{algorithmic}[1]
		\Require $w  \in \{0, 1\}^{*}$
		\Procedure{$M'$}{$w$}
		\State Create a singleton list $\ell_1 : L(\text{Sym})$ containing $\vdash$
		\State Create a list $\ell_2 : L(\text{Sym})$ of size $p(\abs{w})$ filled with $\sqcup$ \label{algline:l_2 is generated}
		\State Prepend $\ell_2$ with $w$ \label{algline:prepend l_2 with w}
		\State Create a list $\mi{ps} : L(\mathbf{1})$ of size $p (\abs{w})$ \Comment{Reservoir of potential}
		\State $s \gets q_0$ \Comment{Initialize the current state}
		\While{$s \neq q_{\text{final}} \land \mi{ps} \neq [\,]$}
		\State $\mi{ps} \gets \text{tail} \; \mi{ps}$ \Comment{Potential is released} \label{algline:first line of the loop's body}
		\State Compute $\delta(s, \ell_2[0])$
		\State Update $s$ and $\ell_2[0]$ appropriately
		\State Update the tape head's position by moving the head of $\ell_1$ or $\ell_2$ to the other \label{algline:last line of the loop's body}
		\EndWhile
		\State \textbf{return} $\text{append} (\text{reverse} \; \ell_1, \ell_2)$
		\EndProcedure
	\end{algorithmic}
	\caption{Operational working of target RaML program $M'$}
	\label{alg:target RaML program}
\end{algorithm}
The list $\ell_1$ represents the region on $M$'s tape to the left of the tape head (in the reverse order and excluding the cell where the tape head is currently on), and $\ell_2$ represents the region to the right of the head (including the current cell). 
Since it is assumed that $p (\abs{w})$, where $\abs{w}$ denotes the length of input list $w$, is an upper bound on $M$'s running time, we are assured that $M$ requires at most $p (\abs{w})$ many cells on the working tape.
This is why $\ell_2$ initially has size $p (\abs{w})$.
In fact, because we prepend $\ell_2$ with $w$ in line~\ref{algline:prepend l_2 with w}, we have $\abs{w}$ more cells than necessary. 

The list $\mi{ps}$ acts as a reservoir of potential, storing constant potential in each element. 
As the head of $\mi{ps}$ is removed in line~\ref{algline:first line of the loop's body}, the potential stored in this element is freed and will be consumed in subsequent lines inside the loop's body. 

It is technically possible to store potential directly in $\ell_1$ and $\ell_2$, which together simulate $M$'s working tape. 
However, not all cells on the working tape of $M$ are accessed equally often---some cells are accessed more often than others, and the maximum number of accesses to a given cell may not be bounded by a constant.
If we are to store potential in $\ell_1$ and $\ell_2$, each cell of $\ell_1$ and $\ell_2$ needs to store $p (n)$ units of potential at the beginning. 
As a result, the total amount of potential supplied to $M'$ is $p^{2} (n)$, which is a gross over-approximation of the actual running time. 
Therefore, to have a tighter cost bound, a separate list, namely $\mi{ps}$, is employed as a reservoir of potential. 


\section{Inherently Polynomial Time}
\label{sec:inherently polynomial time}

Section~\ref{sec:embedding polynomial-time Turing machines in AARA} investigates the expressive power of AARA from the viewpoint of programming language semantics, disregarding the issue of how to algorithmically turn an arbitrary Turing machine into a typable RaML program. 
By contrast, in this section, we aim to identify a typable fragment of AARA that is defined statically/axiomatically.
Henceforth, we will call the sufficient condition corresponding to the typable fragment that this section presents \emph{inherently polynomial time}. 

A key requirement is that the typable fragment should not resemble AARA's type system, which itself is also defined axiomatically.
Otherwise, it would be trivial to prove that any term in this fragment is typable in AARA. 
Because we want users of AARA to benefit from our findings of the present work, another requirement is that the definition (or at least the informal definition) of inherently polynomial time should be easy to convey to users of AARA. 
On the other hand, it is not our priority to find as large a typable fragment as we can. 

In the remaining of the article, we will focus on the running time as a cost metric of RaML, unless stated otherwise. 

\subsection{High-Level Design}
\label{sec:high-level desgin of inherently polynomial time}

By Theorem~\ref{theorem:soundness of univariate AARA} (and its multivariate equivalent), AARA is sound: if a program is typable in AARA, its resource-annotated type is a correct upper bound on the running time. 
Hence, to be typable in AARA, the worst-case running time of a program must be polynomial. 
To ensure termination of programs, we first restrict recursion to primitive recursion. 

Furthermore, the type system of AARA is compositional: if term $e$ is typable, so is every sub-expression of $e$. 
Hence, in order for $e$ to be typable, not only $e$ but also all of its sub-expressions must be polynomial-time. 
This suggests that we should define the sufficient condition inductively, hence the name \emph{inherently} polynomial time. 

It is straightforward to determine whether each of the base cases of the inductive definition is typable or not. 
It remains to work out inductive cases in the inductively defined sufficient condition for typability. 
The most interesting case is primitive recursion. 
A primitive recursion will be written as
\begin{equation*}
	e := \m{rec} \; x \; \{[\,] \hookrightarrow e_0 \mid (y::\mi{ys}) \; \m{with} \; z \hookrightarrow e_1 \}, 
\end{equation*}
where $x$ is matched against $y :: \mi{ys}$ in the second branch, and $z$ is the result of a recursive call. 
The stepping function $e_1$ can only contain $y$, $\mi{ys}$, and $z$ as free variables; i.e.~$\mathtt{FV} (e_1) \subseteq \{y, \mi{ys}, z\}$. 
From compositionality, we know that $e_0$ and $e_1$ are both typable and hence run in polynomial time. 
Under what condition does the entire $e$ run in polynomial time as well?

To answer this question, we first observe the following.
Without any restrictions on $e_0$ and $e_1$ apart from that they should be typable, AARA may project $e$'s worst-case time complexity to be exponential even if the actual running time of $e$ is polynomial. 
To illustrate this, consider 
\begin{equation} \label{eq:example of exponential blowup}
	e := \m{rec} \; x \; \{[\,] \hookrightarrow [\,] \mid (y::\mi{ys}) \; \m{with} \; z \hookrightarrow \m{share} \; z \; \m{as} \; z_1, z_2 \; \m{in} \; \mi{append} \; \abra{z_1, z_2} \}. 
\end{equation}
Although the actual running time of $e$ is $O (\abs{x})$ and hence is linear, $e$ is untypable in polynomial AARA. 
The problem of \eqref{eq:example of exponential blowup} is that the stepping function doubles the input size.
This makes AARA conclude (na\"ively) that the worst-case total running time is $O (2^{\abs{x}})$, and this cost bound is beyond the expressive power of AARA (exponential AARA \cite{Kahn2020}, however, can handle exponential cost bounds). 

To preclude the example \eqref{eq:example of exponential blowup}, it is reasonable to require the running time of $e_1$ (i.e.~a stepping function inside primitive recursion) to be constant in the size of $z$ (i.e.~the result of a recursive call). 
More concretely, if $T (\abs{y}, \abs{\mi{ys}}, \abs{z})$ is the running time of a stepping function, we demand $T (\abs{y}, \abs{\mi{ys}}, \abs{z}) \leq p (\abs{y}, \abs{\mi{ys}})$, where $p (\abs{y}, \abs{\mi{ys}})$ is a polynomial in $\abs{y}$ and $\abs{\mi{ys}}$ (i.e.~the sizes of $y$'s and $\mi{ys}$'s semantic values\footnotemark). 
We will adopt this idea in the formulation of inherently polynomial time. 

\footnotetext{A formal definition of the size of RaML's base-type semantic values is not provided in this article.
However, the idea is intuitive.
For example, the size of a list is given by the sum of all elements' sizes.}

Although this idea results in a fairly simple inductive definition of inherently polynomial time, a major drawback is that some realistic programs are not admitted by the current formulation of inherently polynomial time.
For instance, consider $\mi{multiply}$ that, given input lists $\ell_1$ and $\ell_2$, produces a list of size $\abs{\ell_1} \cdot \abs{\ell_2}$: 
\begin{equation} \label{eq:program that returns the product of two input lists}
	\mi{multiply} := \lam \ell_1. \lam \ell_2. \m{rec} \; \ell_1 \; \{[\,] \hookrightarrow \abra{\ell_2, [\,]} \mid (y::\mi{ys}) \; \m{with} \; z \hookrightarrow e_1 \},
\end{equation}
where the stepping function of primitive recursion is 
\begin{equation*}
	e_1 \equiv \m{case} \; z \; \{\abra{z_1, z_2} \hookrightarrow \m{share} \; z_1 \; \m{as} \; z_{1,1}, z_{1,2} \; \m{in} \; \abra{z_{1, 1}, \mi{append} \; \abra{z_{1, 2}, z_2}}\}. 
\end{equation*}
The first component of $z$ stores $\ell_2$, while the second component of $z$ acts as an accumulator. 
The running time of $e_1$ is polynomial in $\abs{z_1}$ but constant in $\abs{z_2}$. 
Therefore, $e_1$'s running time is only polynomial \emph{partially} in $\abs{z}$.
This is why the overall time complexity of $e$ remains polynomial instead of becoming exponential. 
Nonetheless, \eqref{eq:program that returns the product of two input lists} is not inherently polynomial time according to the current formulation, since the formulation does not allow $e_1$'s running time to have any dependence on $\abs{z}$. 

Furthermore, \eqref{eq:program that returns the product of two input lists} can only be typed in multivariate AARA and not in univariate AARA. 
This means our formulation of inherently polynomial time fails to capture some of the realistic programs that are typable only in multivariate AARA. 
In view of this, one might wonder whether inherently polynomial time is completely encapsulated by univariate AARA; that is, every inherently polynomial-time RaML program is typable in univariate AARA. 
The answer is negative. 

As a counterexample, consider the standard $\mi{append}$ defined as
\begin{equation} \label{eq:example of append that demonstrates a limitation in the expressive power of univariate polynomial AARA}
\mi{append} := \lam \ell_1. \lam \ell_2. \m{rec} \; \ell_1 \; \{[\,] \hookrightarrow \ell_2 \mid (y::\mi{ys}) \; \m{with} \; z \hookrightarrow y::z \}. 
\end{equation}
Note that it is inherently polynomial time. 
$\mi{append}$ alone is typable in univariate AARA as well as multivariate AARA. 
However, if we require the output of $\mi{append}$ to carry quadratic potential (because it will be later fed to a function that demands quadratic potential from inputs, for example), then univariate AARA cannot type $\mi{append}$---we need to resort to multivariate AARA to type it. 

In summary, our formulation of inherently polynomial time goes beyond the remit of univariate AARA, but does not capture the full range of realistic programs that require multivariate potential. 

\subsection{Formulation of Inherently Polynomial Time}
\label{sec:formulation of inherently polynomial time}

\paragraph*{Restricting the Syntax of Resource-Aware ML}

To ensure termination of programs, we require programs to use primitive recursion in place of general recursion. 
Hence, we will from now on work with a fragment of RaML wherein general recursion is replaced by primitive recursion. 
This fragment removes $\m{fun} \; f \; x = e$ from the original RaML (Section~\ref{sec:resource-aware ML}) and adds the following: 
\begin{enumerate}
	\item $\lam (x: b). e$ for a lambda abstraction, where $b \in \mathbb{B}$; 
	\item $\m{rec} \; x \; \{[\,] \hookrightarrow e_0 \mid (y::\mi{ys}) \; \m{with} \; z \hookrightarrow e_1 \}$, where $z$ denotes the result of the recursive call. 
\end{enumerate}
In primitive recursion, $e_1$ is only allowed to mention $\{y, \mi{ys}, z\}$. 
If $e_1$ needs access to a global variable $v$ (i.e.~a variable from outside the primitive recursion), $v$ should be transferred to $e_1$ by placing $v$ inside $z$. 

The reason why we deny $e_1$ access to a global variable is that every variable symbol can only be accessed at most once in RaML. 
However, this is in fact already violated by $e_1$ having access to $\mi{ys}$ (because this means some elements of the input $x$ are accessed multiple times during primitive recursion). 
Further, even if we let $e_1$ access global variables, AARA can be easily adapted.
Also, it will result in a less strict formulation of inherently polynomial time that admits $\mi{multiply}$ in \eqref{eq:program that returns the product of two input lists}. 
Nonetheless, for simplicity, this article assumes that $e_1$ can only mention $y$, $\mi{ys}$, and $z$. 

Primitive recursion can be encoded using general recursion as 
\begin{equation*}
	\m{fun} \; f \; \abra{x, \Gamma} = \m{case} \; x \; \{[\,] \hookrightarrow e_0 \mid y::\mi{ys} \hookrightarrow {} \m{share} \; \mi{ys} \; \m{as} \; \mi{ys}_1, \mi{ys}_2 \; \m{in} \; \m{let} \; z = f \; \abra{\mi{ys}_1, \Gamma} \; \m{in} \; e_1 \}.
\end{equation*}
Here, $\Gamma$ is a set/sequence of those variables that do not appear in $e_1$, but $e_0$. 
Variable $\mi{ys}_1$ is passed to the recursive call, and $\mi{ys}_2$ is used in $e_1$ (if $e_1$ mentions $\mi{ys}$). 

\paragraph*{Judgments}

The primary judgment of inherently polynomial time is
\begin{equation} \label{eq:judgment for base-type terms in partial polynomial dependence}
	\Delta; \Gamma \vdash e \; \inhpoly{V}, 
\end{equation}
where
\begin{itemize}
	\item $\Gamma$ is a typing context containing both base-type and arrow-type variables such that $\Gamma \vdash e: b$ for base type $b$.
	\item $V \subseteq \dom{\Gamma}$ is a set of variables. 
	\item $\Delta$ is a set of $f \; \m{time}$, where $f \in \dom{\Gamma}$ is an arrow-type variable and $\m{time} \in \{\m{const}, \m{poly}\}$.  
\end{itemize}
Sometimes we split $\Gamma$ into $\Sigma$ for arrow-type variables and $\Gamma$ for base-type variables, writing the judgment as $\Delta; \Sigma; \Gamma \vdash e \; \inhpoly{V}$. 
\eqref{eq:judgment for base-type terms in partial polynomial dependence} is only applicable to base-type expressions $e$. 

An informal interpretation of \eqref{eq:judgment for base-type terms in partial polynomial dependence} is
\begin{itemize}
	\item $f \; \m{const}$ denotes that the running time of $f$ is constant with respect to the input size, and likewise, $f \; \m{poly}$ denotes that $f$'s running time is polynomial\footnote{$f$'s running time being polynomial does NOT mean that it is \emph{strictly} polynomial---it can also be constant in the input size.} in the input size. 
	\item The running time of $e$ is (i) polynomial\footnote{Again, the running time of $e$ may be constant as well as polynomial in the size of any $v \in V$.} in the sizes of those variables in $V$ but (ii) constant in the sizes of $\dom{\Gamma} \setminus V$. 
	\item Every sub-expression of $e$ runs in polynomial time. 
\end{itemize}

The judgments for an arrow-type expression $e$ are
\begin{equation} \label{eq:judgment for arrow-type terms in partial polynomial dependence}
	\Delta; \Gamma \vdash e \; \m{const} \qquad \Delta; \Gamma \vdash e \; \m{poly},  
\end{equation}
$\Delta; \Gamma \vdash e \; \m{const}$ means $e$ runs in constant time with respect to the input size, and $\Delta; \Gamma \vdash e \; \m{poly}$ likewise means $e$'s running time is polynomial in the input size. 

\paragraph*{Inference Rules}

The most important inference rules defining \eqref{eq:judgment for base-type terms in partial polynomial dependence} are displayed in Figure~\ref{fig:syntax-directed rules of inherently polynomial time}. 
Throughout these rules, $b$ denotes a base type, $\m{time}$ is drawn from $\{\m{const}, \m{poly}\}$, and $V$ is a set of variables. 
The remaining rules are deferred to Figure~\ref{fig:remaining rules for inherently polynomial time} in Appendix~\ref{sec:proof of the typability theorem}. 

\begin{figure}[hbt!]
	\begin{small}
	\begin{center}
		\AxiomC{}
		\RightLabel{\textsc{(IP:Base)}}
		\UnaryInfC{$\cdot; x:b \vdash x \; \inhpoly{\emptyset}$}
		\DisplayProof
		\qquad
		\AxiomC{$\Delta = \{f \; \m{time}\}$}
		\RightLabel{\textsc{(IP:Arrow)}}
		\UnaryInfC{$\Delta; f: b_1 \rightarrow b_2 \vdash f \; \m{time}$}
		\DisplayProof
	\end{center}
	\begin{center}
		\AxiomC{$\cdot; x:b \vdash x \; \inhpoly{\emptyset}$}
		\RightLabel{\textsc{(IP:SumL)}}
		\UnaryInfC{$\cdot; x:b \vdash \ell \cdot x \; \inhpoly{\emptyset}$}
		\DisplayProof
		\qquad
		\AxiomC{$\cdot; x:b \vdash x \; \inhpoly{\emptyset}$}
		\RightLabel{\textsc{(IP:SumR)}}
		\UnaryInfC{$\cdot; x:b \vdash r \cdot x \; \inhpoly{\emptyset}$}
		\DisplayProof
	\end{center}
	\begin{prooftree}
		\AxiomC{$\cdot; x_1: b_1 \vdash x_1 \; \inhpoly{\emptyset}$}
		\AxiomC{$\cdot; x_2: b_2 \vdash x_2 \; \inhpoly{\emptyset}$}
		\RightLabel{\textsc{(IP:Pair)}}
		\BinaryInfC{$\cdot; x_1: b_1, x_2: b_2 \vdash \abra{x_1, x_2} \; \inhpoly{\emptyset}$}
	\end{prooftree}

	\begin{center}
		\AxiomC{}
		\RightLabel{\textsc{(IP:Unit)}}
		\UnaryInfC{$\cdot; \cdot \vdash \abra{\,} \; \inhpoly{\emptyset}$}
		\DisplayProof
		\quad
		\AxiomC{$\Delta = \{x_1 \; \m{const}\}$}
		\RightLabel{\textsc{(IP:App-Const)}}
		\UnaryInfC{$\Delta; x_1: b_1 \rightarrow b_2, x_2: b_1 \vdash x_1 \; x_2 \; \inhpoly{\emptyset}$}
		\DisplayProof
	\end{center}

	\begin{center}
		\AxiomC{}
		\RightLabel{\textsc{(IP:Nil)}}
		\UnaryInfC{$\cdot; \cdot \vdash [\,] \; \inhpoly{\emptyset}$}
		\DisplayProof
		\quad
		\AxiomC{$\Delta = \{x_1 \; \m{poly}\}$}
		\RightLabel{\textsc{(IP:App-Poly)}}
		\UnaryInfC{$\Delta; x_1: b_1 \rightarrow b_2, x_2: b_1 \vdash x_1 \; x_2 \; \inhpoly{\{x_2\}}$}
		\DisplayProof
	\end{center}

	\begin{prooftree}
		\AxiomC{$\cdot; x_1: b \vdash x_1 \; \inhpoly{\emptyset}$}
		\AxiomC{$\cdot; x_2: L (b) \vdash x_2 \; \inhpoly{\emptyset}$}
		\RightLabel{\textsc{(IP:Cons)}}
		\BinaryInfC{$\cdot; x_1: b, x_2: L (b) \vdash x_1::x_2 \; \inhpoly{\emptyset}$}
	\end{prooftree}

	\begin{prooftree}
		\AxiomC{$\Delta; \Gamma, y: b_1 \vdash e_{\ell} \; \inhpoly{V [x \mapsto y]}$}
		\AxiomC{$\Delta; \Gamma, y: b_2 \vdash e_{r} \; \inhpoly{V [x \mapsto y]}$}
		\RightLabel{\textsc{(IP:Case-Sum)}}
		\BinaryInfC{$\Delta; \Gamma, x: b_1 + b_2 \vdash \m{case} \; x \; \{\ell \cdot y \hookrightarrow e_{\ell} \mid r \cdot y \hookrightarrow e_{r}\} \; \inhpoly{V}$}
	\end{prooftree}
	\begin{prooftree}
		\AxiomC{$\Delta; \Gamma, x_1: b_1, x_2: b_2 \vdash e \; \inhpoly{V [x \mapsto x_1, x_2]} $}
		\RightLabel{\textsc{(IP:Case-Prod)}}
		\UnaryInfC{$\Delta; \Gamma, x: b_1 \times b_2 \vdash \m{case} \; x \; \{\abra{x_1, x_2} \hookrightarrow e\} \; \inhpoly{V}$}
	\end{prooftree}	
	\begin{prooftree}
		\AxiomC{$\Delta; \Gamma \vdash e_0 \; \inhpoly{V \setminus \{x\}}$}
		\AxiomC{$\Delta; \Gamma, x_1: b, x_2: L (b) \vdash e_1 \; \inhpoly{V [x \mapsto x_1, x_2]}$}
		\RightLabel{\textsc{(IP:Case-List)}}
		\BinaryInfC{$\Delta; \Gamma, x: L (b) \vdash \m{case} \; x \; \{[\,] \hookrightarrow e_0 \mid (x_1::x_2) \hookrightarrow e_1 \} \; \inhpoly{V}$}
	\end{prooftree}
	\begin{prooftree}
		\AxiomC{$\Delta; \Gamma \vdash e_0 \; \inhpoly{V}$}
		\AxiomC{$\cdot; y:b, \mi{ys}: L (b), z: b_2 \vdash e_1 \; \inhpoly{\{y, \mi{ys}\}}$}
		\RightLabel{\textsc{(IP:Rec)}}
		\BinaryInfC{$\Delta; \Gamma, x: L (b) \vdash \m{rec} \; x \; \{[\,] \hookrightarrow e_0 \mid (y::\mi{ys}) \; \m{with} \; z \hookrightarrow e_1\} \; \inhpoly{V \cup \{x\}}$}
	\end{prooftree}
	\begin{prooftree}
		\AxiomC{$\Delta_1; \Sigma_1; \Gamma_1 \vdash e_1 \; \inhpoly{V_1}$}
		\AxiomC{$\Delta_2; \Gamma_2, x: b \vdash e_2 \; \inhpoly{V_2}$}
		\RightLabel{\textsc{(IP:Let-Base)}}
		\BinaryInfC{$\Delta_1 \cup \Delta_2; \Sigma_1 \cup \Gamma_1 \cup \Gamma_2 \vdash \m{let} \; x = e_1 \; \m{in} \; e_2 \; \inhpoly{V_3}$}
	\end{prooftree}
	\begin{prooftree}
		\AxiomC{$\Delta; \Gamma, x_1: b, x_2: b \vdash e \; \inhpoly{V [x \mapsto x_1, x_2]}$}
		\RightLabel{\textsc{(IP:Share-Base)}}
		\UnaryInfC{$\Delta; \Gamma, x: b \vdash \m{share} \; x \; \m{as} \; x_1, x_2 \; \m{in} \; e \; \inhpoly{V}$}
	\end{prooftree}
	\end{small}
	\caption{Key inference rules of inherently polynomial time.}
	\label{fig:syntax-directed rules of inherently polynomial time}
\end{figure}

In \textsc{(IP:Case-Sum)}, the notation $V [x \mapsto y]$ refers to the result of replacing $x$ in $V$ with $y$ (if $x \in V$); otherwise, $V$ remains intact. 
If the running time of $\m{case} \; x \; \{\ell \cdot y \hookrightarrow e_{\ell} \mid r \cdot y \hookrightarrow e_{r}\}$ in the rule's conclusion is allowed to be polynomial in $\abs{x}$ (i.e.~$x \in V$), then $e_{i \in \{\ell, r\}}$ in the two premises is allowed to run in polynomial time in $\abs{y} = \abs{x} - 1$. 

Similarly, in \textsc{(IP:Case-Prod)}, $V [x \mapsto x_1, x_2]$ means $(V \setminus \{x\}) \cup \{x_1, x_2\}$ if $x \in V$; otherwise, $V$ remains unchanged. 

\textsc{(IP:Rec)} is the crux of the notion of inherently polynomial time. 
Observe that the stepping function $e_1$ must be constant-time in $\abs{z}$ (i.e.~the size of $z$'s semantic value). 

In \textsc{(IP:Let-Base)}, we use a finer-grained notation where the typing context of $e_1$ is split into $\Sigma_1$ for arrow-type variables and $\Gamma_1$ for base-type variables. 
$V_3$ is determined by
\begin{equation*}
V_3 := 
\begin{cases}
\dom{\Gamma_1} \cup (V_2 \setminus\{x\}) & \text{if } x \in V_2; \\
V_1 \cup V_2 & \text{otherwise}. 
\end{cases}
\end{equation*}
If $x \in V_2$, it means that $e_2$ runs in polynomial time in $\abs{x}$.
In the worst case, not only the running time of $e_1$ but $\abs{e_1}$ (i.e.~the output size of $e_1$) is polynomial in the sizes of those variables in $V_1$. 
Hence, in the worst case, the overall running time of $\m{let} \; x = e_1 \; \m{in} \; e_2$ is polynomial in $\dom{\Gamma_1}$, which contains all base-type variables appearing in $e_1$, and $V_2 \setminus \{x\}$. 
Note that \textsc{(IP:Let-Base)} considers the worst case---if we had information about the output size, we might be able to derive a more precise judgment. 

Finally, the judgment \eqref{eq:judgment for arrow-type terms in partial polynomial dependence} is defined by the following inference rules: 
\begin{center}
	\AxiomC{$\Delta; x:b \vdash e \; \inhpoly{\emptyset}$}
	\RightLabel{\textsc{(IP:Const)}}
	\UnaryInfC{$\Delta; \cdot \vdash \lam (x: b). e \; \m{const}$}
	\DisplayProof
	\qquad
	\AxiomC{$\Delta; x:b \vdash e \; \inhpoly{\{x\}}$}
	\RightLabel{\textsc{(IP:Poly)}}
	\UnaryInfC{$\Delta; \cdot \vdash \lam (x: b). e \; \m{poly}$}
	\DisplayProof
\end{center}
In \textsc{(IP:Const)}, because the conclusion indicates that the $\lam$-abstraction's running time is constant in the input size, the premise states that the running time of the body $e$ can only be polynomial in $\dom{\Gamma}$, which excludes $x$. 
By contrast, in the premise of \textsc{(IP:Poly)}, the set of variables contains $x$. 


\section{Typable Fragment of Resource-Aware ML}
\label{sec:typable fragment of resource-aware ML}

It is nontrivial to prove that inherently polynomial time (Section~\ref{sec:formulation of inherently polynomial time}) implies typability in multivariate AARA. 
The chief challenge is to come up with a suitable statement of a typability theorem (i) that we can prove by induction and (ii) that satisfies the following two requirements. 
Firstly, because a term $e$ may later be used as an input to a function, it must be possible to type $e$ such that a user-specified (i.e.~arbitrary) amount of potential remains in $e$'s output. 
Secondly, to type primitive recursion, we need to establish an invariant of resource annotations that is analogous to a loop invariant in Hoare logic. 
Specifically, given a primitive recursion $\m{rec} \; x \; \{[\,] \hookrightarrow e_0 \mid (y::\mi{ys}) \; \m{with} \; z \hookrightarrow e_1 \}$, we must give an (almost) identical annotation to both $z$, which is the result of a recursive call, and $e_1$, which is a stepping function. 

\paragraph*{Typability Theorem}

We have partially overcome this challenge, and this section presents the result that inherently polynomial time implies typability in multivariate AARA under some restrictions. 
Detailed proofs of Theorem~\ref{theorem:inherently polynomial time implies the existence of a multivariate annotation} and Theorem~\ref{theorem:inherently polynomial time implies the existence of a multivariate annotation where a user-specified amount of potential is available in the output} are available in Appendix~\ref{sec:proof of the typability theorem}. 

\begin{definition}[Variables with zero potential]
\label{def:variables that contain zero potential in the context of multivariate AARA}
Let $\Gamma \cup \{v: b\}$ be a base-type typing context and $P$ be its multivariate annotation. 
Variable $v$ is said to contain zero potential in $P$ if and only if $P (i, j) = 0$ for every $i \in \mathcal{I} (\Gamma)$ and $j \in \mathcal{I} (\{v: b\})$ such that $j \neq 0_{b}$. 
In other words, the potential represented by $P$ is constant with respect to $\abs{v}$. 
\end{definition}

\begin{assumption}
\label{assumption:restrictions on variable sharing and pattern matching on nested lists}
Suppose we are given $\Delta; \Sigma; \Gamma \vdash e \; t$ for $t \in \{\inhpoly{V}, \m{const}, \m{poly} \}$. 
For every sub-derivation $\Delta_{s}; \Sigma_{s}; \Gamma_{s} \vdash e_{s} \; \inhpoly{V_{s}}$ inside the derivation of $\Delta; \Sigma; \Gamma \vdash e \; t$, we assume the following:
\begin{itemize}
	\item If $e_{s} \equiv \m{share} \; v \; \m{as} \; v_1, v_2 \; \m{in} \cdots$, then $v$ must be in $V_{s}$; 
	\item If $e_{s} \equiv \m{case} \; x \; \{[\,] \hookrightarrow \cdots \mid (y::\mi{ys}) \hookrightarrow \cdots \}$, then the type of $x$ is of the form $L (b)$ where $b \in \mathbb{B}$ does not contain a list type; that is, $x$ cannot be a nested list.  
\end{itemize}
\end{assumption}

The next theorem establishes that inherently polynomial time implies typability in multivariate AARA under Assumption~\ref{assumption:restrictions on variable sharing and pattern matching on nested lists}, which restricts variable sharing and pattern matching on nested lists. 

\begin{restatable}[Inherently polynomial time implies typability]{theorem}{existenceofannotation}
\label{theorem:inherently polynomial time implies the existence of a multivariate annotation}
Suppose we are given a term $\Sigma; \Gamma \vdash e: b$ with base type $b \in \mathbb{B}$, where $\Delta; \Sigma; \Gamma \vdash e \; \inhpoly{V}$ holds for some $V \subseteq \dom{\Gamma}$. 
Additionally, assume Assumption~\ref{assumption:restrictions on variable sharing and pattern matching on nested lists}.
There exist $P$ and $Q$ satisfying $\Sigma; \Gamma; P \vdash e : \abra{b, Q}$ such that each $v \in \dom{\Gamma} \setminus V$ contains zero potential (Definition~\ref{def:variables that contain zero potential in the context of multivariate AARA}). 

Consider an arrow-type term $\Sigma; \cdot \vdash e: b_1 \rightarrow b_2$ and assume Assumption~\ref{assumption:restrictions on variable sharing and pattern matching on nested lists}.
There exist $P$ and $Q$ such that $\Sigma; \cdot; 1 \vdash e: \abra{b_1, P} \rightarrow \abra{b_2, Q}$. 
Additionally, if $\Delta; \cdot; \Gamma \vdash e \; \m{const}$ is true, $P$ contains constant potential; i.e.~$b_1$ stores zero potential in $P$. 
\end{restatable}

Given a base-type expression $e$, if $\Delta; \Sigma; \Gamma \vdash e \; \inhpoly{V}$ holds, the running time of $e$ is constant in the size of any $v \in \dom{\Gamma} \setminus V$. 
In other words, such $v$ does not contribute to the computational cost of $e$. 
Therefore, it intuitively makes sense that such $v$ contains zero potential in Theorem~\ref{theorem:inherently polynomial time implies the existence of a multivariate annotation}. 

However, Theorem~\ref{theorem:inherently polynomial time implies the existence of a multivariate annotation} cannot be immediately proved by induction on $\inhpoly{V}$, since the statement of the theorem is not strong enough for an inductive proof to go through. 
Specifically, a problem arises in the inductive case for \textsc{(IP:Let-Base)}. 
In a let-binding $\m{let} \; x = e_1 \; \m{in} \; e_2$, $e_1$ must carry sufficient potential to be transferred to $e_2$. 
However, Theorem~\ref{theorem:inherently polynomial time implies the existence of a multivariate annotation} does not allow us to specify how much potential will remain available in the output of $e$. 

Prior to remedying this issue, we first introduce the notion of \emph{uniform resource annotations} for multivariate AARA. 

\begin{definition}[Uniform resource annotations for base types in multivariate AARA]
\label{def:uniform resource annotations in multivariate AARA}
Given a base type $b \in \mathbb{B}$, let $P$ be a multivariate resource annotation of $b$. 
$P$ is said to be a \emph{uniform multivariate annotation} with degree $d \in \mathbb{N}$ and number $n \in \mathbb{N}$ if and only if the following conditions hold
\begin{enumerate}
	\item The maximum degree of $P$ is at most $d$; 
	\item $P (i) = n$ for every $i \in \mathcal{I} (b)$ such that $\degree{i} = d$.
\end{enumerate}
In words, all coefficients of base polynomials with degree $d$ (which should be the maximum degree) are equal to $n$. 
This will be denoted by a judgment $P \; \m{uniform} (d, n)$. 
\end{definition}

\begin{definition}[Uniform annotations for typing contexts in multivariate AARA]
\label{def:uniform resource annotations for typing contexts in multivariate AARA}
Consider a term $\Sigma; \Gamma \vdash e : b$ of base type. 
Suppose that $\Delta; \Sigma; \Gamma \vdash e \; \inhpoly{V}$ holds. 
Let $P$ be a multivariate annotation for the base-type typing context $\Gamma$. 
We say that $P$ is \emph{uniform} with respect to degree $d \in \mathbb{N}$, number $n \in \mathbb{N}$, and set $V$ of variables if and only if the following conditions hold: 
\begin{enumerate}
	\item For any base-type variable $v \in \dom{\Gamma} \setminus V$ of type $b_{v}$, we have
	\begin{equation*}
		\forall i \in \mathcal{I} (\{v: b_{v}\}), j \in \mathcal{I} (\Gamma \setminus \{v: b_{v}\}). \degree{i} > d \implies P (i, j) = 0. 
	\end{equation*}
	In words, for any base polynomial with a non-zero coefficient in $P$, its projection on $v$ must have degree at most $d$.
	\item For any $v \in \dom{\Gamma} \setminus V$ of base type $b_{v}$, we have
	\begin{equation*}
		\forall i \in \mathcal{I} (\{v: b_{v}\}), j \in \mathcal{I} (\Gamma \setminus \{v: b_{v}\}). (\degree{i} = d \land j \neq 0) \implies P (i, j) = 0.
	\end{equation*}
	In words, if a base polynomial has a non-zero coefficient and its projection on $v$ has degree $d$, then the base polynomial is not allowed to involve size variables of any other base-type variables from $\dom{\Gamma}$. 
	\item For any $v \in \dom{\Gamma} \setminus V$ of base type $b_{v}$, we have
	\begin{equation*}
		\forall i \in \mathcal{I} (\{v: b_{v}\}). \degree{i} = d \implies P (i, 0) = n. 
	\end{equation*}
	That is, every base polynomial whose projection on $v$ has degree $d$ has coefficient $n$. 
\end{enumerate}
If these conditions hold, we denote $P$ being a uniform annotation by a judgment $P \; \m{uniform} (d, n, V)$. 
\end{definition}

Note that Definition~\ref{def:uniform resource annotations for typing contexts in multivariate AARA} is a generalization of Definition~\ref{def:uniform resource annotations in multivariate AARA}.
$P \; \m{uniform} (d, n)$ in Definition~\ref{def:uniform resource annotations in multivariate AARA} is equivalent to $P \; \m{uniform} (d, n, \emptyset)$ in Definition~\ref{def:uniform resource annotations for typing contexts in multivariate AARA}. 

Now that we have the notion of uniform annotations in place, we next present Theorem~\ref{theorem:inherently polynomial time implies the existence of a multivariate annotation where a user-specified amount of potential is available in the output} that allows us to specify the amount of potential remaining in the output of a program. 
The major difficulty of the proof lies in establishing an invariant for primitive recursion as explained at the start of Section~\ref{sec:typable fragment of resource-aware ML}. 
We employ the notion of uniform annotations to characterize this invariant. 

\begin{restatable}[Existence of a multivariate annotation with arbitrary potential in the output]{theorem}{typablilitytheorem}
\label{theorem:inherently polynomial time implies the existence of a multivariate annotation where a user-specified amount of potential is available in the output}
Given a term $\Sigma; \Gamma \vdash e: b$ with $b \in \mathbb{B}$, suppose that $\Delta; \Sigma; \Gamma \vdash e \; \inhpoly{V}$ holds, where $V \subseteq \dom{\Gamma}$.
Also, assume Assumption~\ref{assumption:restrictions on variable sharing and pattern matching on nested lists}. 
Fix a multivariate annotation $Q$ for the base type $b$ such that $Q \; \m{uniform} (d, n)$. 
Then there exists a multivariate annotation $P$ such that $\Sigma; \Gamma; P \vdash e: \abra{b, Q}$ under the cost-free metric. 
Furthermore, $P \; \m{uniform} (d, n, V)$ holds. 

Consider an arrow-type term $\Sigma; \cdot \vdash e: b_1 \rightarrow b_2$ and assume Assumption~\ref{assumption:restrictions on variable sharing and pattern matching on nested lists}. 
Fix a multivariate annotation $Q$ for base type $b_2$ such that $Q \; \m{uniform} (d, n)$. 
Then there exists $P$ such that $\Sigma; \cdot; 0 \vdash e: \abra{b_1, P} \rightarrow \abra{b_2, Q}$ under the cost-free metric. 
Furthermore, if $\Delta; \Sigma; \cdot \vdash e \; \m{const}$ is true, $P \; \m{uniform} (d, n)$ holds. 
\end{restatable}

The cost-free metric in Theorem~\ref{theorem:inherently polynomial time implies the existence of a multivariate annotation where a user-specified amount of potential is available in the output} refers to the cost metric in which all evaluation costs are zero. 
For instance, if $f: L (\mathbf{1}) \to L (\mathbf{1})$ is a function that doubles the size of an input list, it can be typed as $f: \abra{L^{2} (\mathbf{1}), 0} \to \abra{L^{1} (\mathbf{1}), 0}$ under the cost-free metric\footnote{For readability, I use univariate AARA instead of multivariate AARA to denote resource-annotated types, although Theorem~\ref{theorem:inherently polynomial time implies the existence of a multivariate annotation where a user-specified amount of potential is available in the output} concerns multivariate AARA}.
That is, the potential stored in each element is halved because the length of the list is doubled. 
If the cost metric is the running time, we instead have $f: \abra{L^{2 + c} (\mathbf{1}), 0} \to \abra{L^{1} (\mathbf{1}), 0}$, where $c$ is the cost of processing each list element. 
The type system of multivariate AARA under the cost-free metric is provided in Appendix~\ref{sec:type system of multivariate AARA}. 
Theorem~\ref{theorem:inherently polynomial time implies the existence of a multivariate annotation where a user-specified amount of potential is available in the output} uses the cost-free metric (as opposed to the running time) since Theorem~\ref{theorem:inherently polynomial time implies the existence of a multivariate annotation} has already considers the cost of evaluating programs. 

Theorem~\ref{theorem:inherently polynomial time implies the existence of a multivariate annotation where a user-specified amount of potential is available in the output} assumes Assumption~\ref{assumption:restrictions on variable sharing and pattern matching on nested lists} as the proof of the theorem poses technical challenges in variable sharing and pattern matching on nested lists.
We will now look at these challenges more closely. 

\paragraph*{Variable Sharing}

Theorem~\ref{theorem:inherently polynomial time implies the existence of a multivariate annotation where a user-specified amount of potential is available in the output} is false if we impose no restrictions on variable sharing. 
To illustrate this, consider $e$ defined as 
\begin{equation} \label{eq:pathological example of variable sharing}
	e := \m{rec} \; x \; \{[\,] \hookrightarrow \abra{\ell, \ell} \mid (y::\uscore) \; \m{with} \; z \hookrightarrow e_1 \}, 
\end{equation}
where the stepping function is $e_1 \equiv \m{case} \; z \; \{\abra{z_1, z_2} \hookrightarrow \m{share} \; z_1 \; \m{as} \; z_{1, 1}, z_{1, 2} \; \m{in} \; \abra{z_{1,1}, z_{1,2}} \}$. 
The typing context of $e$ in \eqref{eq:pathological example of variable sharing} is $\Gamma = \{x: L (\mathbf{1}), \ell: L (\mathbf{1})\}$. 
The stepping function satisfies $e_1 \; \inhpoly{\{y, \mi{ys}\}}$. 
Hence, \eqref{eq:pathological example of variable sharing} is indeed inherently polynomial time. 
However, inside $e_1$, we have $\m{share} \; z_1$, which Assumption~\ref{assumption:restrictions on variable sharing and pattern matching on nested lists} forbids. 

Let $\abra{\ell_1, \ell_2}$ be the output of \eqref{eq:pathological example of variable sharing}. 
Suppose that both $\ell_1$ and $\ell_2$ are to be annotated with $L^{1} (\mathbf{1})$. 
To type \eqref{eq:pathological example of variable sharing} under the cost-free metric such that $\ell_1, \ell_2 : L^{1} (\mathbf{1})$, the typing context $\Gamma$ of $e$ needs to be annotated with $2 \abs{\ell} + \abs{x} \cdot \abs{\ell}$, where $\abs{\cdot}$ denotes the size of an input list. 
Observe that we need to use multivariate AARA rather than univariate AARA to type \eqref{eq:pathological example of variable sharing}. 

In the notation\footnote{Although we are concerned with multivariate AARA, I will use univariate AARA to denote the resource annotation of $e_1$ because it happens to be describable by univariate AARA and it is easier to read.} of univariate AARA, the stepping function of \eqref{eq:pathological example of variable sharing} can be typed as 
\begin{equation*}
	y: \mathbf{1}, \mi{ys}: L^{0} (\mathbf{1}), z: L^2 (\mathbf{1}) \times L^0 (\mathbf{1}); 0 \vdash e_1 : \abra{L^1 (\mathbf{1}) \times L^1 (\mathbf{1}), 0}. 
\end{equation*}
Here, the maximum degree is $d = 1$. 
It is impossible for both $z$ and $e_1$ to have the same coefficient for all base polynomials of degree $d = 1$. 
Therefore, Theorem~\ref{theorem:inherently polynomial time implies the existence of a multivariate annotation where a user-specified amount of potential is available in the output} is false for \eqref{eq:pathological example of variable sharing}. 
To accommodate the multivariate annotation of \eqref{eq:pathological example of variable sharing}, it is necessary to relax the notion of uniform resource annotations, but this will make the typability proof more challenging. 

\paragraph*{Nested Lists in Pattern Matching}

Theorem~\ref{theorem:inherently polynomial time implies the existence of a multivariate annotation where a user-specified amount of potential is available in the output} is false for pattern matching on nested lists. 
For example, consider $e$ defined as
\begin{equation*}
	e := \m{case} \; x \; \{[\,] \hookrightarrow \uscore \mid (y::\mi{ys}) \hookrightarrow \abra{y, \mi{ys} } \}, 
\end{equation*}
where the first branch is unimportant in the present discussion. 
The typing context of $e$ is $\Gamma = \{x: L (L (\mathbf{1})) \}$. 
Assume that we consider multivariate annotations of degree up to $d = 2$. 
Let $P$ denote a multivariate annotation of $\Gamma$. 
The multivariate annotation for context $\{y: L (\mathbf{1}), \mi{ys}: L (L (\mathbf{1})) \}$ as a result of pattern matching on $x: L (L (\mathbf{1}))$ is given by the \emph{additive shift} of $P$, denoted by $\lhd (P)$.
It is defined as
\begin{equation} \label{eq:definition of multivariate additive shift in the discussion on a technical difficulty posed by nested lists in pattern matching}
	\lhd (P) (i, j) :=
	\begin{cases}
	P (0_{L (\mathbf{1})}:: j) + P (j) & \text{if } i = 0_{L (\mathbf{1})}; \\
	P (i::j) & \text{otherwise},
	\end{cases}
\end{equation}
where $i \in \mathcal{I} (\{y: L (\mathbf{1} )\})$ and $j \in \mathcal{I} (\{\mi{ys}: L (L (\mathbf{1}))\})$. 
The problem is that the base polynomial $(i, j)$ on the left hand side of \eqref{eq:definition of multivariate additive shift in the discussion on a technical difficulty posed by nested lists in pattern matching} has degree $\degree{i} + \degree{j}$, while $(i::j)$ in the second branch of the right hand side has degree $1 + \degree{i} + \degree{j}$. 
As a consequence, if $1 + \degree{i} + \degree{j} = 2$, $P (i::j)$ is required to be equal to $n$ because Theorem~\ref{theorem:inherently polynomial time implies the existence of a multivariate annotation where a user-specified amount of potential is available in the output} requires $P \; \m{uniform} (d, n)$ to be true. 
This means $\lhd (P) (i, j) = n$ must hold as well.
But $\lhd (P) (i, j) = n$ is not necessarily the case, since Theorem~\ref{theorem:inherently polynomial time implies the existence of a multivariate annotation where a user-specified amount of potential is available in the output} imposes no requirements on the coefficients of lower-degree base polynomials. 


\section{Conclusion}

In this work, we have shown that polynomial-time Turing machines can be embedded in a typable fragment of RaML in such a way that the semantics and worst-case cost bounds are preserved. 
Moreover, we have proved that if a first-order program $P$ satisfies the following conditions, it is guaranteed to be typable in multivariate polynomial AARA:
\begin{enumerate}
	\item $P$ uses primitive recursion instead of general recursion;
	\item $P$ is (axiomatically) inherently polynomial-time; 
	\item No variable sharing is applied to variable $v$, where $P$'s running time is (axiomatically) constant in $v$; 
	\item No pattern matching is applied to a nested list. 
\end{enumerate}

We have neither found a counterexample to the full typability theorem (i.e.~Theorem~\ref{theorem:inherently polynomial time implies the existence of a multivariate annotation} without Assumption~\ref{assumption:restrictions on variable sharing and pattern matching on nested lists}) nor proved it. 
As future work, we are looking to investigate how to prove or disprove the full typability theorem. 
To lift the restriction on nested lists, we expect that it suffices to modify the statement of the theorem such that we can keep track of the largest coefficient. 
However, lifting the restriction on variable sharing will be more challenging because it certainly requires a drastically different inductive hypothesis.



\bibliography{references}

\appendix

\ifLongVersion

\section{Supplementary Results for the Embedding of Polynomial-Time Turing Machines in AARA}
\label{sec:upplementary results for the embedding of polynomial-time Turing machines in AARA}

Throughout Appendix~\ref{sec:upplementary results for the embedding of polynomial-time Turing machines in AARA}, we will write $b_1 \xrightarrow{q_1/q_2} b_2$ for resource-annotated arrow types, instead of $\abra{b_1, q_1} \to \abra{b_2, q_2}$.
Here, $b_i$ is a resource-annotated base type and $q_{i} \in \mathbb{Q}_{\geq 0}$. 

\subsection{Standard-Form Polynomials in Univariate AARA}

In conventional mathematics taught at school, (univariate) polynomials are expressed in the standard form of $a_n x^{n} + \cdots + a_0 x^{0}$, where $x$ is a variable and $a_i$'s are coefficients. 
On the other hand, in univariate AARA, polynomials are encoded as linear combinations of binomial coefficients $\binom{x}{a}$, where $x$ is a variable and $a$ is a constant. 
The following proposition establishes that these two representations of univariate polynomials are interchangeable. 

\begin{lemma} 
\label{lemma:polynomial can be expressed as a linear combination of binomial coefficients}
For any $d \in \mathbb{N}$, the polynomial function $n^{d}$ can be expressed as $\sum_{i = 0}^{d} q_{i} \binom{n}{i}$, where $q_i \in \mathbb{Q}_{\geq 0}$ for all $0 \leq i \leq d$. 
\end{lemma}
\begin{proof}
The proof goes by induction on $d$. 
The claim clearly holds when $d = 0$.

For the inductive case, by way of example, we will first illustrate how to prove the lemma for $d = 4$, given that the claim holds when $d = 3$. 
For simplicity, we will use the basis of $i! \cdot \binom{n}{i}$ rather than $\binom{n}{i}$ in this example. 
Suppose we have 
\begin{equation*}
	n^3 = p_1 n + p_2 n (n-1) + p_3 n (n-1) (n-2).
\end{equation*}
This yields
\begin{align*}
	n^4 & = n \cdot \left( p_1 n + p_2 n (n-1) + p_3 n (n-1) (n-2) \right) \\
	& = p_1 n^2 + p_2 n^2 (n-1) + p_3 n^2 (n-1) (n-2) \\
	& = p_1 n ((n-1) + 1) + p_2 n (n-1) ((n-2) + 2) + p_3 n (n-1) (n-2) ((n-3) + 3) \\
	& = p_1 n + (p_1 + 2 p_2) n (n-1) + (p_2 + 3 p_3) n (n-1) (n-2) + p_3 n (n-1) (n-2) (n-3),
\end{align*}
where all coefficients are non-negative, provided that each $p_i$ is non-negative as well. 
Generalizing this, we learn that $n^{k} = \sum_{i = 0}^{k} p_i \cdot n (n-1) \cdots (n - i + 1)$ gives
\begin{equation*}
	n^{k+1} = p_k \cdot n (n-1) \cdots (n-k) + \sum_{i = 1}^{k} (p_{i-1} + i \cdot p_{i}) \cdot n (n-1) \cdots (n-i+1).  
\end{equation*}
If the coefficients $p_i$ for $n^{k}$ are non-negative, so are the coefficients for $n^{k+1}$.

Finally, to switch from the new basis to the original basis of binomial coefficients, we use the identity $p_i = \frac{q_i}{i !}$. 
This gives
\begin{equation} \label{eq:coefficients of n to the power of k+1}
	\begin{split}
	n^{k+1} & = \frac{q_k}{k !} \cdot n (n-1) \cdots (n-k) + \sum_{i = 1}^{k} \left( \frac{q_{i-1}}{(i-1)!} + i \cdot \frac{q_{i}}{i!} \right) \cdot n (n-1) \cdots (n-i+1) \\
	& = q_k k \cdot \binom{n}{k+1} + \sum_{i = 1}^{k} i (q_{i-1} + q_{i}) \binom{n}{i}. 
	\end{split}
\end{equation}
This concludes the proof. 
\end{proof}

\subsection{Generating Lists of Polynomial Size}

We will explain how to generate, in Resource-Aware ML (RaML), a list of polynomial size $p (n)$ with constant potential stored in each cell. 
To this end, it suffices to show how to generate a list of size $\binom{n}{d}$ for a fixed $d \in \mathbb{N}$ because $n^{k}$ for any $k \in \mathbb{N}$ can be expressed as a non-negative linear combination of $\binom{n}{0}, \ldots, \binom{n}{k}$. 
This has formally been established by Lemma~\ref{lemma:polynomial can be expressed as a linear combination of binomial coefficients}.

Without loss of generality, we assume that the output is a list of blank symbols (as required in line~\ref{algline:l_2 is generated} of Algorithm~\ref{alg:target RaML program}). 
Let $\mi{amp}_{d} : L(\text{Sym}) \rightarrow L(\text{Sym}) \rightarrow L(\text{Sym})$ denote a RaML function that (i) generates a list of size $\binom{n}{d}$ in which each cell stores one unit of potential and (ii) appends it to an accumulator, which is assumed to already contain one unit of potential in each cell. 
$n$ is the size of the first input to $\mi{amp}_{d}$, and the accumulator is the second input. 
Here, $\mi{amp}$ stands for amplification.

$\mi{amp}_{0}$ is defined as 
\begin{equation} \label{eq:definition of amp_0}
	\m{fun} \; \mi{amp}_{0} \; w \; \mi{acc}= \sqcup :: \mi{acc}. 
\end{equation}
For $i \geq 0$, $\mi{amp}_{i+1}$ is inductively defined as
\begin{equation} \label{eq:definition of amp_d}
	\begin{split}
	\m{fun} \; \mi{amp}_{i+1} \; w \; \mi{acc} = \m{case} \; w \; & \{[\,] \hookrightarrow \mi{acc} \\
	& \mid x :: \mi{xs} \hookrightarrow \m{share} \; \mi{xs} \; \m{as} \; \mi{xs}_1, \mi{xs}_2 \; \m{in} \\
	& \qquad \qquad \quad \m{let} \; \uscore = \m{tick} \; 1 \; \m{in} \\
	& \qquad \qquad \quad \m{let} \; \mi{acc}' = \mi{amp}_{i} \; \mi{xs}_1 \; \mi{acc} \; \m{in} \\ 
	& \qquad \qquad \quad \mi{amp}_{i+1} \; \mi{xs}_2 \; \mi{acc}' \}. 
	\end{split}
\end{equation}
We use $\m{tick}$ to account for only the cost of function application but not costs of other operations such as the list constructor. 
This is why \eqref{eq:definition of amp_0} does not generate any costs. 
Although the syntax of RaML presented in Section~\ref{sec:resource-aware ML} only permits uncurried functions, we will use curried functions throughout Appendix~\ref{sec:upplementary results for the embedding of polynomial-time Turing machines in AARA} without loss of generality. 

This implementation is analogous to the example given in Section~7.1 of \cite{Hoffmann2010}, where given a list $\ell$, all subsets of $\ell$ with a fixed size are computed. 
However, our implementation in \eqref{eq:definition of amp_d} differs from the implementation in \cite{Hoffmann2010} in that ours uses an accumulator, while the one in \cite{Hoffmann2010} explicitly uses the $\mi{append}$ function. 
The use of an accumulator allows us to embed both generation of elements and their concatenation in the implementation of $\mi{amp}_{d}$, saving us the need to explicitly reason about the computational cost of $\mi{append}$. 
Hence, using an accumulator can simplify the cost analysis of $\mi{amp}_{d}$, although admittedly $\mi{amp}_{d}$ with an accumulator is not the most natural implementation from the perspective of programmers. 

We consider the generation of lists of size $\binom{n}{d}$ instead of $n^{d}$ for a similar reason.
Due to the identity $\binom{n+1}{d+1} = \binom{n}{d+1} + \binom{n}{d}$, we do not need to appeal to any auxiliary function in \eqref{eq:definition of amp_d}. 
On the other hand, to recursively create a list of size $n^{d} = n \cdot n^{d-1}$, one needs to use an iterator function, and this will complicate the analysis of total costs since we will need to account for the cost of invoking the iterator function. 

The next proposition establishes the correctness of the above implementation (i.e.~\eqref{eq:definition of amp_0} and \eqref{eq:definition of amp_d}) and provides an upper bound on the evaluation cost. 

\begin{lemma}[Correctness of $\mi{amp}_{d}$]
\label{lemma:correctness of the implementation of amp_d}
The computation of $\mi{amp}_{d} \; w \; \mi{acc}$ produces a list of size $\binom{\abs{w}}{d} + \abs{\mi{acc}}$, where $\binom{n}{k} = 0$ for $n < k$. 
Also, assuming that each cell in the output list is required to contain one unit of potential, the cost of evaluating $\mi{amp}_{d} \; w \; \mi{acc}$ is bounded above by $2 \abs{w}^{d}$.  
\end{lemma}
\begin{proof}
The proof goes by nested induction: outer induction on $d$ and inner induction on $\abs{w}$.
For the base case where $d = 0$, the output size is indeed $\binom{n}{0} + \abs{\mi{acc}} = 1 + \abs{\mi{acc}}$. 
With regard to the evaluation cost, we need one unit of potential to execute \eqref{eq:definition of amp_0} since the new cell requires one unit of potential. 
As $2 n^{0} = 2$ for every $n \in \mathbb{N}$, $2 n^{d}$ is a correct upper bound in this case. 
Here, we adopt the convention of $0^{0} = 1$. 

For the inductive case, suppose that the claim holds when $d = k$ for some $k \geq 0$. 
The proof proceeds by (inner) induction on $\abs{w}$.
When $\abs{w} = 0$, we have
\begin{equation*}
	\abs{\mi{amp}_{k+1} \; w \; \mi{acc}} = 0 + \abs{\mi{acc}}
\end{equation*}  
according to the first branch of pattern matching in \eqref{eq:definition of amp_d}.
Also, the evaluation cost is 0.
Hence, the claim holds when $w$ is empty.

Conversely, if $w = x::\mi{xs}$, we have
\begin{alignat*}{2}
	\abs{\mi{amp}_{k+1} \; w \; \mi{acc}} & = \abs{\mi{amp}_{k+1} \; \mi{xs} \; (\mi{amp}_{k} \; \mi{xs} \; \mi{acc})} &\qquad& \text{by \eqref{eq:definition of amp_d}} \\
	& = \binom{\abs{xs}}{k+1} + \binom{\abs{xs}}{k} + \abs{\mi{acc}} && \text{by the inductive hypothesis} \\
	& = \binom{\abs{w} - 1}{k+1} + \binom{\abs{w} - 1}{k} + \abs{\mi{acc}} && \text{because }  w = x:: xs \\
	& = \binom{\abs{w}}{k+1} + \abs{\mi{acc}}. 
\end{alignat*}

Regarding the evaluation cost, we write $\text{cost} (\mi{amp}_{k+1} \; w \; \mi{acc})$ for the evaluation cost of $\mi{amp}_{k+1} \; w \; \mi{acc}$. 
If $k \geq 1$, we have
\begin{alignat*}{2}
	\text{cost} (\mi{amp}_{k+1} \; w \; \mi{acc}) & = 1 + \text{cost} (\mi{amp}_{k} \; \mi{xs} \; \mi{acc}) + \text{cost} (\mi{amp}_{k+1} \; \mi{xs} \; \mi{acc}') &\qquad& \text{by \eqref{eq:definition of amp_d}} \\
	& \leq 1 + 2 \abs{\mi{xs}}^{k} + 2 \abs{\mi{xs}}^{k+1} && \text{by the inductive hypothesis} \\
	& \leq 2 (1 + \abs{\mi{xs}})^{k} + 2 \abs{\mi{xs}}^{k+1} && \text{because } k \geq 1 \\
	& \leq 2 (1 + \abs{\mi{xs}})^{k} + 2 \abs{\mi{xs}} \cdot (1+\abs{\mi{xs}})^{k} \\
	& = 2 (1+\abs{\mi{xs}})^{k+1} \\
	& = 2 (\abs{w})^{k+1}. 
\end{alignat*}
If $k = 0$, we have
\begin{alignat*}{2}
	\text{cost} (\mi{amp}_{1} \; \mi{xs} \; \mi{acc}) & = 1 + \text{cost} (\mi{amp}_{0} \; \mi{xs} \; \mi{acc}) + \text{cost} (\mi{amp}_{1} \; w \; \mi{acc}') &\qquad& \text{by \eqref{eq:definition of amp_d}} \\
	& \leq 1 + 1 + 2 \abs{\mi{xs}} && \text{by the inductive hypothesis} \\
	& = 2 (\abs{w}),
\end{alignat*}
where in the second line, we use the tight bound $\text{cost} (\mi{amp}_{0} \; w \; \mi{acc}) = 1$. 
Therefore, the claim is true regardless of whether $k = 0$ or $k \geq 1$.
This concludes the proof. 
\end{proof}

$2 n^{d}$ is a tight cost bound of $\mi{amp}_{d}$ when $d = 1$. 
However, $2 n^{d}$ is not a tight bound anymore when $d = 0$ or $d > 1$. 
The general tight bound is probably complicated to express. 

The next proposition claims that AARA can infer that $2 n^{d}$ is an upper bound on the evaluation cost. 

\begin{lemma}[Typability of $\mi{amp}_{d}$]
\label{lemma:resource annotation of amp_d}
AARA can infer the resource-annotated type
\begin{equation} \label{eq:resource annotation of amp_d}
	\mi{amp}_{d}: L^{2 \vec{q}_{d}} (\text{Sym}) \xrightarrow{0/0} L^{1}(\text{Sym}) \xrightarrow{0/0} L^{1}(\text{Sym}), 
\end{equation}
where vector $\vec{q}_{d} \in \mathbb{Q}_{\geq 0}^{d}$ represents the function $n \mapsto n^{d}$. 
To be more precise, since vectors from $\mathbb{Q}_{\geq 0}^{0}$ cannot express constants, when $d = 0$, the resource-annotated type should be written as 
\begin{equation} \label{eq:resource annotation of amp_0}
	\mi{amp}_{0}: L^{0} (\text{Sym}) \xrightarrow{2/0} L^{1}(\text{Sym}) \xrightarrow{0/0} \abra{L^{1}(\text{Sym}), 0}. 
\end{equation}
Keep in mind that the resource annotation that AARA returns in reality can be a more accurate bound than \eqref{eq:resource annotation of amp_d}. 
\end{lemma}
\begin{proof}
To prove the claim, it is sufficient to show that \eqref{eq:resource annotation of amp_d} is a valid resource annotation that satisfies all relevant typing rules of AARA.
The proof goes by induction on $d$.
For the base case of $d = 0$, the claim holds since $\mi{amp}_{d}$ requires exactly 1 potential unit.

We now turn to the inductive case.  
We can assign resource-annotated types to some variables appearing in \eqref{eq:definition of amp_d} as
\begin{equation*}
	w : L(\abra{\text{Sym}, 2 \vec{q}_{d}}) \qquad \mi{xs}_1 : L(\abra{\text{Sym}, (\triangleleft \, 2 \vec{q}_{d}) - 2 \vec{q}_{d}})  \qquad \mi{xs}_2 : L(\abra{\text{Sym}, 2 \vec{q}_{d}}).
\end{equation*}

We will now argue that this annotation correctly accounts for the evaluation cost in every recursive call. 

Firstly, from \eqref{eq:coefficients of n to the power of k+1}, we can derive that the first component of vector $\vec{q}_{d}$ for any $d \geq 1$ is 1. 
This means that 2 units of potential is available in each recursive call.
Hence, we use this constant potential to account for $\m{tick} \; 1$ in the definition of $\mi{amp}_{d}$. 

After deducting 2 units from the potential stored in $w$, we have 
$\triangleleft \, 2 \vec{q}_{d}$ units of potential remaining, and this represents $2 n^{d} - 2$.
We need to split it between $\mi{xs}_1$ and $\mi{xs}_2$ in such a way that we can pay for the costs of $\mi{amp}_{d-1} \; \mi{xs}_1 \; \mi{acc}$ and $\mi{amp}_{d} \; \mi{xs}_2 \; \mi{acc}'$. 
For the former, the inductive hypothesis suggests $\mi{amp}_{d - 1} : L^{2 \vec{q}_{d-1}} (\text{Sym}) \xrightarrow{0/0} L^{1}(\text{Sym}) \xrightarrow{0/0} L^{1}(\text{Sym})$, provided that $d \geq 2$. 
If $d = 1$, we need to conduct separate analysis since the type of $\mi{amp}_{0}$ in \eqref{eq:resource annotation of amp_0} is distinct from the type for $\mi{amp}_{d}$ for $d \geq 1$.
Nonetheless, we will assume $d \geq 2$ in the present proof as it is straightforward to adapt this proof to the case of $d = 1$. 
For the recursive call $\mi{amp}_{d} \; \mi{xs}_2 \; \mi{acc}'$, ideally, we would like to reuse the resource-annotated type of $\mi{amp}_{d}$; otherwise, resource-polymorphic recursion would arise, complicating the proof. 
As a consequence, our goal is to show 
\begin{equation*}
	2 \vec{q}_{d-1} + 2 \vec{q}_{d} \leq \triangleleft \, 2 \vec{q}_{d},
\end{equation*}
where $+$ and $\leq$ are applied component-wise. 
This is equivalent to $\vec{q}_{d-1} + \vec{q}_{d} \leq \triangleleft \, \vec{q}_{d}$ because $\triangleleft$ is linear and hence $\triangleleft \, 2 \vec{q}_{d} = 2 \cdot (\triangleleft \, \vec{q}_{d})$, where $\cdot$ is scalar multiplication. 

If $\vec{q}_{d - 1} = (q_1, q_2, \ldots, q_{d-1})$, it follows from \eqref{eq:coefficients of n to the power of k+1} that
\begin{equation*}
	\vec{q}_{d} = (q_1, 2 (q_1 + q_2), 3 (q_2 + q_3), \dots, (d-1) (q_{d-2} + q_{d-1}), (d-1) q_{q}).
\end{equation*} 
This yields 
\begin{equation*}
	\vec{q}_{d-1} + \vec{q}_{d} = (2 q_1, 2q_1 + 3q_2, \ldots, (d-1) q_{d-1} + d q_{d-1}, (d-1) q_{d-1}),
\end{equation*}
which is smaller than $\triangleleft \, \vec{q}_{d}$ component-wise.
Therefore, $\vec{q}_{d-1} + \vec{q}_{d} \leq \triangleleft \, \vec{q}_{d}$ indeed holds. 
\end{proof}

\subsection{Target RaML Programs}
\label{sec:target RaML programs in the embedding of polynomial-time Turing machines}

For convenience, the definition of Turing machines is reproduced below. 
\Turingmachine*

Given a source program $M$, the target program $M'$ can be expressed as 
\begin{equation} \label{eq:definition of M'}
\begin{split}
	\m{fun} \; M' \; w = {} & \m{share} \; w \; \m{as} \; w_1, w_2, w_3 \; \m{as} \\
	& \m{let} \; \ell_1 = {\vdash} :: [\,] \; \m{in} \\
	& \m{let} \; \ell_2' = \mi{amp}_{d, \sqcup} \; w_1 \; [\,] \; \m{in} \\
	& \m{let} \; \ell_2 = \mi{append} \; w_2 \; \ell_2' \; \m{in} \\
	& \m{let} \; \mi{ps} = \mi{amp}_{d, \abra{\,}} \; w_3 \; [\,] \; \m{in} \\
	& \mi{simulate} \; q_0 \; \ell_1 \; \ell_2 \; \mi{ps},
\end{split}
\end{equation}
where $\mi{amp}_{d, \sqcup}$ creates a list of size $\binom{\abs{w}}{d}$ filled with blank symbols, and $\mi{amp}_{d, \abra{\,}}$ performs the same task, except that the output is filled with $\abra{\,}$ instead of $\sqcup$. 
If $p (n)$, which is the polynomial representing $M$'s running time, cannot be expressed in the form of $\binom{n}{d}$ for any $d \in \mathbb{N}$, we express $p (n)$ as a linear combination of binomial coefficients (due to Lemma~\ref{lemma:polynomial can be expressed as a linear combination of binomial coefficients}) and hard-code this linear combination inside \eqref{eq:definition of M'}. 

The auxiliary functions $\mi{append}$ and $\mi{simulate}$ are defined as
\begin{alignat*}{2}
		\m{fun} \; \mi{append} \; \ell_1 \; \ell_2 = \m{case} \; \ell_1 \; & \{ [\,] \hookrightarrow \ell_2 \\
		& \mid x::\mi{xs} \hookrightarrow {} && \m{let} \; \uscore = \m{tick} \; 1 \; \m{in} \\
		& &&\m{let} \; \mi{xs}' = \mi{append} \; \mi{xs} \; \ell_2 \; \m{in} \\
		& && x:: \mi{xs}'
		\}
\end{alignat*}

\begin{equation} \label{eq:definition of simulate}
	\begin{split}
	\m{fun} \; \mi{simulate} \; & s \; \ell_1 \; \ell_2 \; \mi{ps} \\
	= \m{case} \;\mi{ps} \; & \{ [\,] \hookrightarrow && \mi{shift} \; \ell_1 \; \ell_2 \\
	& \mid p::\mi{ps}' \hookrightarrow {} &&\m{let} \; \uscore = \m{tick} \; 1 \; \m{in} \\
	& &&\m{let} \; (s', b, \text{direction}) = \delta (s, \mi{head} \; \ell_2) \; \m{in} \\
	& && \m{if} \; s' = q_{\text{final}} \; \m{then} \\
	& && \mi{shift} \; \ell_1 \; \ell_2 \\
	& && \m{else} \; \m{if} \; \text{direction} = L \; \m{then} \\
	& && \mi{simulate} \; s' \; (\mi{tail} \; \ell_1) \; ((\mi{head} \; \ell_1) :: b :: (\mi{tail} \; \ell_2)) \; \mi{ps}' \\
	& && \m{else} \\
	& && \mi{simulate} \; s' \; (b :: \ell_1) \; (\mi{tail} \; \ell_2) \; \mi{ps}' \}.
	\end{split}
\end{equation}
The if-else constructs in \eqref{eq:definition of simulate} (e.g.~$\m{if} \; s' = q_{\text{final}} \; \m{then} \; \cdots$) can be encoded using $\m{case}$ for the Boolean type (i.e.~$\mathbf{1} + \mathbf{1}$). 
The function $\mi{shift}$ used in \eqref{eq:definition of simulate} reverses the first input list and appends it to the second input list:
\begin{equation*}
	\begin{split}
	\m{fun} \; \mi{shift} \; \ell_1 \; \ell_2 = \m{case} \; \ell_1 \; & \{ [\,] \hookrightarrow \ell_2 \\
	& \mid x::\mi{xs} \hookrightarrow \m{let} \; \uscore = \m{tick} \; 1 \; \m{in} \\
	& \qquad \qquad \quad \m{let} \; \mi{ys} = x :: \ell_2 \; \m{in} \\
	& \qquad \qquad \quad \mi{shift} \; xs \; ys \}. 
	\end{split}
\end{equation*}

In \eqref{eq:definition of simulate}, for the sake of brevity, we use the standard form of function application in place of let-normal form. 
Also, $(s', b, \text{direction}) = \delta (s, \mi{head} \; \ell_2)$ is a slight abuse of notation because this is ill-formed with respect to the syntax of RaML and also because we would need to introduce a new type for directions (i.e.~$L$ or $R$).
Nevertheless, we write it this way to keep $\delta$ general.
If we are given a specific transition function, we can embed it in the code, dedicating one branch of if-else statements (or pattern matching) to each possible combination of $s$ (i.e.~the current machine state) and $\mi{head} \; \ell_2$ (i.e.~the symbol in the current cell). 

Lastly, for completeness, $\mi{head}$ and $\mi{tail}$ are defined as
\begin{align*}
\m{fun} \; \mi{head} \; \ell & = \m{case} \; \ell \; \{ [\,] \hookrightarrow \m{error} \mid x::xs \hookrightarrow x \} \\
\m{fun} \; \mi{tail} \; \ell & = \m{case} \; \ell \; \{ [\,] \hookrightarrow \m{error} \mid x::\mi{xs} \hookrightarrow \mi{xs} \}. 
\end{align*}

Finally, we are now in a position to prove the theorem about embedding polynomial-time Turing machines in RaML. 

\embeddingofTuringmachines*

\begin{proof}
Without loss of generality, assume that the running time of $M$ is bounded by $n \mapsto \binom{n}{d}$ for some fixed $d \in \mathbb{N}$. 
In this case, a desirable $M'$ is defined in \eqref{eq:definition of M'}. 
If this assumption is false, we can use Lemma~\ref{lemma:polynomial can be expressed as a linear combination of binomial coefficients} to express a polynomial as a liner combination of binomial coefficients and hard-code it in \eqref{eq:definition of M'}. 
By construction, $M' (w) = M (w)$ for every $w \in \{0, 1\}^{*}$. 

Throughout the execution of $M'$, $\ell_1$, $\ell_2$, and $\mi{ps}$ must contain one unit of potential in each cell.
The potential in $\ell_1$ and $\ell_2$ will be used to account for $\mi{shift} \; \ell_1 \; \ell_2$ right before $M'$ terminates. 
The potential stored in $\mi{ps}$ is for the execution of $\mi{simulate}$. 

Hence, the input $w$ to $M'$ must contain sufficient potential to pay for the following costs:
\begin{itemize}
	\item Creating singleton list $\ell_1$ that contains $\vdash$. 
	Due to the invariant we impose on $\ell_1$'s potential, it requires one unit of potential to create $\ell_1$ in the initial configuration. 
	\item Creating list $\ell_2'$ of size $\binom{\abs{w}}{d}$, which has one unit of potential in each cell. 
	\item Appending $w$ to $\ell_2'$ to create $\ell_2$, which has size $\abs{w} + \binom{\abs{w}}{d}$ and stores one unit of potential in each cell. 
	Thus, we have $\abs{w}$ many units of potential to execute $\mi{append}$ and another $\abs{w}$ units to be stored in the first $\abs{w}$ cells of $\ell_2$. This gives a total of $2 \abs{w}$ units of potential. 
	\item Creating list $\mi{ps}$ of size $\binom{\abs{w}}{d}$, which stores one unit of potential in each cell. 
\end{itemize}
By Proposition~\ref{lemma:resource annotation of amp_d}, the second (for $\ell_2'$) and fourth (for $\mi{ps}$) costs above can be each covered by $2 \abs{w}^{d}$ units of potential. 
Hence, summing the above three costs, we obtain $1 + 2 \abs{w}^{d} + 2 \abs{w}^{d} + 2 \abs{w}$.
This is the amount of potential that must be stored in the input $w$ to $M'$ at the start of computation. 

Each auxiliary function appearing in \eqref{eq:definition of M'} can be type-annotated. 
More concretely, $\mi{amp}_{d}$ can be type-annotated as shown in Proposition~\ref{lemma:resource annotation of amp_d}, and $\mi{simulate}$ can be assigned this type: 
\begin{equation*}
\mi{simulate} : \text{State} \rightarrow L^{1}(\text{Sym}) \xrightarrow{0/0} L^{1}(\text{Sym}) \xrightarrow{0/0} L^{1}(\text{Sym}) \xrightarrow{0/0} L^{0}(\text{Sym}). 
\end{equation*} 
Since it is relatively easy to see that this type can be inferred using AARA, we will omit its formal proof. 
In summary, univariate AARA can infer a polynomial cost bound of $M'$. 
\end{proof}

\else
\fi


\section{Resource-Aware ML (RaML)}

This section gives the type system of RaML (Section~\ref{sec:resource-aware ML}) and defines its the running time.

\subsection{Simple Type System of RaML}
\label{sec:type system of RaML}

The simple (i.e.~non-resource-annotated) type system of RaML is displayed in Figure~\ref{fig:type system of RaML}. 
Throughout Figure~\ref{fig:type system of RaML}, $b \in \mathbb{B}$ denotes a base type, and $\tau$ denotes a simple type. 
Because RaML is a first-order language, $\tau$ is either $b \in \mathbb{B}$ or $b_1 \to b_2$ (Section~\ref{sec:basics of univariate AARA}). 

\begin{figure}[hbt!]
	\begin{small}
	\begin{center}
		\AxiomC{}
		\RightLabel{\textsc{(T:Var)}}
		\UnaryInfC{$x: \tau \vdash x: \tau$}
		\DisplayProof
		\quad
		\AxiomC{$\Gamma \vdash x: b_1$}
		\RightLabel{\textsc{(T:SumL)}}
		\UnaryInfC{$\Gamma \vdash \ell \cdot x : b_1 + b_2$}
		\DisplayProof
		\quad
		\AxiomC{$\Gamma \vdash x: b_2$}
		\RightLabel{\textsc{(T:SumR)}}
		\UnaryInfC{$\Gamma \vdash r \cdot x : b_1 + b_2$}
		\DisplayProof
	\end{center}
	\begin{center}
		\AxiomC{}
		\RightLabel{\textsc{(T:Unit)}}
		\UnaryInfC{$\cdot \vdash \abra{\,} : \mathbf{1}$}
		\DisplayProof
		\quad
		\AxiomC{$\Gamma_1 \vdash x_1: b_1$}
		\AxiomC{$\Gamma_2 \vdash x_2: b_2$}
		\RightLabel{\textsc{(T:Pair)}}
		\BinaryInfC{$\Gamma_1 \cup \Gamma_2 \vdash \abra{x_1, x_2}: b_1 \times b_2$}
		\DisplayProof
	\end{center}
	\begin{center}
		\AxiomC{}
		\RightLabel{\textsc{(T:Nil)}}
		\UnaryInfC{$\cdot \vdash [\,]: L(b)$}
		\DisplayProof
		\quad
		\AxiomC{$\Gamma_1 \vdash x_1: b$}
		\AxiomC{$\Gamma_2 \vdash x_2: L (b)$}
		\RightLabel{\textsc{(T:Cons)}}
		\BinaryInfC{$\Gamma_1 \cup \Gamma_2 \vdash (x_1 :: x_2) : L(b)$}
		\DisplayProof
	\end{center}

	\begin{center}
		\AxiomC{$\Gamma, f: b_1 \to b_2, x: b_1 \vdash e: b_2$}
		\RightLabel{\textsc{(T:Fun)}}
		\UnaryInfC{$\Gamma \vdash \m{fun} \; f \; x = e : b_1 \to b_2$}
		\DisplayProof
		\quad
		\AxiomC{$\Gamma_1 \vdash x_1: b_1 \rightarrow b_2$}
		\AxiomC{$\Gamma_2 \vdash x_2: b_1$}
		\RightLabel{\textsc{(T:App)}}
		\BinaryInfC{$\Gamma_1 \cup \Gamma_2 \vdash x_1 \; x_2 : b_2$}
		\DisplayProof
	\end{center}

	\begin{prooftree}
		\AxiomC{$\Gamma_1 \vdash x: b_1 + b_2$}
		\AxiomC{$\Gamma_2, y: b_1 \vdash e_{\ell}: b_3$}
		\AxiomC{$\Gamma_2, y: b_2 \vdash e_{r}: b_3$}
		\RightLabel{\textsc{(T:Case-Sum)}}
		\TrinaryInfC{$\Gamma_1 \cup \Gamma_2 \vdash \m{case} \; x \; \{\ell \cdot y \hookrightarrow e_{\ell} \mid r \cdot y \hookrightarrow e_{r} \} : b_3$}
	\end{prooftree}
	\begin{prooftree}
		\AxiomC{$\Gamma_1 \vdash x: b_1 \times b_2$}
		\AxiomC{$\Gamma_2, x_1: b_1, x_2: b_2 \vdash e: b_3$}
		\RightLabel{\textsc{(T:Case-Product)}}
		\BinaryInfC{$\Gamma_1 \cup \Gamma_2 \vdash \m{case} \; x \; \{ \abra{x_1, x_2} \hookrightarrow e \} : b_3$}
	\end{prooftree}
	\begin{prooftree}
		\AxiomC{$\Gamma_1 \vdash x: L(b)$}
		\AxiomC{$\Gamma_2 \vdash e_0 : b_2$}
		\AxiomC{$\Gamma_2, x_1: b, x_2: L(b) \vdash e_1: b_2$}
		\RightLabel{\textsc{(T:Case-List)}}
		\TrinaryInfC{$\Gamma_1 \cup \Gamma_2 \vdash \m{case} \; x \; \{ [\,] \hookrightarrow e_0 \mid (x_1 :: x_2) \hookrightarrow e_1 \}: b_2 $}
	\end{prooftree}

	\begin{center}
		\AxiomC{}
		\RightLabel{\textsc{(T:Tick)}}
		\UnaryInfC{$\cdot \vdash \m{tick} \; q : \mathbf{1}$}
		\DisplayProof
		\quad
		\AxiomC{$\Gamma_1 \vdash e_1: \tau$}
		\AxiomC{$\Gamma_2, x: \tau \vdash e_2 : b$}
		\RightLabel{\textsc{(T:Let)}}
		\BinaryInfC{$\Gamma_1 \cup \Gamma_2 \vdash \m{let} \; x = e_1 \; \m{in} \; e_2 : b$}
		\DisplayProof
	\end{center}

	\begin{center}
		\AxiomC{$\Gamma, x_1: \tau, x_2: \tau \vdash e: \sigma$}
		\RightLabel{\textsc{(T:Share)}}
		\UnaryInfC{$\Gamma, x: \tau \vdash \m{share} \; x \; \m{as} \; x_1, x_2 \; \m{in} \; e : \sigma$}
		\DisplayProof
		\quad
		\AxiomC{$\Gamma_1 \vdash e: \tau$}
		\AxiomC{$\Gamma_1 \subseteq \Gamma_2$}
		\RightLabel{\textsc{(T:Weak)}}
		\BinaryInfC{$\Gamma_2 \vdash e : \tau$}
		\DisplayProof
	\end{center}

	\end{small}
	\caption{Simple type system of RaML.}
	\label{fig:type system of RaML}
\end{figure}

\subsection{Running Time of RaML}
\label{sec:running time of resource-aware ML}

A \emph{cost semantics} of programming language $L$ is a mapping $\cost{\cdot}: \dbra{L} \rightarrow \mathbb{Q}_{\geq 0}$.
Here, $\dbra{L}$ is the set of all programs in $L$ where inputs are already included in them; that is, programs in $\dbra{L}$ have base types as opposed to arrow types. 
In automatic amortized resource analysis (AARA), cost semantics are defined by specifying how each inference rule of the big-step operational semantics gives rise to computational costs. 
A mapping from inference rules of the operational semantics to $\mathbb{Q}$ (or vectors of $\mathbb{Q}$) is referred to as a \emph{cost metric}. 
Examples of cost metrics include the running time and memory usage. 

The running time is given by a judgment $V \vdash e \Downarrow v \mid n$, where $V$ is an environment (i.e.~a set of pairs of variable symbols and semantic values), $v$ is a semantic value, and $n \in \mathbb{N}$ is the running time of evaluating program $e$ to $v$. 
This judgment is defined in Figure~\ref{fig:inference rules defining the running time of RaML}. 

Environment $V$ may contain redundant variables that are not allowed to appear in the program. 
However, this is not problematic---specifying what variables can appear and what cannot is the job of RaML's type system (Figure~\ref{fig:type system of RaML}), not the job of RaML's cost semantics. 

\begin{figure}[hbt!]
	\begin{small}
	\begin{center}
		\AxiomC{}
		\RightLabel{\textsc{(E:Var)}}
		\UnaryInfC{$V \vdash x \Downarrow V (x) \mid 1$}
		\DisplayProof
		\qquad
		\AxiomC{}
		\RightLabel{\textsc{(E:Unit)}}
		\UnaryInfC{$V \vdash \abra{\,} \Downarrow \abra{\,} \mid 0$}
		\DisplayProof
	\end{center}
	\begin{center}
		\AxiomC{$d \in \{\ell, r\}$}
		\AxiomC{$V \vdash x \Downarrow v \mid 1$}
		\RightLabel{\textsc{(E:Sum)}}
		\BinaryInfC{$V \vdash d \cdot x \Downarrow d \cdot v \mid 2$}
		\DisplayProof
		\quad
		\AxiomC{$V \vdash x_1 \Downarrow v_1 \mid 1$}
		\AxiomC{$V \vdash x_2 \Downarrow v_2 \mid 1$}
		\RightLabel{\textsc{(E:Pair)}}
		\BinaryInfC{$V \vdash \abra{x_1, x_2} \Downarrow \abra{v_1, v_2} \mid 3$}
		\DisplayProof
	\end{center}
	\begin{center}
		\AxiomC{}
		\RightLabel{\textsc{(E:Nil)}}
		\UnaryInfC{$V \vdash [\,] \Downarrow [\,] \mid 0$}
		\DisplayProof
		\qquad
		\AxiomC{$V \vdash x_1 \Downarrow v \mid 1$}
		\AxiomC{$V \vdash x_2 \Downarrow \ell \mid 1$}
		\RightLabel{\textsc{(E:Cons)}}
		\BinaryInfC{$V \vdash (x_1 :: x_2) \Downarrow (v :: \ell) \mid 3$}
		\DisplayProof
	\end{center}
	\begin{prooftree}
		\AxiomC{}
		\RightLabel{\textsc{(E:Fun)}}
		\UnaryInfC{$V \vdash \m{fun} \; f \; x = e \Downarrow \closure{V}{f, x}{e} \mid 1 $}
	\end{prooftree}
	\begin{prooftree}
		\AxiomC{$V (x_1) = \closure{U}{f, x}{e}$}
		\AxiomC{$U, f \mapsto \closure{V}{f, x}{e}, x \mapsto V (x_2) \vdash e \Downarrow v \mid q$}
		\RightLabel{\textsc{(E:App)}}
		\BinaryInfC{$V \vdash x_1 \; x_2 \Downarrow v \mid 1+q$}
	\end{prooftree}
	\begin{prooftree}
		\AxiomC{$d \in \{\ell, r\}$}
		\AxiomC{$V (x) = d \cdot v$}
		\AxiomC{$V, y \mapsto v \vdash e_{d} \Downarrow v_{d} \mid q$}
		\RightLabel{\textsc{(E:Case-Sum)}}
		\TrinaryInfC{$V \vdash \m{case} \; x \; \{\ell \cdot y \hookrightarrow e_{\ell} \mid r \cdot y \hookrightarrow e_{r} \} \Downarrow v_{d} \mid 1+q$}
	\end{prooftree}
	\begin{prooftree}
		\AxiomC{$V (x) = \abra{v_1, v_2}$}
		\AxiomC{$V, x_1 \mapsto v_1, x_2 \mapsto v_2 \vdash e \Downarrow v \mid q$}
		\RightLabel{\textsc{(E:Case-Product)}}
		\BinaryInfC{$V \vdash \m{case} \; x \; \{\abra{x_1, x_2} \hookrightarrow e\} \Downarrow v \mid 1+q$}
	\end{prooftree}
	\begin{prooftree}
		\AxiomC{$V (x) = [\,]$}
		\AxiomC{$V \vdash e_0 \Downarrow v \mid q$}
		\RightLabel{\textsc{(E:Case-Nil)}}
		\BinaryInfC{$V \vdash \m{case} \; x \; \{[\,] \hookrightarrow e_0 \mid (x_1 :: x_2) \hookrightarrow e_1\} \Downarrow v \mid 1+q$}
	\end{prooftree}
	\begin{prooftree}
		\AxiomC{$V (x) = v :: \ell$}
		\AxiomC{$V, x_1 \mapsto v, x_2 \mapsto \ell \vdash e_1 \Downarrow v' \mid q$}
		\RightLabel{\textsc{(E:Case-Cons)}}
		\BinaryInfC{$V \vdash \m{case} \; x \; \{[\,] \hookrightarrow e_0 \mid (x_1 :: x_2) \hookrightarrow e_1\} \Downarrow v' \mid 1+q$}
	\end{prooftree}

	\begin{prooftree}
		\AxiomC{$V \vdash e_1 \Downarrow v_1 \mid q$}
		\AxiomC{$V, x \mapsto v_1 \vdash e_2 \Downarrow v_2 \mid p$}
		\RightLabel{\textsc{(E:Let)}}
		\BinaryInfC{$V \vdash \m{let} \; x = e_1 \; \m{in} \; e_2 \Downarrow v_2 \mid 1 + q + p $}
	\end{prooftree}
	\begin{prooftree}
		\AxiomC{$V (x) = v$}
		\AxiomC{$V, x_1 \mapsto v, x_2 \mapsto v \vdash e \Downarrow v' \mid q$}
		\RightLabel{\textsc{(E:Share)}}
		\BinaryInfC{$V \vdash \m{share} \; x \; \m{as} \; x_1, x_2 \; \m{in} \; e \Downarrow v' \mid q$}
	\end{prooftree}
	\end{small}
	\caption{Inference rules defining the running time of RaML.}
	\label{fig:inference rules defining the running time of RaML}
\end{figure}

In \textsc{(E:Fun)}, $\closure{V}{f, x}{e}$ is a function closure where $V$ is an environment and $e$ is code that may mention $f$ and $x$. 
The notation $f, x. e$ means that $f, x$ are bound in $e$. 

In \textsc{(E:Unit)} and \textsc{(E:Nil)}, the running time is zero. 
One might argue that these cases could have positive running time in practice. 
If we are to assign positive running time to \textsc{(E:Unit)} and \textsc{(E:Nil)}, we will need to revise the type systems of unvarite and multivariate AARA accordingly. 
However, it will not fundamentally affect the theorems and their proofs in this article. 
For instance, consider a primitive recursion $\m{rec} \; x \; \{[\,] \hookrightarrow e_0 \mid (y::\mi{ys}) \; \m{with} \; z \hookrightarrow e_1 \}$, which can be encoded using (i) general recursion and (ii) lists' pattern matching. 
Even if the stepping function $e_1$ has zero running time (e.g.~due to \textsc{(E:Nil)}), the total running time of the whole primitive recursion is strictly positive and is proportional to the length of an input list, thanks to \textsc{(E:Case-Cons)}. 
Consequently, even if \textsc{(E:Unit)} and \textsc{(E:Nil)} have zero running time as they do now, the running time of any primitive recursion is non-zero. 

We do not have a rule for $\m{tick}$, because it is only necessary for the tick metric. 


\section{Univariate Polynomial AARA}
\label{sec:univariate polynomial AARA}

\subsection{Type system}
\label{sec:type system of univariate AARA}

This section presents a type system for univariate polynomial AARA. 
The content and presentation style of this section are attributed to Hoffmann's PhD thesis \cite{Hoffmann2011}. 

The typing judgment for univariate polynomial AARA has the form 
\begin{equation} \label{eq:typing judgment of univariate AARA}
	\Gamma_{\text{anno}}; q \vdash e: A,
\end{equation}
where $q$ is a non-negative rational number and $\Gamma_{\text{anno}}$ is a resource-annotated typing context (i.e.~a set of pairs of variable symbols and their resource-annotated types). 
$A$ has the form $\abra{b, p}$, where $b$ is a resource-annotated base type and $p \in \mathbb{Q}_{\geq 0}$ (Section~\ref{sec:basics of univariate AARA}).  

\paragraph*{Syntax-Directed Rules}
The syntax-directed rules for univariate polynomial AARA are presented in Figure~\ref{fig:syntax-directed rules of univariate AARA}. 

\begin{figure}[hbt!]
	\begin{small}
	\begin{center}
		\AxiomC{}
		\RightLabel{\textsc{(U:Var)}}
		\UnaryInfC{$x: \tau; 1 \vdash x: \abra{\tau, 0}$}
		\DisplayProof
		\qquad
		\AxiomC{}
		\RightLabel{\textsc{(U:Unit)}}
		\UnaryInfC{$\cdot; 0 \vdash \abra{\,} : \abra{\mathbf{1}, 0}$}
			\DisplayProof
	\end{center}
	\begin{center}
		\AxiomC{$\Gamma; 1 \vdash x: \abra{b_1, 0}$}
		\RightLabel{\textsc{(U:SumL)}}
		\UnaryInfC{$\Gamma; 2 \vdash \ell \cdot x : \abra{b_1 + b_2, 0}$}
		\DisplayProof
		\qquad
		\AxiomC{$\Gamma; 1 \vdash x: b_2$}
		\RightLabel{\textsc{(U:SumR)}}
		\UnaryInfC{$\Gamma; 2 \vdash r \cdot x : \abra{b_1 + b_2, 0}$}
		\DisplayProof
	\end{center}

	\begin{center}
		\AxiomC{}
		\RightLabel{\textsc{(U:Pair)}}
		\UnaryInfC{$x_1: b_1, x_2: b_2; 3 \vdash \abra{x_1, x_2} : \abra{b_1 \times b_2, 0}$}
		\DisplayProof
		\quad
		\AxiomC{}
		\RightLabel{\textsc{(U:Nil)}}
		\UnaryInfC{$\cdot; 0 \vdash [\,] : \abra{L^{\vec{q}} (b), 0}$}
		\DisplayProof
	\end{center}
	\begin{prooftree}
		\AxiomC{}
		\RightLabel{\textsc{(U:Cons)}}
		\UnaryInfC{$x_1: b, x_2: L^{\lhd (\vec{q})} (b); 3+q_1 \vdash x_1 :: x_2 : \abra{L^{\vec{q}} (b), 0}$}
	\end{prooftree}

	\begin{prooftree}
		\AxiomC{$\Gamma = \abs{\Gamma}$}
		\AxiomC{$\forall (\abra{b, q} \to B_2) \in \mathcal{T}. \abs{\Gamma}, f: \mathcal{T}, x: b; q \vdash e: B_2$}
		\RightLabel{\textsc{(U:Fun)}}
		\BinaryInfC{$\Gamma; 1 \vdash \m{fun} \; f \; x = e : \abra{\mathcal{T}, 0}$}
	\end{prooftree}

	\begin{prooftree}
		\AxiomC{$(B_1 \to B_2) \in \mathcal{T}$}
		\AxiomC{$B_1 = \abra{b, q}$}
		\RightLabel{\textsc{(U:App)}}
		\BinaryInfC{$f : \mathcal{T}, x : b; 1+q \vdash f \; x : B_2$}
	\end{prooftree}

	\begin{prooftree}
		\AxiomC{$\Gamma, y: b_1; q \vdash e_{\ell}: B$}
		\AxiomC{$\Gamma, y: b_2; q \vdash e_{r}: B$}
		\RightLabel{\textsc{(U:Case-Sum)}}
		\BinaryInfC{$x: b_1 + b_2, \Gamma; 1+q \vdash \m{case} \; x \; \{\ell \cdot y \hookrightarrow e_{\ell} \mid r \cdot y \hookrightarrow e_{r} \} : B$}
	\end{prooftree}
	\begin{prooftree}
		\AxiomC{$\Gamma, x_1: b_1, x_2: b_2; q \vdash e : B$}
		\RightLabel{\textsc{(U:Case-Product)}}
		\UnaryInfC{$x: b_1 \times b_2, \Gamma; 1+q \vdash \m{case} \; x \; \{\abra{x_1, x_2} \hookrightarrow e\} : B$}
	\end{prooftree}
	\begin{prooftree}
		\AxiomC{$\Gamma; q \vdash e_0: B$}
		\AxiomC{$\Gamma, x_1: b, x_2: L^{\lhd (\vec{p})} (b); q + p_1 \vdash e_1 : B$}
		\RightLabel{\textsc{(U:Case-List)}}
		\BinaryInfC{$x: L^{\vec{p}} (b), \Gamma; 1+q \vdash \m{case} \; x \; \{[\,] \hookrightarrow e_0 \mid x_1 :: x_2 \hookrightarrow e_1 \} : B$}
	\end{prooftree}
	\begin{prooftree}
		\AxiomC{$\Gamma_1; q \vdash e_1 : \abra{\tau, p}$}
		\AxiomC{$\Gamma_2, x: \tau; p \vdash e_2 : B$}
		\RightLabel{\textsc{(U:Let)}}
		\BinaryInfC{$\Gamma_1 \cup \Gamma_2; 1+ q \vdash \m{let} \; x = e_1 \; \m{in} \; e_2 : B$}
	\end{prooftree}
	\begin{prooftree}
		\AxiomC{$\tau \curlyveedownarrow (\tau_1, \tau_2)$}
		\AxiomC{$\Gamma, x_1: \tau_1, x_2: \tau_2; q \vdash e: \abra{\tau, q}$}
		\RightLabel{\textsc{(U:Share)}}
		\BinaryInfC{$\Gamma, x: \tau; q \vdash \m{share} \; x \; \m{as} \; x_1, x_2 \; \m{in} \; e : \abra{\tau, q}$}
	\end{prooftree}
	\end{small}
	\caption{Syntax-directed rules of univariate AARA.}
	\label{fig:syntax-directed rules of univariate AARA}
\end{figure}

In \textsc{(U:Cons)}, $q_1$ is not a fresh variable---it denotes the first component of vector $\vec{q}$. 

In \textsc{(U:Fun)} and \textsc{(U:App)}, $\mathcal{T}$ is a (finite or infinite) set of resource-annotated arrow types of the form $B_1 \to B_2$, where $b_{i} = \abra{b_{i}, q_{i}}$. 
Each function $\m{fun} \; f \; x = e$ is associated with such set $\mathcal{T}$. 
For instance, the identity function $\mi{id} : L (\mathbf{1}) \to L (\mathbf{1})$ always consumes one unit of potential. 
Therefore, the largest possible $\mathcal{T}$ for $\mi{id}$ is 
\begin{equation*}
	\begin{split}
	\mathcal{T} = \{\abra{L^{\vec{p}} (\mathbf{1}), q + c} \to \abra{L^{\vec{p}} (\mathbf{1}), q} & \mid q \in \mathbb{Q}_{\geq 0}, c \in \mathbb{Q}_{\geq 1}, \\
	& \vec{p} \text{ is any valid univariate polynomial annotation} \}. 
	\end{split}
\end{equation*}
Of course, $\mathcal{T}$ need not be as large as it can---it only needs to be self-sufficient in the sense that every $(B_1 \to B_2) \in \mathcal{T}$ can be derived using $\mathcal{T}$ itself (See \textsc{(U:Fun)}). 
Since $\mathcal{T}$ is a set of types rather than a single type, the typing judgment $f: \mathcal{T}$ does not conform to \eqref{eq:typing judgment of univariate AARA}. 
Nonetheless, we will keep writing $f: \mathcal{T}$ because it conveniently conveys that we are allowed to assign multiple resource annotations to $f$. 

The use of $f: B_1 \to B_2$ to derive a different arrow type $f: C_1 \to C_2$ of $f$ itself, where $f$ is recursively defined (i.e.~$f$ mentions itself), is called \emph{resource-polymorphic recursion} \cite{Hoffmann2010}. 
This phenomenon arises in many recursive functions. 
An example is provided in \cite{Hoffmann2010}.

It is reasonable to wonder how we can possibly ``deterministically'' infer resource annotations, given that $\mathcal{T}$ is allowed to be infinite. 
For instance, we may have $\mathcal{T} = \{B_1 \to C_1, B_2 \to C_2, B_3 \to C_3, \ldots\}$, where in order to type $f: B_{i} \to C_{i}$ for each $i$, we need to use $f: B_{i+1} \to C_{i+1}$. 
Consequently, resource-polymorphic recursion may induce an infinite chain of \textsc{(U:Fun)}. 
This problem will be discussed in detail in the following section. 

The first premise of \textsc{(U:Fun)} states $\Gamma = \abs{\Gamma}$. 
Informally, $\abs{\cdot}$ deletes all resource annotations inside resource-annotated base types (but does nothing to resource-annotated arrow types). 
Formally, $\abs{\tau}$, where $\tau$ is a resource-annotated simple type, is defined as follows:
\begin{align*}
	\abs{\mathbf{1}} & := \mathbf{1} & \abs{L^{\vec{q}} (b)} & := L^{0} (\abs{b}) \\
	\abs{b_1 + b_2} & := \abs{b_1} + \abs{b_2} & 	\abs{\abra{b_1, q_1} \to \abra{b_2, q_2}} & := \abra{b_1, q_1} \to \abra{b_2, q_2} \\
	\abs{b_1 \times b_2} & := \abs{b_1} \times \abs{b_2}. 
\end{align*}
This can be generalized to typing contexts as follows: if $\Gamma = {x_1: \tau_1, \ldots, x_n : \tau_n}$, we have $\abs{\Gamma} = {x_1: \abs{\tau_1}, \ldots, x_n : \abs{\tau_n}}$. 

The reason why we require $\Gamma = \abs{\Gamma}$ in the first premise of \textsc{(U:Fun)} is that we want function closures (i.e.~semantic values $\closure{V}{f, x}{e}$) to store zero potential. 
If a function $f$ is defined in the presence of a positive-potential base-type variable, it means $f$ consumes potential every time it is invoked. 
However, in AARA, we do not keep track of how many times $f$ is invoked---instead, we let it be invoked freely. 
Of course, we could modify AARA such that it computes an upper bound on the number of times $f$ is invoked. 
However, for simplicity, we do not adopt this approach. 

In \textsc{(U:Share)}, $\tau \curlyveedownarrow (\tau_1, \tau_2)$ means the resource-annotated type $\tau$ can be split into $\tau_1$ and $\tau_2$. 
This is defined in Figure~\ref{fig:definition of the sharing relation of types}. 
When we split a variable $x: B_1 \rightarrow B_2$ into $x_1$ and $x_2$ (i.e.~$\m{share} \; x \; \m{as} \; x_1, x_2$), $x_1$ and $x_2$ are allowed to have distinct annotations as long as these annotations can be legally derived using \textsc{(U:Fun)}. 
Thus, it is possible to use a different resource annotation for a function $f$ each time $f$ is invoked. 
In Figure~\ref{fig:definition of the sharing relation of types}, $\mathcal{T}_{f}$ denotes a set of resource-annotated arrow types associated with function $f$ that we consider now. 

\begin{figure}[hbt!]
	\begin{small}
	\begin{center}
		\AxiomC{}
		\UnaryInfC{$\mathbf{1} \curlyveedownarrow (\mathbf{1}, \mathbf{1})$}
		\DisplayProof
		\quad
		\AxiomC{$b_1 \curlyveedownarrow (b_{1,1}, b_{1, 2})$}
		\AxiomC{$b_2 \curlyveedownarrow (b_{2,1}, b_{2, 2})$}
		\BinaryInfC{$(b_1 + b_2) \curlyveedownarrow (b_{1,1} + b_{2, 1}, b_{1,2} + b_{2, 2})$}
		\DisplayProof
		\quad
		\AxiomC{$b_1 \curlyveedownarrow (b_{1,1}, b_{1, 2})$}
		\AxiomC{$b_2 \curlyveedownarrow (b_{2,1}, b_{2, 2})$}
		\BinaryInfC{$(b_1 \times b_2) \curlyveedownarrow (b_{1,1} \times b_{2, 1}, b_{1,2} \times b_{2, 2})$}
		\DisplayProof
	\end{center}
	
	\begin{center}
		\AxiomC{$\vec{p} = \vec{p_1} + \vec{p_2}$}
		\AxiomC{$b \curlyveedownarrow (b_1, b_2)$}
		\BinaryInfC{$L^{\vec{p}} (b) \curlyveedownarrow (L^{\vec{p_1}} (b_1), L^{\vec{p_2}} (b_2))$}
		\DisplayProof
		\quad
		\AxiomC{$(B_1 \to B_2), (C_1 \to C_2), (D_1 \to D_2) \in \mathcal{T}_{f}$}
		\UnaryInfC{$(B_1 \to B_2) \curlyveedownarrow (C_1 \to C_1, D_1 \to D_2)$}
		\DisplayProof
	\end{center}
	\end{small}
	\caption{Definition of the sharing relation $\tau \curlyveedownarrow (\tau_1, \tau_2)$.}
	\label{fig:definition of the sharing relation of types}
\end{figure}

In \textsc{(U:Cons)} and \textsc{(U:Case-List)}, we use the univariate additive shift operator (denoted by $\lhd$) that splits the potential of a list between the head element and the tail. 
This is formally defined below. 

\begin{definition}[Additive shift of potential vectors]
\label{def:additive shift of potential vectors}
Given a potential vector $\vec{q} = (q_1, q_2, \cdots, q_{k})$, its additive shift is defined as 
\begin{equation*}
	\lhd (\vec{q}) := (q_1 + q_2, q_2 + q_3, \cdots, q_{k-1} + q_{k}, q_k). 
\end{equation*}
\end{definition}

\paragraph*{Structural Rules}
The structural rules of univariate AARA are displayed in Figure~\ref{fig:structural rules of univariate AARA}. 
Also, the subtyping relation $\tau_1 <: \tau_2$ is defined in Figure~\ref{fig:definition of the subtyping relation for univariate AARA}.

\begin{figure}[hbt!]
	\begin{small}
	\begin{center}
		\AxiomC{$\Gamma; q \vdash e: \abra{\tau, p}$}
		\AxiomC{$\tau <: \tau'$}
		\RightLabel{\textsc{(U:Sub)}}
		\BinaryInfC{$\Gamma; q \vdash e: \abra{\tau'. p}$}
		\DisplayProof
		\quad
		\AxiomC{$\Gamma, x: \tau; q \vdash e: B$}
		\AxiomC{$\tau' <: \tau$}
		\RightLabel{\textsc{(U:Sup)}}
		\BinaryInfC{$\Gamma, x: \tau'; q \vdash e: B$}
		\DisplayProof
	\end{center}
	\begin{center}
		\AxiomC{$\Gamma_1; q \vdash e: B$}
		\RightLabel{\textsc{(U:Weak)}}
		\UnaryInfC{$\Gamma_1, \Gamma_2; q \vdash e: B $}
		\DisplayProof
		\;
		\AxiomC{$\Gamma; p_1 \vdash e: \abra{\tau, p_2}$}
		\AxiomC{$q_1 \geq p_1$}
		\AxiomC{$q_1 - q_2 \geq p_1 - p_2$}
		\RightLabel{\textsc{(U:Relax)}}
		\TrinaryInfC{$\Gamma; q_1 \vdash e: \abra{\tau, q_2}$}
		\DisplayProof
	\end{center}
	\end{small}
	\caption{Structural rules of univariate AARA. }
	\label{fig:structural rules of univariate AARA}
\end{figure}

\begin{figure}[hbt!]
	\begin{small}
	\begin{center}
		\AxiomC{}
		\UnaryInfC{$\mathbf{1} <: \mathbf{1}$}
		\DisplayProof
		\quad
		\AxiomC{$b_1 <: c_1$}
		\AxiomC{$b_2 <: c_2$}
		\BinaryInfC{$b_1 + b_2 <: c_1 + c_2$}
		\DisplayProof
		\quad
		\AxiomC{$b_1 <: c_1$}
		\AxiomC{$b_2 <: c_2$}
		\BinaryInfC{$b_1 \times b_2 <: c_1 \times c_2$}
		\DisplayProof
	\end{center}
	\begin{center}
		\AxiomC{$b_1 <: b_2$}
		\AxiomC{$\vec{p_1} \geq \vec{p_2}$}
		\BinaryInfC{$L^{\vec{p_1}} (b_1) <: L^{\vec{p_2}} (b_2)$}
		\DisplayProof
		\quad
		\AxiomC{$C_1 <: B_1$}
		\AxiomC{$B_2 <: C_2$}
		\BinaryInfC{$B_1 \to B_2 <: C_1 \to C_2$}
		\DisplayProof
		\quad
		\AxiomC{$b_1 <: b_2$}
		\AxiomC{$q_1 \geq q_2$}
		\BinaryInfC{$\abra{b_1, q_1} <: (b_2, q_2)$}
		\DisplayProof
	\end{center}
	\end{small}
	\caption{Definition of the subtyping relation $\tau_1 <: \tau_2$ in univariate AARA.}
	\label{fig:definition of the subtyping relation for univariate AARA}
\end{figure}

\subsection{Practical Type Inference}
\label{sec:practical type inference for univariate AARA}

As mentioned above, resource-polymorphic recursion may induce an infinite chain of \textsc{(U:Fun)}, causing type checking/inference to continue indefinitely. 
To work around this problem, the type inference algorithm of AARA in \cite{Hoffmann2010,Hoffmann2017} uses the following heuristic. 
Suppose we are to derive a resource-annotated arrow type $f: B_1 \to B_2$, where $f$ is recursively defined. 
To derive $f: B_1 \to B_2$, suppose we need to use a distinct type $f: C_1 \to C_2$. 
Now consider the type $(C_1 - B_1) \to (C_2 - B_2)$, where $C_{i} - B_{i}$ denotes annotation-wise subtraction. 
For instance, if $C_i \equiv L^{\vec{p}_{C}} (L^{\vec{q}_{C}} (\mathbf{1}))$ and $B_{i} \equiv L^{\vec{p}_{B}} (L^{\vec{q}_{B}} (\mathbf{1}))$, then $C_{i} - B_{i} \equiv L^{\vec{p}_{C} - \vec{p}_{B}} (L^{\vec{q}_{C} - \vec{q}_{B}} (\mathbf{1}))$, where $\vec{p}, \vec{q}$ are vectors of $\mathbb{Q}_{\geq 0}$. 

AARA's type inference algorithm requires the following:
\begin{itemize}
	\item $f: (C_1 - B_1) \to (C_2 - B_2)$ can be derived under the \emph{cost-free} metric;
	\item $(C_1 - B_1) \to (C_2 - B_2)$ has a strictly lower degree than $B_1 \to B_2$ and $C_1 \to C_2$. 
\end{itemize}
In the first condition, the cost-free metric means any computational cost (e.g.~function application, pattern matching, and variable lookups) is zero. 
The cost-free metric (as opposed to the cost metric of the running time) is used to derive $f: (C_1 - B_1) \to (C_2 - B_2)$ because the derivations of $B_1 \to B_2$ and $C_1 \to C_2$ already consider the actual computational cost (i.e.~running time). 

To paraphrase the second condition, $B_1 \to B_2$ and $C_1 \to C_2$ must share the same coefficients for the maximum degree. 
For example, consider $B_1 \equiv (6, 6, 1)$, which represents $6 \cdot \binom{n}{3} + 6 \cdot \binom{n}{2} + 1 \cdot \binom{n}{1} = n^{3}$. 
The degree of $B_1$ is 3 because $B_1$ is a vector of size (that is, $B_1$ represents a cubic polynomial function). 
Therefore, $C_1$ must have the shape $C_1 \equiv (6, p_2, p_3)$ for some rational numbers $p_2, p_3$ so that $C_1 - B_1$ has a strictly lower degree than $B_1$ and $C_1$.
Thanks to this restriction, AARA's inference algorithm is guaranteed to terminate. 
However, it is not complete with respect to the type system in Figure~\ref{fig:syntax-directed rules of univariate AARA}: some polynomial-time recursive functions require an infinite set $\mathcal{T}$ in \textsc{(U:Fun)}. 

One might wonder whether the derivation of $f: (C_1 - B_1) \to (C_2 - B_2)$ under the cost-free metric requires resource-polymorphic recursion; that is, whether we will have a cascade/chain of resource-polymorphic recursion. 
Interestingly enough, according to \cite{Hoffmann2011}, the derivations of cost-free types (e.g.~$f: (C_1 - B_1) \to (C_2 - B_2)$) in univariate AARA only need resource-monomorphic recursion; i.e.~$\mathcal{T}$ in \textsc{(U:Fun)} is a singleton set. 
In other words,  in univariate AARA, the type inference algorithm only applies resource-polymorphic recursion once, which will then be followed by resource-monomorphic recursion. 
There is no justification for using resource-monomorphic recursion instead of resource-polymorphic recursion in the derivations of cost-free types. 
This approach probably works well for all programs that (i) are tested in \cite{Hoffmann2011} and (ii) are typable in univariate AARA even if we use a cascade of resource-polymorphic recursion. 
By contrast, in multivariate AARA, we sometimes have a cascade of resource-polymorphic recursion under the cost-free metric. 

Finally, it is necessary to clarify what we mean by typability in this article. 
In univariate (and multivariate) AARA, we have two notions of typability:
\begin{enumerate}
	\item A proof based on Figure~\ref{fig:syntax-directed rules of univariate AARA}, where $\mathcal{T}$ may be an infinite set, exists;
	\item The type inference algorithm of AARA actually terminates and produces a polynomial cost bound, provided that the user specifies a sufficiently high degree.
\end{enumerate}
In this article, we will use the latter, stronger notion of typability. 
For example, Theorem~\ref{theorem:inherently polynomial time implies the existence of a multivariate annotation} and Theorem~\ref{theorem:inherently polynomial time implies the existence of a multivariate annotation where a user-specified amount of potential is available in the output} use the stronger notion of typability: after each invocation of \textsc{(U:Fun)} (actually, the multivariate version \textsc{(M:Fun)}), the degree strictly decreases. 

\subsection{Examples}
\label{sec:implementation of append and quicksort presented as examples of typing judgments of univariate AARA}

Section~\ref{sec:basics of univariate AARA} presents univariate resource annotations of $\mi{append}$ and $\mi{quicksort}$. 
They are implemented as follows: 
\begin{align*}
	\mi{append} & : \abra{L (\mathbf{1}), L (\mathbf{1})} \rightarrow L (\mathbf{1}) \\
	\mi{append} & := \m{fun} \; f \; \abra{\ell_1, \ell_2} = \m{case} \; \ell_1 \; \{[\,] \hookrightarrow \ell_2 \mid (y::\mi{ys}) \hookrightarrow y :: (f \; \abra{\mi{ys}, \ell_2}) \} \\
	\mi{quicksort} & : L (b) \rightarrow L (b) \\
	\mi{quicksort} & := \m{fun} \; f \; \ell = \m{case} \; \ell \; \{[\,] \hookrightarrow [\,] \mid (y::\mi{ys}) \hookrightarrow e_1 \} \\
	e_1 & \equiv \m{let} \; \abra{\ell_1, \ell_2} = \mi{split} \; \abra{y, \mi{ys}}, x_1 = f \; \ell_1, x_2 = f \; \ell_2 \; \m{in} \; \mi{append} \; \abra{x_1, y :: x_2}. 
\end{align*}
For readability, we have abbreviate $\m{let} \; x_1 = e_1 \; \m{in} \; \m{let} \; x_2 = e_2 \; \m{in} \; e_3$ to $\m{let} \; x_1 = e_1, x_2 = e_2 \; \m{in} \; e_3$. 
Also, we use the notation $\m{let} \; \abra{x_1, x_2} = e_1$ as syntactic sugar for $\m{let} \; x = e_1 \; \m{in} \; \m{case} \; e_1 \; \{\abra{x_1, x_2} \hookrightarrow \cdots \}$. 

Here, $\mi{split} \; \abra{y, \mi{ys}}$ classifies each element in the list $\mi{ys}$ according to whether it is smaller than $y$ or not. 
$\mi{split}$ is defined as
\begin{align*}
	\mi{split} & : (b \times L (b)) \rightarrow (L (b) \times L (b)) \\
	\mi{split} & := \m{fun} \; f \; \abra{x, \ell} = \m{case} \; \ell \; \{[\,] \hookrightarrow \abra{[\,], [\,]} \mid (y::\mi{ys}) \hookrightarrow e_1 \} \\
	e_1 & \equiv \m{let} \; \abra{\ell_1, \ell_2} = f \; \abra{x, \mi{ys}} \; \m{in} \; \m{if} \; y < x \; \m{then} \; \abra{y::\ell_1, \ell_2} \; \m{else} \; \abra{\ell_1, y::\ell_2}. 
\end{align*}


\section{Multivariate Polynomial AARA}

This section describes multivariate polynomial AARA. 
As before, the content and presentation style of this section are attributed to Hoffmann's PhD thesis \cite{Hoffmann2011}. 

\subsection{Notation}

We will introduce the projection and extension operators on multivariate annotations. 
To define projection, suppose we are given a base-type typing context $\Gamma_1 \cup \Gamma_2$. 
If $j \in \mathcal{I} (\Gamma_2)$, the projection $\pi^{\Gamma_1}_{j} (Q)$ is defined as
\begin{equation} \label{eq:definition of projection of a multivariate resource annotation}
	\pi^{\Gamma_1}_{j} (Q) (i) := Q (i, j). 
\end{equation}

Next, to define the extension operator, let $Q$ be a multivariate annotation over $\Gamma_1$. 
Given $r \in \mathcal{I} (\Gamma_2)$, the extension $\eta^{\Gamma_1 \cup \Gamma_2}_{r} (Q)$ is defined as
\begin{equation} \label{eq:definition of extension of a multivariate resource annotation}
	\eta^{\Gamma_1 \cup \Gamma_2}_{r} (Q) (i, j) := 
	\begin{cases}
		Q (i) & \text{if } j = r; \\
		0 & \text{otherwise}. 
	\end{cases}
\end{equation}

Denoted by $\lhd$, additive shift specifies how to split the potential of a list between the head element and the tail. 
We now explain additive shift of multivariate AARA. 
Suppose that we are given a base-type context $\Gamma_1 = \Gamma \cup \{\ell: L (b)\}$, where $\ell$ is the last element without loss of generality. 
Let $Q$ be a multivariate annotation over $\Gamma \cup \{\ell: L (b)\}$. 
$\ell: L (b)$ is split into the head element $x: b$ and the tail $\mi{xs}: L (b)$, yielding a new typing context $\Gamma_2 = \Gamma \cup \{x: b, \mi{xs}: L (b)\}$. 
The resource annotation for the new context is given by the (multivariate) additive shift of $Q$ with respect to $\ell: L (b)$. 
This is denoted by $\lhd_{\ell} Q$ and is defined as 
\begin{equation} \label{eq:definition of multivariate additive shift}
	\lhd (Q) (i, j, k) := 
	\begin{cases}
		Q (i, 0_{b}::k) + Q (i, k) & \text{if } j = 0_{b}; \\
		Q (i, j::k) & \text{otherwise}. 
	\end{cases}
\end{equation}
Here, $i$, $j$, and $k$ are indexes of $\Gamma$, $x: b$, and $\mi{xs}: L (b)$, respectively. 
Recall that $j = 0_{b}$ refers to the base polynomial $\lam v. 1$ for type $b$. 
The notation $j::k$ denotes the concatenation of $j$, which is an index for type $b$, to $k$, which is an index for type $L (b)$. 

Lastly, for variable sharing, it is sufficient to know that for any multivariate resource annotation $Q$ for a base-type typing context $\Gamma \cup \{x_1: b, x_2: b\}$, we have a way to produce $P$ for the unified typing context $\Gamma \cup \{x: b\}$ such that no potential from $Q$ is lost. 
The relation between $Q$ and $P$ will be denoted by $P \curlyveedownarrow_{x_1, x_2} Q$. 
The details of variable sharing are deferred to \cite{Hoffmann2011,Hoffmann2012}. 

\subsection{Type System}
\label{sec:type system of multivariate AARA}

We will consider both (i) the type system with the cost metric being the running time and (ii) the type system under the cost-free metric. 

A typing judgment in multivariate AARA has the form
\begin{equation*}
	\Sigma_{\text{anno}}; \Gamma; Q \vdash e: \abra{\tau, P}. 
\end{equation*}
$\Sigma_{\text{anno}}$ is a resource-annotated arrow-type typing context, where each variable binding is of the form $f: \abra{b_1, Q_1} \rightarrow \abra{b_2, Q_2}$. 
$\Gamma$ is a base-type typing context without resource annotations, and $Q$ is a multivariate resource annotation/polynomial for $\Gamma$. 
Likewise, $P$ is a multivariate resource annotation for $\tau$. 
When $\tau$ is a resource-annotated arrow type, $P$ will just be a constant. 

\paragraph*{Syntax-Directed Rules}

Syntax-directed rules of multivariate AARA are presented in Figure~\ref{fig:syntax-directed rules of multivariate AARA}. 
Anything enclosed by square brackets should be ignored if we use the cost-free metric. 
For instance, in \textsc{(M:Var-Base)}, the resource annotation of the typing context in the conclusion is $Q$ if we use the cost-free metric. 
Otherwise, if the cost metric is the running time, the resource annotation should be $1+Q$. 

\begin{figure}[hbt!]
	\begin{small}
	\begin{center}
		\AxiomC{$Q (0_{b}) = 0$}
		\RightLabel{\textsc{(M:Var-Base)}}
		\UnaryInfC{$\cdot; x: b; [1+] Q \vdash x: \abra{b, Q}$}
		\DisplayProof
		\quad
		\AxiomC{$\tau \equiv B_1 \to B_2$}
		\RightLabel{\textsc{(M:Var-Arrow)}}
		\UnaryInfC{$f: \tau; \cdot; [1] \vdash f: \abra{\tau, 0}$}
		\DisplayProof
	\end{center}

	\begin{center}
		\AxiomC{$Q (0_{b}) = 0$}
		\RightLabel{\textsc{(M:SumL)}}
		\UnaryInfC{$\cdot; x: b_1; [1+]Q \vdash \ell \cdot x : \abra{b_1 + b_2, Q}$}
		\DisplayProof
		\;
		\AxiomC{$Q (0_{b}) = 0$}
		\RightLabel{\textsc{(M:SumR)}}
		\UnaryInfC{$\cdot; x: b_2; [1+]Q \vdash r \cdot x : \abra{b_1 + b_2, Q}$}
		\DisplayProof
	\end{center}

	\begin{center}
		\AxiomC{}
		\RightLabel{\textsc{(M:Unit)}}
		\UnaryInfC{$\cdot; \cdot; 0 \vdash \abra{\,} : \abra{\mathbf{1}, 0}$}
		\DisplayProof
		\quad
		\AxiomC{$Q (0_{b}) = 0$}
		\RightLabel{\textsc{(M:Pair)}}
		\UnaryInfC{$\cdot; x_1: b_1, x_2: b_2; [3+]Q \vdash \abra{x_1, x_2}: \abra{b_1 \times b_2, Q}$}
		\DisplayProof
	\end{center}

	\begin{center}
		\AxiomC{$Q (0_{b}) = 0$}
		\RightLabel{\textsc{(M:Nil)}}
		\UnaryInfC{$\cdot; \cdot; 0 \vdash [\,] : \abra{L (b), Q}$}
		\DisplayProof
		\quad
		\AxiomC{$Q = \lhd_{x_1 :: x_2} (Q')$}
		\RightLabel{\textsc{(M:Cons)}}
		\UnaryInfC{$\cdot; x_1: b, x_2: L (b); [3+]Q \vdash x_1 :: x_2 : \abra{L (b), Q'}$}
		\DisplayProof
	\end{center}

	\begin{prooftree}
		\AxiomC{$\Gamma = \abs{\Gamma}$}
		\AxiomC{$\forall (\abra{b, Q} \to B_2) \in \mathcal{T}. \Sigma, f: \mathcal{T}; \Gamma, x: b; \eta_{0}^{\Gamma, x: b} (Q) \vdash e: B_2$}
		\RightLabel{\textsc{(M:Fun)}}
		\BinaryInfC{$\Sigma; \Gamma; 1 \vdash \m{fun} \; f \; x = e : \abra{\mathcal{T}, 0}$}
	\end{prooftree}

	\begin{prooftree}
		\AxiomC{$(B_1 \to B_2) \in \mathcal{T}$}
		\AxiomC{$B_1 = \abra{b, Q}$}
		\RightLabel{\textsc{(M:App)}}
		\BinaryInfC{$f: \mathcal{T}; x: b; [1+]Q \vdash f \; x : B_2$}
	\end{prooftree}

	\begin{prooftree}
		\AxiomC{$\Sigma; y: b_1, \Gamma; Q [(\ell \cdot y) / x] \vdash e_{\ell}: B$}
		\AxiomC{$\Sigma; y: b_2, \Gamma; Q [(r \cdot y) / x] \vdash e_{r}: B$}
		\RightLabel{\textsc{(M:Case-Sum)}}
		\BinaryInfC{$\Sigma; x: b_1 + b_2, \Gamma; [1+]Q \vdash \m{case} \; x \; \{\ell \cdot y \hookrightarrow e_{\ell} \mid r \cdot y \hookrightarrow e_{r} : B \}$}
	\end{prooftree}
	\begin{prooftree}
		\AxiomC{$\Sigma; x_1: b_1, x_2: b_2, \Gamma; Q \vdash e: B$}
		\RightLabel{\textsc{(M:Case-Prod)}}
		\UnaryInfC{$\Sigma; x: b_1 \times b_2, \Gamma; [1+]Q \vdash \m{case} \; x \; \{\abra{x_1, x_2} \hookrightarrow e \} : B$}
	\end{prooftree}
	\begin{prooftree}
		\AxiomC{$\Sigma; \Gamma; \pi^{\Gamma}_{0} (Q) \vdash e_0 : B$}
		\AxiomC{$\Sigma; \Gamma, y: b, \mi{ys}: L (b); \lhd_{x} (Q) \vdash e_1: B$}
		\RightLabel{\textsc{(M:Case-List)}}
		\BinaryInfC{$\Sigma; x: L (b), \Gamma; [1+]Q \vdash \m{case} \; x \; \{[\,] \hookrightarrow e_0 \mid (y::\mi{ys}) \hookrightarrow e_1 \} : B$}
	\end{prooftree}
	\begin{prooftree}
		\AxiomC{$\Sigma_1; \Gamma_1; \pi^{\Gamma_1}_{0} (Q) \vdash e_1 : \abra{b_1, \pi^{x: b_1}_{0} (R)}$}
		\AxiomC{$\Sigma_2; x: b_1, \Gamma_2; R \vdash e_2 : B$}
		\RightLabel{\textsc{(M:Let)}}
		\BinaryInfC{$\Sigma_1 \cup \Sigma_2; \Gamma_1 \cup \Gamma_2; [1+]Q \vdash \m{let} \; x = e_1 \; \m{in} \; e_2 : B$}
	\end{prooftree}
	\begin{prooftree}
		\AxiomC{$\Sigma; \Gamma, x_1: b, x_2: b; P \vdash e: B$}
		\AxiomC{$Q \curlyveedownarrow_{x_1, x_2} P$}
		\RightLabel{\textsc{(M:Share-Base)}}
		\BinaryInfC{$\Sigma; \Gamma, x: b; Q \vdash \m{share} \; x \; \m{as} \; x_1, x_2 \; \m{in} \; e : B$}
	\end{prooftree}
	\begin{prooftree}
		\AxiomC{$\Sigma, f_1: \tau_1, f_2: \tau_2; \Gamma; Q \vdash e : B$}
		\AxiomC{$\tau, \tau_1, \tau_2 \in \mathcal{T}_{f}$}
		\RightLabel{\textsc{(M:Share-Arrow)}}
		\BinaryInfC{$\Sigma, f: \tau; \Gamma; Q \vdash \m{share} \; f \; \m{as} \; f_1, f_2 \; \m{in} \; e : B$}
	\end{prooftree}
	\end{small}
	\caption{Syntax-directed rules of multivariate AARA.
	Anything enclosed by square brackets is ignored if the cost-free metric is used. 
	For instance, if we have $[1+]Q$, it represents $Q$ in case of the cost-free metric and $1+Q$ in case the cost metric is the running time.}
	\label{fig:syntax-directed rules of multivariate AARA}
\end{figure}

In \textsc{(M:Var-Base)}, $c + Q$, where $c \in \mathbb{Q}_{\geq 0}$ and $Q$ is a multivariate resource annotation, denotes the addition of $x$ to $Q (0)$; i.e.~the constant potential of $Q$. 
The coefficients of all the other indexes, i.e.~$i \neq 0 \in \mathcal{I} (\Gamma)$, remain unchanged. 

In \textsc{(M:Case-Sum)}, $Q [(\ell \cdot y) / x]$ denotes the result of discarding the resource annotations within $Q$ for $x = r \cdot y$, which has a wrong tag. 

In \textsc{(M:Case-List)}, $\pi^{\Gamma}_{i} (Q)$ denotes the projection of $Q$ on $\Gamma$, which is defined in \eqref{eq:definition of projection of a multivariate resource annotation}. 

According to the first premise of \textsc{(M:Let)}, $e_1$ must carry $\pi^{x: b_1}_{0} (R)$ much potential after computation.
This potential is part of $R$ that only concerns size variables inside $x: b_1$. 
However, being multivariate in nature, $R$ may involve a product of both size variables outside $x$ and those inside $x$. 
How can we squeeze such multivariate potential into $Q$ such that it becomes available in $R$? 
To achieve this, we additionally impose this requirement: for each $j \in \mathcal{I} (\Gamma_2)$ such that $j \neq 0_{\Gamma_2}$, we have 
\begin{equation*}
	\Sigma_1; \Gamma_1; \pi^{\Gamma_1}_{j} Q \vdash e_1 : \abra{b_1, \pi^{x:b_1}_{j} (R)}
\end{equation*}
under the cost-free metric.
We could not place this requirement inside the rule \textsc{(M:Let)} simply due to the shortage of space. 

$Q \curlyveedownarrow_{x_1, x_2} P$ in \textsc{(M:Share-Base)} is a multivariate extension of the sharing relation $\tau \curlyveedownarrow (\tau_1, \tau_2)$ in Figure~\ref{fig:definition of the sharing relation of types}. 
In the interest of space, $Q \curlyveedownarrow_{x_1, x_2} P$ is not formally defined in this article. 

Finally, in the premise of \textsc{(M:Share-Arrow)}, $\tau_{i}$ must be a valid resource annotation of $f_{i}$ derived using \textsc{(M:Fun)}.
However, $\tau_1$ does not need to be identical to $\tau_2$.

Like in univariate AARA (Appendix~\ref{sec:univariate polynomial AARA}), the type system of multivariate AARA in Figure~\ref{fig:syntax-directed rules of multivariate AARA} can give rise to an infinite chain of \textsc{(M:Fun)}. 
To circumvent this, the type inference algorithm of multivariate AARA requires adopts the heuristic described in Appendix~\ref{sec:practical type inference for univariate AARA}. 
The only difference is that, under the cost-free metric, multivariate AARA's type inference creates a cascade of \textsc{(M:Fun)} where each invocation of \textsc{(M:Fun)} reduces the degree by one. 
On the other hand, under the cost-free metric, univariate AARA's type inference uses resource-monomorphic recursion rather than resource-polymorphic recursion. 

\paragraph*{Structural Rules}
Structural rules are presented in Figure~\ref{fig:structural rules of multivariate AARA}. 

\begin{figure}[hbt!]
	\begin{small}
	\begin{center}
		\AxiomC{$\Sigma; \Gamma; Q \vdash e: \abra{\tau, P_1}$}
		\AxiomC{$P_1 <: P_2$}
		\RightLabel{\textsc{(M:Sub)}}
		\BinaryInfC{$\Sigma; \Gamma; Q \vdash e : \abra{\tau, P_2}$}
		\DisplayProof
		\quad
		\AxiomC{$\Sigma; \Gamma; Q_1 \vdash e: B$}
		\AxiomC{$Q_2 <: Q_1$}
		\RightLabel{\textsc{(M:Sup)}}
		\BinaryInfC{$\Sigma; \Gamma; Q_2 \vdash e : B$}
		\DisplayProof
	\end{center}
	\begin{prooftree}
		\AxiomC{$\Sigma_1 \subseteq \Sigma_2$}
		\AxiomC{$\Sigma_1; \pi^{\Gamma}_{0} (Q) \vdash e: B$}
		\RightLabel{\textsc{(M:Weak)}}
		\BinaryInfC{$\Sigma_2; \Gamma; Q \vdash e: B$}
	\end{prooftree}
	\begin{prooftree}
		\AxiomC{$\Sigma; \Gamma; P \vdash e: \abra{\tau, P'}$}
		\AxiomC{$Q = P + c$}
		\AxiomC{$Q' = P' + c$}
		\RightLabel{\textsc{(M:Relax)}}
		\TrinaryInfC{$\Sigma; \Gamma; Q \vdash e: \abra{\tau, Q'}$}
	\end{prooftree}
	\end{small}
	\caption{Structural rules of multivariate AARA.}
	\label{fig:structural rules of multivariate AARA}
\end{figure}

The subtyping relationship is determined by the point-wise inequality of coefficients of resource polynomials: $Q <: P$ if and only if $\forall i \in \mathcal{I} (\Gamma). Q (i) \geq P (i)$.

\ifLongVersion

\section{Proof of the Typability Theorem}
\label{sec:proof of the typability theorem}

This section provides detailed proofs of Theorem~\ref{theorem:inherently polynomial time implies the existence of a multivariate annotation} and Theorem~\ref{theorem:inherently polynomial time implies the existence of a multivariate annotation where a user-specified amount of potential is available in the output}. 
The inductive proof of Theorem~\ref{theorem:inherently polynomial time implies the existence of a multivariate annotation} makes use of Theorem~\ref{theorem:inherently polynomial time implies the existence of a multivariate annotation where a user-specified amount of potential is available in the output}, while the inductive proof of Theorem~\ref{theorem:inherently polynomial time implies the existence of a multivariate annotation where a user-specified amount of potential is available in the output} is self-contained; that is, its statement is strong enough for an inductive proof to go through. 

First of all, in Figure~\ref{fig:remaining rules for inherently polynomial time}, we present three inference rules for inherently polynomial time that are missing from Figure~\ref{fig:syntax-directed rules of inherently polynomial time}. 

\begin{figure}[hbt!]
	\begin{small}
	\begin{prooftree}
		\AxiomC{$\Delta_1; \cdot \vdash e_1 \; \m{time}$}
		\AxiomC{$\Delta_2, x \; \m{time}; \Gamma, x: b \vdash e_2 \; \inhpoly{V}$}
		\RightLabel{\textsc{(IP:Let-Arrow)}}
		\BinaryInfC{$\Delta_1 \cup \Delta_2; \Gamma \vdash \m{let} \; x = e_1 \; \m{in} \; e_2 \; \inhpoly{V}$}
	\end{prooftree}
	\begin{prooftree}
		\AxiomC{$\Delta [x \; \m{time} \mapsto x_1 \; \m{time}, x_2 \; \m{time}]; \Gamma, x_1: b_1 \rightarrow b_2, x_2: b_1 \rightarrow b_2 \vdash e \; \inhpoly{V}$}
		\RightLabel{\textsc{(IP:Share-Arrow)}}
		\UnaryInfC{$\Delta; \Gamma, x: b_1 \rightarrow b_2 \vdash \m{share} \; x \; \m{as} \; x_1, x_2 \; \m{in} \; e \; \inhpoly{V}$}
	\end{prooftree}
	\begin{prooftree}
		\AxiomC{$\Delta_1; \Gamma_1 \vdash e \; \inhpoly{V_1}$}
		\AxiomC{$\Delta_1 \subseteq \Delta_2$}
		\AxiomC{$\Gamma_1 \subseteq \Gamma_2$}
		\AxiomC{$V_1 \subseteq V_2$}
		\RightLabel{\textsc{(IP:Weaken-Base)}}
		\QuaternaryInfC{$\Delta_2; \Gamma_2 \vdash e \; \inhpoly{V_2}$}
	\end{prooftree}
	\end{small}
	\caption{Remaining rules for inherently polynomial time}
	\label{fig:remaining rules for inherently polynomial time}
\end{figure}
In \textsc{(IP:Weaken-Base)}, we require $V_2$ to be a subset of base-type variables in $\dom{\Gamma}$. 

\existenceofannotation*

\begin{proof}
The proof proceeds by structural induction on $\Delta; \Sigma; \Gamma \vdash e \; \inhpoly{V}$. 

For base cases, we have \textsc{(IP:Base)}, \textsc{(IP:Arrow)}, \textsc{(IP:Unit)}, and \textsc{(IP:Nil)}. 
In all of them, since the running time is constant, we only need constant potential. 
Thus, the theorem is indeed true. 

Moving on to inductive cases, let us consider \textsc{(IP:SumL)}: 
\begin{prooftree}
	\AxiomC{$\cdot; x:b \vdash x \; \inhpoly{\emptyset}$}
	\RightLabel{\textsc{(IP:SumL)}}
	\UnaryInfC{$\cdot; x:b \vdash \ell \cdot x \; \inhpoly{\emptyset}$}
\end{prooftree}
Again, $\ell \cdot x$ runs in constant time. 
Therefore, it is easy to see that the theorem holds. 
The same reasoning applies to \textsc{(IP:SumR)}, \textsc{(IP:Pair)}, and \textsc{(IP:Cons)}. 

\textsc{(IP:Const)} and \textsc{(IP:Poly)} are straightforward. 

Next, we consider \textsc{(IP:App-Const)}: 
\begin{prooftree}
	\AxiomC{$\Delta = \{x_1 \; \m{const}\}$}
	\RightLabel{\textsc{(IP:App-Const)}}
	\UnaryInfC{$\Delta; x_1: b_1 \rightarrow b_2, x_2: b_1 \vdash x_1 \; x_2 \; \inhpoly{\emptyset}$}
\end{prooftree}
From the inductive hypothesis, it is given that $x_1: b_1 \rightarrow b_2$ can be typed as $\abra{b_1, P} \rightarrow \abra{b_2, Q}$, where $P$ contains constant potential. 
Therefore, $x_1 \; x_2$ can be annotated in such a way that the annotation for $x_2$ contains zero potential; that is, only constant potential is needed. 
Thus, the theorem holds.
The inductive case for \textsc{(IP:App-Poly)} can be proved in the same manner. 

The next case we consider is \textsc{(IP:Case-Sum)}: 
\begin{prooftree}
	\AxiomC{$\Delta; \Gamma, y: b_1 \vdash e_{\ell} \; \inhpoly{V [x \mapsto y]}$}
	\AxiomC{$\Delta; \Gamma, y: b_2 \vdash e_{r} \; \inhpoly{V [x \mapsto y]}$}
	\RightLabel{\textsc{(IP:Case-Sum)}}
	\BinaryInfC{$\Delta; \Gamma, x: b_1 + b_2 \vdash \m{case} \; x \; \{\ell \cdot y \hookrightarrow e_{\ell} \mid r \cdot y \hookrightarrow e_{r}\} \; \inhpoly{V}$}
\end{prooftree}
Applying the inductive hypothesis to the two premises, we obtain that both $e_{\ell}$ and $e_{r}$ are typable in multivariate AARA. 
Let $P_1$ and $P_2$ be the multivariate annotations for the two premises. 
We can derive a multivariate annotation for the conclusion by taking $\max \{P_1 (i), P_2 (i)\}$ for each index/base polynomial $i \in \mathcal{I} (\Gamma \cup \{ y:b_{1,2} \})$. 
Furthermore, it follows from the inductive hypothesis that any variable $v \in \dom{\Gamma} \cup \{y\} \setminus V [x \mapsto y]$ has zero potential in both $P_1$ and $P_2$. 
Consequently, $\max \{P_1 (i), P_2 (i)\} = \max \{0, 0\} = 0$ if index $i$ involves a size variable from $v$. 
This establishes the theorem. 
\textsc{(IP:Case-Prod)} can be proved in the same fashion. 

We next consider \textsc{(IP:Case-List)}: 
\begin{prooftree}
	\AxiomC{$\Delta; \Gamma \vdash e_0 \; \inhpoly{V \setminus \{x\}}$}
	\AxiomC{$\Delta; \Gamma, x_1: b, x_2: L (b) \vdash e_1 \; \inhpoly{V [x \mapsto x_1, x_2]}$}
	\RightLabel{\textsc{(IP:Case-List)}}
	\BinaryInfC{$\Delta; \Gamma, x: L (b) \vdash \m{case} \; x \; \{[\,] \hookrightarrow e_0 \mid (x_1::x_2) \hookrightarrow e_1 \} \; \inhpoly{V}$}
\end{prooftree}
Assume that the inductive hypothesis gives us annotations $P_0$ and $P_1$ for the two premises of \textsc{(IP:Case-List)}. 
From $P_1$, we can construct an annotation $P_2$ over the typing context $\Gamma \cup \{x: L (b)\}$ such that $P_1 = \lhd (P_2)$. 
It follows from the definition of $\lhd$ that
\begin{equation} \label{eq:definition of the multivariate additive shift in the inductive case for (IP:Case-List) in the proof for the first typability theorem}
	\lhd (P_2) (i, j, k) := 
	\begin{cases}
		P_2 (i, 0_{b} :: k) + P_2 (i, k) & \text{if } j = 0_{b}; \\
		P_2 (i, j::k) & \text{otherwise}. 
	\end{cases}
\end{equation}
We will define $P_2$ as follows. 
For any variable $v \in (\dom{\Gamma} \cup \{x\}) \setminus V$, the coefficient of a base polynomial in $P_2$ that involves $v$'s size variables is set to $0$.
All the other coefficients are set to the largest coefficient that appears in $P_1$. 
Consequently, by construction, every $v \in (\dom{\Gamma} \cup \{x\}) \setminus V$ contains zero potential in $P_2$.
Let $M$ denote the largest coefficient in $P_1$. 

It remains to ascertain that $\lhd (P_2)$ is a subtype of $P_1$. 
We will first consider the first clause of \eqref{eq:definition of the multivariate additive shift in the inductive case for (IP:Case-List) in the proof for the first typability theorem}. 
Suppose that $P_1 (i, j, k) > 0$ for some $i, j, k$. 
This implies that $(i, j, k)$ does not involve any variables from $(\dom{\Gamma} \cup \{x_1, x_2\}) \setminus V[x \mapsto x_1, x_2]$; otherwise, $P_1 (i, j, k) = 0$ due to the inductive hypothesis of the second premise. 
If $j = 0_{b}$, then $P_2 (i, k) = M$ holds because
\begin{equation*}
	\begin{split}
	& (i, 0_{b}, k) \text{ contains size variables of some } v \in (\dom{\Gamma} \cup \{x_1, x_2\}) \setminus V [x \mapsto x_1, x_2] \\
	& \Longleftrightarrow (i, k) \text{ contains size variables of some } v \in (\dom{\Gamma} \cup \{x\}) \setminus V. 
	\end{split}
\end{equation*} 
Thus, $\lhd (P_2) (i, j, k) \geq P_1 (i, j, k)$ in this case.
Conversely, if $j \neq 0_{b}$, $P_2 (i, j::k)$ is guaranteed to be $M$. 
Hence, $\lhd (P_2) (i, j, k) \geq P_1 (i, j, k)$ is true in this case as well.
Therefore, $\lhd (P_2)$ is indeed a subtype of $P_1$. 
Finally, we can easily combine $P_2$ with $P_0$ to yield a desirable annotation for the conclusion of \textsc{(IP:Case-List)}. 

Next is \textsc{(IP:Rec)}: 
\begin{prooftree}
	\AxiomC{$\Delta; \Gamma \vdash e_0 \; \inhpoly{V}$}
	\AxiomC{$\cdot; y:b, \mi{ys}: L (b), z: b_2 \vdash e_1 \; \inhpoly{\{y, \mi{ys}\}}$}
	\RightLabel{\textsc{(IP:Rec)}}
	\BinaryInfC{$\Delta; \Gamma, x: L (b) \vdash \m{rec} \; x \; \{[\,] \hookrightarrow e_0 \mid (y::\mi{ys}) \; \m{with} \; z \hookrightarrow e_1\} \; \inhpoly{V \cup \{x\}}$}
\end{prooftree}
Appealing to the inductive hypothesis, we know that $e_0$ and $e_1$ can be annotated as
\begin{equation*}
	\Gamma; P_0 \vdash e_0 : \abra{b_2, Q_0} \qquad y:b, \mi{ys}: L (b), z: b_2; P_1 \vdash e_1 : \abra{b_2, Q_1}, 
\end{equation*}
where $P_1$ assigns zero potential to $z$ because of $e_1 \; \inhpoly{\{y, \mi{ys}\}}$. 
From $P_0$ and $P_1$, it is possible to construct a multivariate annotation for the conclusion of \textsc{(IP:Rec)}. 
Furthermore, any variable in $\dom{\Gamma} \setminus V$ stores zero potential. 
The details of this construction are presented in Lemma~\ref{lemma:construction of a multivariate annotation for primitive recursion}. 

The next case we consider is \textsc{(IP:Let-Base)}: 
\begin{prooftree}
	\AxiomC{$\Delta_1; \Sigma_1; \Gamma_1 \vdash e_1 \; \inhpoly{V_1}$}
	\AxiomC{$\Delta_2; \Gamma_2, x: b_1 \vdash e_2 \; \inhpoly{V_2}$}
	\RightLabel{\textsc{(IP:Let-Base)}}
	\BinaryInfC{$\Delta_1 \cup \Delta_2; \Sigma_1 \cup \Gamma_1 \cup \Gamma_2 \vdash \m{let} \; x = e_1 \; \m{in} \; e_2 \; \inhpoly{V_3}$}
\end{prooftree}
The inductive hypothesis tells us that there exists a multivariate annotation
\begin{equation*}
	\Gamma_2, x: b; P_2 \vdash e_2 : \abra{b_2, Q_2}. 
\end{equation*}
Let $P$ be a multivariate annotation of the entire let-binding that we aim to derive. 
According to \textsc{(M:Let)}, $P$ must satisfy the following two conditions: 
\begin{enumerate}
	\item $\Sigma_1; \Gamma_1; \pi^{\Gamma_1}_{0} (P) \vdash e_1 : \abra{b_1, \pi^{x: b_1}_{0} (P_2)}$ holds under the cost metric of the running time;
	\item For all $i \neq 0_{b} \in \mathcal{I} (\Gamma_2)$, we have $\Sigma_1; \Gamma_1; \pi^{\Gamma_1}_{i} (P) \vdash e_1 : \abra{b_1, \pi^{x:b_1}_{i} (P_2)}$ under the cost-free metric. 
\end{enumerate}

We will now conduct case analysis on whether $x \in V_2$. 
Assume $x \notin V_2$. 
It follows from the inductive hypothesis that $x$ contains zero potential in $P_2$. 
As a result, $\pi^{x: b_1}_{0} (P_2)$ in the first condition above is essentially constant potential. 
Applying the inductive hypothesis to $e_1$, we obtain a multivariate annotation of $e_1$ under the cost metric of the running time. 
If the output of $e_1$ in this annotation contains less potential than $\pi^{x: b_1}_{0} (P_2)$, we can always inject constant potential into the annotation of $e_1$. 
The resulting annotation will serve as a suitable $\pi^{\Gamma_1}_{0} (P)$. 
Likewise, $\pi^{x:b_1}_{i} (P_2)$ in the second condition above, where $i \neq 0_{b}$, is constant potential. 
Therefore, a suitable $\pi^{\Gamma_1}_{i} (P)$ can be constructed (note that the second condition concerns the cost-free metric). 

Lastly, we need to ensure that any $v \in \dom{\Gamma_1} \cup \dom{\Gamma_2} \setminus V_3$ has zero potential in $P$. 
If $v \in \dom{\Gamma_2} \setminus V_2$, since $\pi^{\Gamma_1}_{i} (P)$ will be constant potential, it is impossible for any base polynomial with a size variable of $v$ to have a non-zero coefficient. 
The same reasoning applies to the case of $v \in \dom{\Gamma_1} \setminus V_1$. 

Conversely, if $x \in V_2$, $e_2$ demands potential from $x = e_1$, meaning that we need a multivariate annotation for $e_1$ with some potential available in the output of $e_1$. 
Although the theorem gives us \emph{some} multivariate annotation of $e_1$, we have no guarantee that the output of $e_1$ contains a desired arbitrary amount of potential. 
This is where Theorem~\ref{theorem:inherently polynomial time implies the existence of a multivariate annotation where a user-specified amount of potential is available in the output} comes in. 
By Theorem~\ref{theorem:inherently polynomial time implies the existence of a multivariate annotation where a user-specified amount of potential is available in the output}, we can derive a cost-free annotation of $e_1$ such that its output stores a desirable amount of potential. 
Finally, summing this cost-free annotation with the annotation given by Theorem~\ref{theorem:inherently polynomial time implies the existence of a multivariate annotation}, we obtain a suitable $\pi^{\Gamma_1}_{0} (P)$ in the first condition above. 
Regarding $\pi^{\Gamma_1}_{i} (P)$ for $i \neq 0_{b}$ in the second condition, again, we resort to Theorem~\ref{theorem:inherently polynomial time implies the existence of a multivariate annotation where a user-specified amount of potential is available in the output}. 

Lastly, we need to ensure that any $v \in \dom{\Gamma_1} \cup \dom{\Gamma_2} \setminus V_3$ contains zero potential in $P$. 
As we assume $x \in V_2$, we have $V_3 = \dom{\Gamma_1} \cup (V_2 \setminus \{x\})$. 
Therefore, it is guaranteed that $v \in \dom{\Gamma_2} \setminus V_2$. 
Due to the inductive hypothesis of the theorem on $e_2$, $v$ contains zero potential in $P_2$. 
Thus, for any $i \in \mathcal{I} (\Gamma)$ that involves a size variable of $v$, $\pi^{x: b_1}_{i} (P_2)$ is essentially constant potential. 
Therefore, $\pi^{\Gamma_1}_{i} (P)$ will be constant potential as well. 
Consequently, $v$ will contain zero potential in $P$, which is the annotation of the whole let-binding. 

The final inductive case we consider is \textsc{(IP:Share-Base)}: 
\begin{prooftree}
	\AxiomC{$\Delta; \Gamma, x_1: b, x_2: b \vdash e \; \inhpoly{V [x \mapsto x_1, x_2]}$}
	\RightLabel{\textsc{(IP:Share-Base)}}
	\UnaryInfC{$\Delta; \Gamma, x: b \vdash \m{share} \; x \; \m{as} \; x_1, x_2 \; \m{in} \; e \; \inhpoly{V}$}
\end{prooftree}
By the inductive hypothesis, the premise of \textsc{(IP:Share-Base)} can be assigned a multivariate annotation. 
Using this, we can build a multivariate annotation for the conclusion by merging the coefficients for those base polynomials that mention either $x_1$ or $x_2$. 
Moreover, if $x \notin V$, then $x_1, x_2 \notin V [x \mapsto x_1, x_2]$. 
According to the inductive hypothesis, both $x_1$ and $x_2$ store zero potential.
As a result, in the multivariate annotation of the rule's conclusion, $x$ stores zero potential as well.
This establishes the theorem. 

The inductive cases for \textsc{(IP:Let-Arrow)} and \textsc{(IP:Share-Arrow)} are straightforward to prove since they do not affect base-type variables; hence, we will not formally present their proof. 
Likewise, it is immediate to prove \textsc{(IP:Weaken-Base)}. 
This concludes the proof. 
\end{proof}

\typablilitytheorem*

\begin{proof}
Like in the above proof of Theorem~\ref{theorem:inherently polynomial time implies the existence of a multivariate annotation}, this proof will proceed by induction on $\Delta; \Gamma \vdash e \; \inhpoly{V}$. 
Although this document does not present inference rules for the cost-free annotations in multivariate AARA, the absence of their formal presentation should not affect the typability proof. 

As before, nothing interesting happens in the base cases: \textsc{(IP:Base)}, \textsc{(IP:Arrow)}, \textsc{(IP:Unit)}, and \textsc{(IP:Nil)}. 
Likewise, it is straightforward to prove the inductive cases of \textsc{(IP:SumL)}, \textsc{(IP:SumR)}, and \textsc{(IP:Pair)}. 

By contrast, \textsc{(IP:Cons)} is nontrivial: 
\begin{prooftree}
	\AxiomC{$\cdot; x_1: b \vdash x_1 \; \inhpoly{\emptyset}$}
	\AxiomC{$\cdot; x_2: L (b) \vdash x_2 \; \inhpoly{\emptyset}$}
	\RightLabel{\textsc{(IP:Cons)}}
	\BinaryInfC{$\cdot; x_1: b, x_2: L (b) \vdash x_1::x_2 \; \inhpoly{\emptyset}$}
\end{prooftree}
Suppose that we would like $x_1 :: x_2$ to be annotated with $Q$. 
If $Q$ is univariate, it is clear that there exists $P$ such that 
\begin{equation*}
	x_1: b, x_2: L (b); P \vdash x_1 :: x_2 : \abra{L (b), Q}
\end{equation*}
and more importantly, $P$ and $Q$ only differ in lower-degree terms. 
This is because $\vec{q}$ and $\lhd (\vec{q})$ have the same coefficient of the maximum degree for any potential vector $\vec{q}$. 
If $Q$ is multivariate, the proof is more complicated and is deferred to Lemma~\ref{lemma:additive shift of multivariate potential preserves uniformity of multivariate annotations}. 

As before, it is straightforward to prove \textsc{(IP:Const)} and \textsc{(IP:Poly)}. 

We next consider \textsc{(IP:App-Const)}: 
\begin{prooftree}
	\AxiomC{$\Delta = \{x_1 \; \m{const}\}$}
	\RightLabel{\textsc{(IP:App-Const)}}
	\UnaryInfC{$\Delta; x_1: b_1 \rightarrow b_2, x_2: b_1 \vdash x_1 \; x_2 \; \inhpoly{\emptyset}$}
\end{prooftree}
From the inductive hypothesis of $x_1 \; \m{const}$, we know that the annotation of $x_1$'s input is identical to that of $x_1$'s output when restricting our attention to base polynomials of degree $d$. 
Therefore, the claim holds. 
\textsc{(IP:App-Poly)} can be proved straightforwardly. 

Next is \textsc{(IP:Case-Sum)}: 
\begin{prooftree}
	\AxiomC{$\Delta; \Gamma, y: b_1 \vdash e_{\ell} \; \inhpoly{V [x \mapsto y]}$}
	\AxiomC{$\Delta; \Gamma, y: b_2 \vdash e_{r} \; \inhpoly{V [x \mapsto y]}$}
	\RightLabel{\textsc{(IP:Case-Sum)}}
	\BinaryInfC{$\Delta; \Gamma, x: b_1 + b_2 \vdash \m{case} \; x \; \{\ell \cdot y \hookrightarrow e_{\ell} \mid r \cdot y \hookrightarrow e_{r}\} \; \inhpoly{V}$}
\end{prooftree}
From the inductive hypothesis, we have $P_1$ and $P_2$ satisfying
\begin{equation*}
	\Gamma, y: b_1; P_1 \vdash e_{\ell} : \abra{b, Q} \qquad \Gamma, y: b_2; P_2 \vdash e_{r} : \abra{b, Q}. 
\end{equation*}
To construct a desirable multivariate annotation for the conclusion, we simply need to integrate $P_1$ and $P_2$ by taking the maximum coefficient for each base polynomial. 
\textsc{(IP:Case-Prod)} can be proved similarly. 

Next, we consider \textsc{(IP:Case-List)}: 
\begin{prooftree}
	\AxiomC{$\Delta; \Gamma \vdash e_0 \; \inhpoly{V \setminus \{x\}}$}
	\AxiomC{$\Delta; \Gamma, x_1: b, x_2: L (b) \vdash e_1 \; \inhpoly{V [x \mapsto x_1, x_2]}$}
	\RightLabel{\textsc{(IP:Case-List)}}
	\BinaryInfC{$\Delta; \Gamma, x: L (b) \vdash \m{case} \; x \; \{[\,] \hookrightarrow e_0 \mid (x_1::x_2) \hookrightarrow e_1 \} \; \inhpoly{V}$}
\end{prooftree}
Lemma~\ref{lemma:pattern matching on lists preserves uniformity of multivariate annotations} provides details of how to construct an annotation for $\m{case} \; x \; \{[\,] \hookrightarrow e_0 \mid (x_1::x_2) \hookrightarrow e_1 \}$. 

The next inductive case is \textsc{(IP:Rec)}: 
\begin{prooftree}
	\AxiomC{$\Delta; \Gamma \vdash e_0 \; \inhpoly{V}$}
	\AxiomC{$\cdot; y:b, \mi{ys}: L (b), z: b_2 \vdash e_1 \; \inhpoly{\{y, \mi{ys}\}}$}
	\RightLabel{\textsc{(IP:Rec)}}
	\BinaryInfC{$\Delta; \Gamma, x: L (b) \vdash \m{rec} \; x \; \{[\,] \hookrightarrow e_0 \mid (y::\mi{ys}) \; \m{with} \; z \hookrightarrow e_1\} \; \inhpoly{V \cup \{x\}}$}
\end{prooftree}
Suppose that the inductive hypothesis of the second premise yields
\begin{equation*}
	y: b, \mi{ys}: L (b), z: b_2; P_1 \vdash e_1 : \abra{b_2, Q}, 
\end{equation*}
where $Q \; \m{uniform} (d, n)$ and $P_1 \; \m{uniform} (d, n, \{y, \mi{ys}\})$.
It follows that $\pi^{z: b_2}_{0} (P) - Q$, which will be paid by resource-polymorphic recursion, has a strictly lower degree than $d$. 
Therefore, resource-polymorphic recursion will derive a multivariate annotation for $z$ such that it has the annotation $\pi^{z: b_2}_{0} (P) - Q$. 
To formally prove this, we should have performed strong induction on $d$ as well as structural induction on $\Delta; \Sigma; \Gamma \vdash e \; \inhpoly{V}$. 
However, since it might make the proof overly complicated and thereby confuse the readers, we decided to hide this detail until now. 
Note that if $d = 0$, the theorem clearly holds.
If we only demand $\Sigma; \Gamma \vdash e: b$ to have constant potential in the output, it is immediate to type $e$'s context under the cost-free metric. 

Assume that resource-polymorphic recursion yields $P_{2, i}$ for $i \in \mathcal{I} (\{y: b, \mi{ys}: L (b)\})$ defined by 
\begin{alignat*}{2}
	\Gamma, x: L (b); P_{2, 0} & \vdash e: \abra{b_2, \pi^{z: b_2}_{0} (P) - Q} && \qquad \text{if } i = 0; \\
	\Gamma, x: L (b); P_{2, i} & \vdash e: \abra{b_2, \pi^{z: b_2}_{i} (P)} && \qquad \text{otherwise}, 
\end{alignat*}
where $e$ refers to the entire primitive recursion. 
Note that all of $\pi^{z: b_2}_{0} (P) - Q$ and $\pi^{z: b_2}_{i} (P)$ for any $i \neq 0$ have degrees lower than $d$; hence, the existence of their annotation can be proved by strong induction on $d$. 
Our goal is to build an annotation $P$ for $e$ from $P_1$ and $P_{2, i}$. 
The details of $P$'s construction are provided in Lemma~\ref{lemma:construction of a multivariate annotation for primitive recursion with resource-polymorphic recursion}.  

Next is \textsc{(IP:Let-Base)}: 
\begin{prooftree}
	\AxiomC{$\Delta_1; \Sigma_1; \Gamma_1 \vdash e_1 \; \inhpoly{V_1}$}
	\AxiomC{$\Delta_2; \Gamma_2, x: b_1 \vdash e_2 \; \inhpoly{V_2}$}
	\RightLabel{\textsc{(IP:Let-Base)}}
	\BinaryInfC{$\Delta_1 \cup \Delta_2; \Sigma_1 \cup \Gamma_1 \cup \Gamma_2 \vdash \m{let} \; x = e_1 \; \m{in} \; e_2 \; \inhpoly{V_3}$}
\end{prooftree}
Suppose that the inductive hypothesis of the second premise yields
\begin{equation*}
	\Gamma_2, x: b_1; P_2 \vdash e_2 : \abra{b_2, Q}, 
\end{equation*}
where $Q \; \m{uniform} (d, n)$. 
As we currently work with the cost-free metric, the premises of \textsc{(M:Let)} can be simplified to
\begin{equation} \label{eq:simplified premise for resource-typing a let-binding under the cost-free metric in multivariate AARA}
	\forall i \in \mathcal{I} (\Gamma_2). \Sigma_1; \Gamma_1; \pi^{\Gamma_1}_{i} (P) \vdash e_1 : \abra{b_1, \pi^{x: b_1}_{i} (P_2)}
\end{equation}
under the cost-free metric. 
Here, $P$ is a multivariate annotation of the whole let-binding---it is what we aim to derive in this proof. 

We now conduct case analysis. 
If $x \in V_2$, it is fairly easy to establish the claim. 
First of all, we fix $i \in \mathcal{I} (\Gamma_2)$. 
From $\pi^{x: b_1}_{i} (P_2)$, the inductive hypothesis allows us to create a suitable $\pi^{\Gamma_1}_{i} (P)$ that satisfies \eqref{eq:simplified premise for resource-typing a let-binding under the cost-free metric in multivariate AARA}. 
If $\pi^{x: b_1}_{i} (P_2)$ is not uniform at the maximum degree (this is required by the theorem), we can easily create a uniform annotation that is a subtype of $\pi^{x: b_1}_{i} (P_2)$. 
It remains to ensure $P \; \m{uniform} (d, n, V_3)$. 
The proof of this case is identical to the proof of the next case; hence, we omit it. 

Next, assume $x \notin V_2$. 
As above, we can create a multivariate annotation for the whole let-binding. 
It remains to ensure that $P$ is a uniform annotation. 
Let $v$ be a variable drawn from $\dom{\Gamma_1 \cup \Gamma_2} \setminus V_3$, where $V_3 = V_1 \cup V_2$ due to the assumption $x \notin V_2$. 
If $v \in \dom{\Gamma_2} \setminus V_2$, the inductive hypothesis of $e_2$ already implies $P_2 \; \m{uniform} (d, n, V_2)$. 
It is easy to see that the first condition of Definition~\ref{def:uniform resource annotations for typing contexts in multivariate AARA} holds for $P$ with respect to variable $v$. 

To establish the third condition (and also the second condition) of Definition~\ref{def:uniform resource annotations for typing contexts in multivariate AARA}, consider 
\begin{equation*}
	i = (h, g) \in \mathcal{I} (\Gamma_2 \setminus \{v: b\}) \times \mathcal{I} (\{v: b\}), 
\end{equation*}
where $h = 0_{\Gamma_2 \setminus \{v: b\}}$ and $\degree{g} = d$. 
Due to the inductive hypothesis $P_2 \; \m{uniform} (d, n, V_2)$, $\pi^{x: b_1}_{i} (P_2)$ must essentially be constant potential of $n$.
Hence, $\pi^{\Gamma_1}_{i} (P)$ in \eqref{eq:simplified premise for resource-typing a let-binding under the cost-free metric in multivariate AARA} is constant potential of $n$ as well. 
Therefore, $P (0_{\Gamma_1}, i) = n$ holds, thereby establishing the third condition of $P \; \m{uniform} (d, n, V)$. 
The second condition of $P \; \m{uniform} (d, n, V)$ can be established by the same reasoning. 
This proof is also applicable to the previous case, where $x \in V_2$ and $v \in \dom{\Gamma_2} \setminus V_2$. 

Lastly, let us consider the final case of the case analysis: $x \notin V_2$ and $v \in \dom{\Gamma_1} \setminus V_1$. 
Due to the inductive hypothesis $P_2 \; \m{uniform} (d, n, V_2)$, $\pi^{x:b_1}_{i} (P_2)$ has degree at most $d$. 
Therefore, it follows from the inductive hypothesis of $e_1$ that, for any base polynomial $r$ in $\pi^{\Gamma_1}_{i} (P)$ with a non-zero coefficient, $r$'s projection on $v$ must have degree at most $d$.
This establishes the first condition of $P \; \m{uniform} (d, n, V)$ with respect to $v$. 

Furthermore, if $i \neq 0_{\Gamma_2}$, then $\pi^{x: b_1}_{i} (P_2)$ has degree at most $d - 1$. 
This means that $\pi^{\Gamma_1}_{i} (P)$ cannot have a base polynomial whose projection on $v$ has degree $d$.
It can only be degree $d - 1$ at largest. 
Consequently, we have established the second condition of $P \; \m{uniform} (d, n, V)$ with respect to $v$. 

Finally, consider $i = 0_{\Gamma_2}$. 
$\pi^{x: b_1}_{0} (P_2)$ satisfies $\m{uniform} (d, n)$.  
Therefore, by the inductive hypothesis of $e_1$, we have $\pi^{\Gamma_1}_{0} (P) \; \m{uniform} (d, n, V_1)$. 
This establishes the third condition of Definition~\ref{def:uniform resource annotations for typing contexts in multivariate AARA}. 

Finally, we will discuss \textsc{(IP:Share-Base)}: 
\begin{prooftree}
	\AxiomC{$\Delta; \Gamma, x_1: b, x_2: b \vdash e \; \inhpoly{V [x \mapsto x_1, x_2]}$}
	\RightLabel{\textsc{(IP:Share-Base)}}
	\UnaryInfC{$\Delta; \Gamma, x: b \vdash \m{share} \; x \; \m{as} \; x_1, x_2 \; \m{in} \; e \; \inhpoly{V}$}
\end{prooftree}
Assume the inductive hypothesis of the premise yields 
\begin{equation*}
	\Gamma, x_1: b, x_2: b; P \vdash e : \abra{b_2, Q}, 
\end{equation*} 
where $Q$ is specified by a user and satisfies $Q \; \m{uniform} (d, n)$. 
From $P$, we can easily obtain a multivariate annotation $P'$ for the context $\Gamma \cup \{x: b\}$. 
If $x \in V$, the theorem is true for $P'$. 
Conversely, if $x \notin V$, we need to establish $\pi^{x:b}_{0} (P') \; \m{uniform} (d, n)$.
Although the inductive hypothesis gives us $\pi^{x:b}_{0} (P) \; \m{uniform} (d, n)$, it is not always the case that $\pi^{x:b}_{0} (P') \; \m{uniform} (d, n)$ holds. 
To circumvent this problem, we impose the restriction that no variable sharing is permitted on those variables that are outside $V$ (or their constituent variables derived by pattern matching). 

The remaining cases (i.e.~\textsc{(IP:Let-Arrow)}, \textsc{(IP:Share-Arrow)}, and \textsc{(IP:Weaken-Base)}) are immediate to prove. 
This concludes the proof. 
\end{proof}

\begin{lemma}[Construction of a multivariate annotation for primitive recursion]
\label{lemma:construction of a multivariate annotation for primitive recursion}
Consider a primitive recursion $e$ of the form
\begin{equation*}
	\Sigma; \Gamma, x: L (b) \vdash \m{rec} \; x \; \{[\,] \hookrightarrow e_0 \mid (y::\mi{ys}) \; \m{with} \; z \hookrightarrow e_1 \} : b_2. 
\end{equation*}
Suppose we have $P_{i}, Q_{i}$ for $i \in \{0, 1\}$ such that
\begin{equation*}
	\Gamma; P_0 \vdash e_0 : \abra{b_2, Q_0} \qquad y: b, \mi{ys}: L (b), z: b_2; P_1 \vdash e_1 : \abra{b_2, Q_1}
\end{equation*}
under the cost metric of the running time. 
Here, $z$ contains zero potential in $P_1$. 
Then there exists a multivariate annotation $R$ such that $\Sigma; \Gamma, x: L (b); R \vdash e : \abra{b_2, Q_1}$ holds under the cost metric of the running time. 

Furthermore, assume $\Delta; \Sigma; \Gamma \vdash e_0 \; \inhpoly{V}$. 
For any $v \in \dom{\Gamma} \setminus V$, if $v$ contains zero potential in $P_0$, then $v$ contains zero potential in $R$ as well. 
\end{lemma}
\begin{proof}
Using general recursion, the primitive recursion $e$ can be written as
\begin{equation*}
	f := \lam x, \Gamma. \m{case} \; x \; \{[\,] \hookrightarrow e_0 \mid (y::\mi{ys}) \hookrightarrow \m{share} \; \mi{ys} \; \m{as} \; \mi{ys}_1, \mi{ys}_2 \; \m{in} \; \m{let} \; z = f \; \mi{ys}_1 \; \Gamma \; \m{in} \; e_1 \}. 
\end{equation*}
Here, $\mi{ys}_1$ is used in the recursive call, and $\mi{ys}_2$ is (possibly) used inside $e_1$. 

It is safe to assume that the potential for $x: L (b)$ and that for $\Gamma$ are ``completely separated'' in $R$.
In other words, for any index $(i_{x}, i_{\Gamma}) \in \mathcal{I} (x: b) \times \mathcal{I} (\Gamma)$, if $i_{x} \neq 0_{b} \land i_{\Gamma} \neq 0_{\Gamma}$, then $R (i_{x}, i_{\Gamma}) = 0$. 
For example, if we have some variable $(u: L (\mathbf{1})) \in \Gamma$, then $R$ never contains a factor like $\abs{x} \cdot \abs{u}$; i.e.~multiplication of size variables from $x: L (b)$ and size variables from $(u: L (\mathbf{1})) \in \Gamma$. 
This is a reasonable assumption because (i) variables in $\Gamma$ are each used only once (namely, the very first iteration of the primitive recursion) and (ii) $z$ is assumed to contain zero potential in $P_1$. 
Therefore, intuitively, the multivariate polynomial function represented by $R$ should have the shape $p_{x} (\abs{x}) + p_{\Gamma} (\abs{\Gamma})$, where $p_{x}, p_{\Gamma}$ are some multivariate polynomial functions and $\abs{\cdot}$ is the set of size variables in the input expression, instead of $p_{x} (\abs{x}) + p_{\Gamma} (\abs{\Gamma}) + p_{x, \Gamma} (\abs{x}, \abs{\Gamma})$ which ``mixes'' the potential from $x$ and the potential from $\Gamma$.  

In the second branch of $e$, $x = (y :: \mi{ys})$ is annotated with $\pi^{x: L (b)}_{0} (R)$, and $\Gamma$ is annotated with $\pi^{\Gamma}_{0} (R)$. 
For $\pi^{x: L (b)}_{0} (R)$, it becomes $\lhd (\pi^{x: L (b)}_{0} (R))$ as a result of pattern matching on $x$. 
Because $\mi{ys}_1$ is used in the recursive call, it ought to be supplied with the same potential as $x: L (b)$; i.e.~$\pi^{x: L (b)}_{0} (R)$. 
The remaining potential for $y$ and $\mi{ys}_2$, which is given by $\lhd (\pi^{x: L (b)}_{0} (R)) - \eta^{\{y: b, \mi{ys}: L (b) \}}_{0} (\pi^{x: L (b)}_{0} (R))$, should be equal to $Q_1$. 

Therefore, our goal is to find a suitable $\pi^{x: L (b)}_{0} (R)$ for $Q_1$. 
After this step, $\pi^{x: L (b)}_{0} (R)$ is combined with $P_0$, yielding a desirable $R$. 
Fortunately, as $Q_1$ does not need potential from $z$, we need not be concerned about it. 

Let $A$ be a multivariate annotation for $x: L (b)$.
$A$ represents $\pi^{x: L (b)}_{0} (R)$, and we will now work out what $A$ should be. 
$\lhd (A)$, which is over the typing context $\{y: b, \mi{ys}: L (b) \}$, is given as
\begin{equation*}
	\lhd (A) (i, j) := 
	\begin{cases}
		A (0_{b} :: j) + A (j) & \text{if } i = 0_{b}; \\
		A (i :: j) & \text{otherwise}. 
	\end{cases}
\end{equation*}
If $i = 0_{b}$, we have
\begin{align*}
	\lhd (A) (0_{b}, j) - A (j) & = A (0_{b} :: j) + A (j) - A (j) \\
	& = A (0_{b} :: j), 
\end{align*}
and this should be equal to $\pi^{\{y: b, \mi{ys}: L (b) \}}_{0} (Q_1) (0_{b}, j)$. 
If $i \neq 0_{b}$, we should have $\lhd (A) (i, j) = A (i::j) = \pi^{\{y: b, \mi{ys}: L (b) \}}_{0} (Q_1) (i, j)$. 
In summary, we have
\begin{equation} \label{eq:condition for A in the construction of a multivariate annotation for a primitive recursion under the cost metric of the running time}
	A (i::j) := \pi^{\{y: b, \mi{ys}: L (b) \}}_{0} (Q_1) (i, j) 
\end{equation}
regardless of whether $i = 0_{b}$ or not. 
It is clear that, given $Q_1$, such $A$ exists. 
As a sanity check for \eqref{eq:condition for A in the construction of a multivariate annotation for a primitive recursion under the cost metric of the running time}, we can check whether it is correct when $Q_1$ is univariate (and hence $A$ is univariate).
However, we will omit the details of the sanity check. 

Finally, if $v \in \dom{\Gamma} \setminus V$ contains zero potential in $P_0$, $v$ also has zero potential in $R$ constructed as above. 
This is essentially because $\pi^{x: L (b)}_{0} (R)$ and $P_0$ are completely separated in $R$. 
Note that the assumption that $v \in \dom{\Gamma} \setminus V$ contains zero potential in $P_0$ is true due to Theorem~\ref{theorem:inherently polynomial time implies the existence of a multivariate annotation}. 
However, it is treated as an assumption rather than a fact in the present lemma. 
This concludes the proof. 
\end{proof}

\begin{lemma}[Preservation of uniformity by list constructors]
\label{lemma:additive shift of multivariate potential preserves uniformity of multivariate annotations}
Consider a typing judgment $x_1: b, x_2: L (b) \vdash x_1 :: x_2 : L (b)$. 
Fix a multivariate annotation $Q$ of degree $d$ for $x_1 :: x_2$ such that $Q \; \m{uniform} (d, n)$. 
There exists a cost-free multivariate annotation $P$ such that
\begin{itemize}
	\item $x_1: b, x_2 : L (b); P \vdash x_1 :: x_2 : \abra{L (b), Q}$; 
	\item $P \; \m{uniform} (d, n, \emptyset)$. 
\end{itemize}
\end{lemma}
\begin{proof}
$P$ is given by $\lhd (Q)$, which is defined as 
\begin{equation} \label{eq:definition of P in the proof of the lemma that the uniformity of annotations is preserved by list constructors}
	P (i, j) := 
	\begin{cases}
		Q (0_{b} :: j) + Q (j) & \text{if } i = 0_{b}; \\
		Q (i :: j) & \text{otherwise}. 
	\end{cases}
\end{equation}
Here, $i \in \mathcal{I} (\{x_1: b\})$ and $j \in \mathcal{I} (\{x_2: L (b)\})$. 

We will now prove $P \; \m{uniform} (d, n, \emptyset)$. 
Firstly, due to the definition of $P$ in \eqref{eq:definition of P in the proof of the lemma that the uniformity of annotations is preserved by list constructors}, if the maximum degree of $Q$ is $d$, so is the maximum degree of $P$. 
That is, $P (i, j) > 0$ implies $\degree{i} + \degree{j} \leq d$. 
Hence, the first and second conditions for $P \; \m{uniform} (d, n, \emptyset)$ (Definition~\ref{def:uniform resource annotations for typing contexts in multivariate AARA}) are met. 

Secondly, in $P$, any base polynomial (i) that has degree $d$ and (ii) that only involves size variables of either $x_1$ or $x_2$ (but not both) must have coefficient $n$. 
To see this, let us first consider $x_2$. 
Fix an arbitrary $j \in \mathcal{I} (\{x_2: L (b)\})$ such that $\degree{j} = d$. 
We get
\begin{align*}
	P (0_{b}, j) & = Q (0_{b} :: j) + Q (j) \\
	& = 0 + Q (j) \\
	& = Q (j). 
\end{align*}
In the second line, $0_{b} :: j$ would have degree $d+1$, which exceeds the highest degree of $Q$. 
Hence, we must have $Q (0::j) = 0$. 
As a result, we obtain $P (0_{b} :: j) = Q (j)$, where both terms have the same degree, namely $d$. 

Next, let us consider $x_1$.
Fix an arbitrary $i \in \mathcal{I} (\{x_1: b\})$ such that $\degree{i} = d$. 
This gives
\begin{align*}
	P (i, 0_{L (b)}) & = Q (i :: 0_{L (b)}) \\
	& = 0, 
\end{align*}
where the last line follows from the fact that $\degree{i :: 0_{L (b)}}$ exceeds $d$ and hence $Q (i :: 0_{L (b)}) = 0$ must hold. 
Thus, \textsc{(M:Sup)} allows us to increase $P (i, 0_{L (b)})$ such that the claim holds. 
Finally, the case where $i = 0_{b}$ is equivalent to the case where $d = 0$, and it is immediate to prove this case. 
This concludes the proof. 
\end{proof}

\begin{lemma}[Preservation of uniformity by list destructors]
\label{lemma:pattern matching on lists preserves uniformity of multivariate annotations}
Consider $e \equiv \m{case} \; x \; \{[\,] \hookrightarrow e_0 \mid (x_1::x_2) \hookrightarrow e_1 \}$, where $\Delta; \Sigma; \Gamma, x : L(b) \vdash e \; \inhpoly{V}$. 
Assume that the typing judgments of $e_0$ and $e_1$ are
\begin{equation*}
	\Gamma \vdash e_0 : b_2 \qquad \Gamma, x: b, x_2: L (b) \vdash e_1 : b_2, 
\end{equation*}
where $b$ does not contain list types inside; that is, $L (b)$ is a non-nested list type. 
Additionally, suppose we are given $P_0$ and $P_1$ such that
\begin{equation*}
	\Gamma; P_0 \vdash e_0 : \abra{b_2, Q} \qquad \Gamma, x_1: b, x_2: L (b); P_1 \vdash e_1 : \abra{b_2, Q}, 
\end{equation*}
where $Q \; \m{uniform} (d, n)$, $P_0 \; \m{uniform} (d, n, V \setminus \{x\})$, and  $P_1 \; \m{uniform} (d, n, V [x \mapsto x_1, x_2])$. 
Then there exists a multivariate annotation $P$ for the entire $e$ such that (i) the output is annotated with $Q$ and (ii) $P \; \m{uniform} (d, n, V)$ holds. 
\end{lemma}
\begin{proof}
From $P_1$, we will construct an annotation $P_2$ over the typing context $\Gamma \cup \{x: L (b) \}$. 
$P_2$ should satisfy $\lhd (P_2) = P_1$ and $P_2 \; \m{uniform} (d, n, V)$. 
$\lhd (P_2)$ is defined as
\begin{equation*}
	\lhd (P_2) (i, j, k) := 
	\begin{cases}
		P_2 (i, 0_{b} :: k)  + P_2 (i, k) & \text{if } j = 0_{b}; \\
		P_2 (i, j::k) & \text{otherwise}. 
	\end{cases}
\end{equation*}
Because $L (b)$ is a non-nested list type by assumption, $j$ is always $0_{b}$. 
If we set 
\begin{equation} \label{eq:definition of P_2 in the proof of the lemma that the uniformity of annotations is preserved by list destructors}
	P_2 (i, k) := P_1 (i, 0_{b}, k), 
\end{equation}
we obtain 
\begin{alignat*}{2}
	\lhd (P_2) (i, 0_{b}, k) & = P_2 (i, 0_{b} :: k) + P_2 (i, k) && \qquad \text{by definition} \\
	& = P_1 (i, 0_{b}, 0_{b} :: k) + P_1 (i, 0_{b}, k) && \qquad \text{by \eqref{eq:definition of P_2 in the proof of the lemma that the uniformity of annotations is preserved by list destructors}} \\
	& \geq P_1 (i, 0_{b}, k). 
\end{alignat*}
Therefore, $\lhd (P_2)$ is a subtype of $P_1$; hence, $\lhd (P_2)$ can be converted to $P_1$ by \textsc{(M:Sub)}. 

It remains to check that $P_2 \; \m{uniform} (d, n, V)$ holds. 
Let $v$ be a variable from $\dom{\Gamma} \cup \{x\} \setminus V$. 
We now conduct case analysis on $v$.
Suppose $v \not\equiv x$; that is, $x \in \dom{\Gamma}$. 
This gives
\begin{alignat*}{2}
	P_2 (i, k) > 0 & \implies P_1 (i, 0_{b}, k) > 0 && \qquad \text{by \eqref{eq:definition of P_2 in the proof of the lemma that the uniformity of annotations is preserved by list destructors}} \\
	& \implies \degree{i} < d \lor (\degree{i} = d \land \degree{k} = 0)  && \qquad \text{by the $\m{uniform}$ assumption}.  
\end{alignat*}
Therefore, the first and second conditions of Definition~\ref{def:uniform resource annotations for typing contexts in multivariate AARA} are satisfied.  
Further, if a base polynomial in $P_2$ only contains size variables of $v$ and has degree $d$, its coefficient is $n$ due to \eqref{eq:definition of P_2 in the proof of the lemma that the uniformity of annotations is preserved by list destructors} and the assumption $P_1 \; \m{uniform} (d, n, V [x \mapsto x_1, x_2])$. 
Hence, the third condition of Definition~\ref{def:uniform resource annotations for typing contexts in multivariate AARA} is true. 

Conversely, suppose $v \equiv x$. 
Consider $k \in \mathcal{I} (\{x: L (b)\})$. 
This yields
\begin{alignat*}{2}
	P_2 (i, k) > 0 & \implies P_1 (i, 0_{b}, k) > 0 && \qquad \text{by \eqref{eq:definition of P_2 in the proof of the lemma that the uniformity of annotations is preserved by list destructors}} \\
	& \implies \degree{k} < d \lor (\degree{k} = d \land \degree{i} = 0) && \qquad \text{by the $\m{uniform}$ assumption}. 
\end{alignat*}
This satisfies the first and second conditions of Definition~\ref{def:uniform resource annotations for typing contexts in multivariate AARA}. 
Furthermore, for any $k \in \mathcal{I} (\{x: L (b)\})$, if $\degree{k} = d$, we have $P_2 (0_{\Gamma}, k) = n$. 
This is because $P_1 (0_{\Gamma}, 0_{b}, k) = n$ holds due to the inductive hypothesis of $P_1$. 
Therefore, $P_2 \; \m{uniform} (d, n, V)$ holds as required.

Finally, we can merge $P_2$ and $\eta^{\Gamma \cup \{x: L (b)\}}_{0} (P_0)$ into $P$ that satisfies $P \; \m{uniform} (d, n, V)$.
This concludes the proof. 
\end{proof}

The following lemma concerns the construction of a multivariate annotation for a primitive recursion and is more general than Lemma~\ref{lemma:construction of a multivariate annotation for primitive recursion}. 

\begin{lemma}[Resource annotation for resource-polymorphic recursion]
\label{lemma:construction of a multivariate annotation for primitive recursion with resource-polymorphic recursion}
Consider a primitive recursion of the form
\begin{equation*}
	f := \lam x, \Gamma. \m{case} \; x \; \{[\,] \hookrightarrow e_0 \mid y::\mi{ys} \hookrightarrow \m{let} \; z = f \; \mi{ys} \; \Gamma \; \m{in} \; e_1 \}. 
\end{equation*}
Let $P_{i}$ and $Q_{i}$ for $i \in \{0, 1\}$ be multivariate annotations that satisfy
\begin{equation*}
	\Gamma; P_0 \vdash e_0 : \abra{b_2, Q_0} \qquad y: b, \mi{ys}: L (b), z: b_2; P_1 \vdash e_1 : \abra{b_2, Q_1}. 
\end{equation*}
For each $i \in \mathcal{I} (\{y: b, \mi{ys}: L (b) \})$, a multivariate annotation $P_{2, i}$ is defined by 
\begin{alignat*}{2}
	\Gamma, x: L (b); P_{2, 0} & \vdash f \; x \; \Gamma : \abra{b_2, \pi^{z: b_2}_{0} (P_1) - Q_1} && \qquad \text{if } i = 0; \\
	\Gamma, x: L (b); P_{2, i} & \vdash f \; x \; \Gamma : \abra{b_2, \pi^{z: b_2}_{i} (P_1)} && \qquad \text{otherwise}, 
\end{alignat*}
where $f \; x \; \Gamma$ denotes the whole primitive recursion. 
From $P_1$ and $P_{2, i}$ for $i \in \mathcal{I} (\{y: b, \mi{ys}: L (b) \})$, it is possible to build a multivariate annotation $P$ for the entire primitive recursion such that its final output is annotated with $Q_1$.

Furthermore, assume $\Delta; \Gamma \vdash e_0 \; \inhpoly{V}$.
If we have
\begin{itemize}
	\item $P_0 \; \m{uniform} (d, n, V)$, where $d > 0$, and 
	\item $P_{2, i} \; \m{uniform} (d - 1, n_{i}, V \cup \{x\})$ for each $i \in \mathcal{I} (\{y: b, \mi{ys}: L (b)\})$ and some $n_{i} \in \mathbb{Q}_{> 0}$
\end{itemize}
then $P \; \m{uniform} (d, n, V \cup \{x\})$ is true as well. 
\end{lemma}
\begin{proof}
It is given that each recursive call needs $P_1$ much potential. 
This information does not directly tell us the total amount of potential.
This is because $P_1$ involves $z$, which is the result of the recursive call, and it is not immediately obvious how to relate $z$ back to $x$, which is the original input. 
Hence, the first task is to work out a multivariate annotation for each recursive call in terms of $\{y, \mi{ys}\} \cup \Gamma$ (instead of $\{y, \mi{ys}, z\}$). 

As stated in the theorem, let $P$ be the annotation for the entire primitive recursion that we aim to construct in this proof. 
In the second branch of the primitive recursion, as a result of pattern matching in $x$, $P$ becomes $\lhd (P)$ whose domain is $\Gamma \cup \{y: b, \mi{ys}: L (b)\}$. 

Base polynomials in $P_1$ can be classified into three categories: 
\begin{itemize}
	\item Base polynomials that only concern $\{y, \mi{ys}\}$. 
	The coefficients of these polynomials are given by $\pi^{\{y: b, \mi{ys}: L (b) \}}_{0} (P_1)$. 
	\item Base polynomials that only concern $z$.
	The coefficients of these base polynomials are given by $\pi^{z: b_2}_{0} (P_1)$. 
	Out of this potential, $Q_1$ will be paid by $P$ because $P$ is exactly what we are trying to establish at the moment and we are allowed to reuse $P$ as an input annotation for $Q_1$. 
	The remaining potential, $\pi^{z: b_2}_{0} (P_1) - Q_1$, will be paid by $P_{2, 0}$ in the assumption. 
	\item Base polynomials that concern both $\{y, \mi{ys}\}$ and $z$.
	We can replace $z$ with $\{y, \mi{ys}\}$ in the representation of these base polynomials by using $P_{2, i}$, where $i \neq 0 \in \mathcal{I} (\{y: b, \mi{ys}: L (b)\})$. 
\end{itemize}

Let $\zeta$ be $\Gamma \cup \{y: b, \mi{ys}: L (b)\}$. 
The total potential needed by each recursive call is then given by
\begin{equation*}
	\eta^{\zeta}_{0_{\Gamma}} (\pi^{\{y: b, \mi{ys}: L (b) \}}_{0} (P_1)) + \eta^{\zeta}_{0_{b}} (P) + \sum_{i \in \mathcal{I} (\{y: b, \mi{ys}: L (b) \})} P_{2, i} \cdot i,
\end{equation*}
where $P_{2, i} \cdot i$ denotes the product of $P_{2, i}$, whose typing context is $\Gamma \cup \{x: L (b)\}$, and base polynomial $i$, whose typing context is $\{y: b, \mi{ys}: L (b)\}$.
Notice that the typing context of $P_{2, i}$ and that of $i$ overlap---they have $\{\mi{ys}: L (b)\}$ in common. 

Because this should be equal to $\lhd (P)$, we obtain
\begin{equation} \label{eq:condition on P in the construction of a multivariate annotation for a primitive recursion with resource-polymorphic recursion}
	\lhd (P) - \eta^{\zeta}_{0_{b}} (P) = \eta^{\zeta}_{0_{\Gamma}} (\pi^{\{y: b, \mi{ys}: L (b) \}}_{0} (P_1)) + \sum_{i \in \mathcal{I} (\{y: b, \mi{ys}: L (b) \})} P_{2, i} \cdot i. 
\end{equation}
From the right hand side, we can construct a desirable $P$ such that \eqref{eq:condition on P in the construction of a multivariate annotation for a primitive recursion with resource-polymorphic recursion} holds, as done in the proof of Lemma~\ref{lemma:construction of a multivariate annotation for primitive recursion}. 
Specifically, $\lhd (P)$ is defined as
\begin{equation*}
	\lhd (P) (i, j, k) := 
	\begin{cases}
		P (i, 0_{b} :: k) + P (i, k) & \text{if } j = 0_{b}; \\
		P (i, j::k) & \text{otherwise}. 
	\end{cases}
\end{equation*}
Hence, $\lhd (P) - \eta^{\zeta}_{0} (P)$ is given as follows. 
If $j = 0_{b}$, then 
\begin{align*}
	(\lhd (P) - \eta^{\zeta}_{0} (P)) (i, j, k) & = \lhd (P) (i, 0_{b}, k) - \eta^{\zeta}_{0} (P) (i, 0_{b}, k) \\
	& = P (i, 0_{b} :: k) + P (i, k) - P (i, k) \\
	& = P (i, 0_{b} :: k). 
\end{align*}
Conversely, if $j \neq 0_{b}$, then
\begin{align*}
	(\lhd (P) - \eta^{\zeta}_{0} (P)) (i, j, k) & = \lhd (P) (i, j, k) - \eta^{\zeta}_{0} (P) (i, j, k) \\
	& = P (i, j :: k) + P (i, k) - 0 \\
	& = P (i, j::k). 
\end{align*}
In conclusion, we obtain
\begin{equation*}
	(\lhd (P) - \eta^{\zeta}_{0} (P)) (i, j, k) = P (i, j::k)
\end{equation*}
regardless of whether $j = 0_{b}$ or not. 
This must be equal to the right hand side of \eqref{eq:condition on P in the construction of a multivariate annotation for a primitive recursion with resource-polymorphic recursion}, and it is clear that such $P$ is guaranteed to exist. 
In addition, it must be possible to extract $Q_0$ from $P$ when $x$ is the empty list. 
This yields $P$'s definition: 
\begin{equation} \label{eq:definition of P in the construction of a multivariate annotation for a primitive recursion with resource-polymorphic recursion}
	\begin{split}
	P (i, j::k) & := \eta^{\zeta}_{0_{\Gamma}} (\pi^{\{y: b, \mi{ys}: L (b) \}}_{0} (P_1)) (i, j, k) + \sum_{r \in \mathcal{I} (\{y: b, \mi{ys}: L (b) \})} (P_{2, r} \cdot r) (i, j, k) \\
	P (i, 0_{b}) & := Q_0 (i). 
	\end{split}
\end{equation}

It remains to check that $P \; \m{uniform}(d, n, V)$ holds. 
Consider $v \in \dom{\Gamma} \setminus V$. 
We will first prove the first and second conditions of Definition~\ref{def:uniform resource annotations for typing contexts in multivariate AARA} for $P$. 
If $\degree{i} \geq d > 0$ (which is a stronger condition than $\degree{i} > d$), it gives
\begin{align*}
	P (i, j::k) & = \eta^{\zeta}_{0_{\Gamma}} (\pi^{\{y: b, \mi{ys}: L (b) \}}_{0} (P_1)) (i, j, k) + \sum_{r \in \mathcal{I} (\{y: b, \mi{ys}: L (b) \})} (P_{2, r} \cdot r) (i, j, k) \\
	& = \sum_{r \in \mathcal{I} (\{y: b, \mi{ys}: L (b) \})} (P_{2, r} \cdot r) (i, j, k) \\
	& = 0. 
\end{align*}
Here, the second line follows from the definition of the extension operator $\eta^{\zeta}_{0_{\Gamma}}$. 
Because we extend $\pi^{\{y: b, \mi{ys}: L (b) \}}_{0} (P_1)$ with $0_{\Gamma}$, we have $\eta^{\zeta}_{0_{\Gamma}} (\pi^{\{y: b, \mi{ys}: L (b) \}}_{0} (P_1)) (i, j, k) = 0$ whenever $i \neq 0_{\Gamma}$ (which is the case since $\degree{i} \geq d > 0$).  
With regard to the third line above, due to the assumption $P_{2, r} \; \m{uniform} (d-1, n_{r}, V \cup \{x\})$, $(P_{2, r} \cdot r) (i, j, k) > 0$ only if $\degree{i} \leq d - 1$. 
Therefore, if $\degree{i} = d$, we have $P (i, j::k) = 0$ for any $j::K$, thereby establishing the second condition of Definition~\ref{def:uniform resource annotations for typing contexts in multivariate AARA}. 

In the case of $P (i, 0_{b})$, if $\degree{i} > d$, we have
\begin{align*}
	P (i, 0_{L (b)}) & = Q_0 (i) \\
	& = 0,
\end{align*}
where the second line follows from the assumption $Q_0 \; \m{uniform} (d, n, V)$. 
Therefore, the first and second conditions of Definition~\ref{def:uniform resource annotations for typing contexts in multivariate AARA} are true for $P$. 

Finally, to prove the third condition of Definition~\ref{def:uniform resource annotations for typing contexts in multivariate AARA}, consider $i \in \mathcal{I} (\Gamma)$, where $\degree{i} = d$. 
This gives
\begin{align*}
	P (i, 0_{L (b)}) & = Q_0 (i) \\
	& = n
\end{align*}
where the second line follows from the assumption $Q_0 \; \m{uniform} (d, n, V)$. 

In summary, all the three conditions of Definition~\ref{def:uniform resource annotations for typing contexts in multivariate AARA} hold for $P$.
Therefore, $P \; \m{uniform} (d, n, V \cup \{x\})$ is indeed true.
This concludes the proof.  
\end{proof}

\else
\fi

\end{document}